\documentclass{article}
\usepackage{amsmath,amssymb,enumerate,bbm,theorem,ifthen}

\newcommand{\AMSclass}[1]{{\textbf{A.M.S. subject classification:} #1}}
\newcommand{\assign}{:=}
\newcommand{\keywords}[1]{{\textbf{Keywords:} #1}}
\newcommand{\nonesep}{}
\newcommand{\tmdummy}{$\mbox{}$}
\newcommand{\tmmathbf}[1]{\ensuremath{\boldsymbol{#1}}}
\newcommand{\tmop}[1]{\ensuremath{\operatorname{#1}}}
\newcommand{\tmscript}[1]{\text{\scriptsize{$#1$}}}
\newcommand{\tmtextbf}[1]{{\bfseries{#1}}}
\newcommand{\tmtextit}[1]{{\itshape{#1}}}
\newcommand{\tmtextmd}[1]{{\mdseries{#1}}}
\newcommand{\tmtextup}[1]{{\upshape{#1}}}
\newenvironment{enumeratenumeric}{\begin{enumerate}[1.] }{\end{enumerate}}
\newenvironment{itemizedot}{\begin{itemize} }{\end{itemize}}
\newenvironment{itemizeminus}{\begin{itemize} }{\end{itemize}}
\newenvironment{proof}{\noindent\textbf{Proof\ }}{\hspace*{\fill}$\Box$\medskip}
\numberwithin{equation}{section}  
\numberwithin{figure}{section}  
\newtheorem{theorem}{Theorem}[section]
\newtheorem{lemma}[theorem]{Lemma}
\newtheorem{proposition}[theorem]{Proposition}
\newtheorem{corollary}[theorem]{Corollary}
\newtheorem{definition}[theorem]{Definition}
\newtheorem{remark}[theorem]{Remark}

\newcommand{\lp}[1]{\ifthenelse{\equal{#1}{0}}{(}{}\ifthenelse{\equal{#1}{1}}{\bigl(}{}\ifthenelse{\equal{#1}{2}}{\Bigl(}{}\ifthenelse{\equal{#1}{3}}{\biggl(}{}\ifthenelse{\equal{#1}{4}}{\Biggl(}{}\ifthenelse{\equal{#1}{5}}{\Biggl(}{}}
\newcommand{\rp}[1]{\ifthenelse{\equal{#1}{0}}{)}{}\ifthenelse{\equal{#1}{1}}{\bigr)}{}\ifthenelse{\equal{#1}{2}}{\Bigr)}{}\ifthenelse{\equal{#1}{3}}{\biggr)}{}\ifthenelse{\equal{#1}{4}}{\Biggr)}{}\ifthenelse{\equal{#1}{5}}{\Biggr)}{}}
\newcommand{\lbc}[1]{\ifthenelse{\equal{#1}{0}}{\{}{}\ifthenelse{\equal{#1}{1}}{\bigl\{}{}\ifthenelse{\equal{#1}{2}}{\Bigl\{}{}\ifthenelse{\equal{#1}{3}}{\biggl\{}{}\ifthenelse{\equal{#1}{4}}{\Biggl\{}{}\ifthenelse{\equal{#1}{5}}{\Biggl\{}{}}
\newcommand{\rbc}[1]{\ifthenelse{\equal{#1}{0}}{\}}{}\ifthenelse{\equal{#1}{1}}{\bigr\}}{}\ifthenelse{\equal{#1}{2}}{\Bigr\}}{}\ifthenelse{\equal{#1}{3}}{\biggr\}}{}\ifthenelse{\equal{#1}{4}}{\Biggr\}}{}\ifthenelse{\equal{#1}{5}}{\Biggr\}}{}}
\newcommand{\lba}[1]{\ifthenelse{\equal{#1}{0}}{\langle}{}\ifthenelse{\equal{#1}{1}}{\bigl\langle}{}\ifthenelse{\equal{#1}{2}}{\Bigl\langle}{}\ifthenelse{\equal{#1}{3}}{\biggl\langle}{}\ifthenelse{\equal{#1}{4}}{\Biggl\langle}{}\ifthenelse{\equal{#1}{5}}{\Biggl\langle}{}}
\newcommand{\rba}[1]{\ifthenelse{\equal{#1}{0}}{\rangle}{}\ifthenelse{\equal{#1}{1}}{\bigr\rangle}{}\ifthenelse{\equal{#1}{2}}{\Bigr\rangle}{}\ifthenelse{\equal{#1}{3}}{\biggr\rangle}{}\ifthenelse{\equal{#1}{4}}{\Biggr\rangle}{}\ifthenelse{\equal{#1}{5}}{\Biggr\rangle}{}}
\newcommand{\ve}[1]{\ifthenelse{\equal{#1}{0}}{|}{}\ifthenelse{\equal{#1}{1}}{\big|}{}\ifthenelse{\equal{#1}{2}}{\Big|}{}\ifthenelse{\equal{#1}{3}}{\bigg|}{}\ifthenelse{\equal{#1}{4}}{\Bigg|}{}\ifthenelse{\equal{#1}{5}}{\Bigg|}{}}
\newcommand{\lb}[1]{\ifthenelse{\equal{#1}{0}}{[}{}\ifthenelse{\equal{#1}{1}}{\bigl[}{}\ifthenelse{\equal{#1}{2}}{\Bigl[}{}\ifthenelse{\equal{#1}{3}}{\biggl[}{}\ifthenelse{\equal{#1}{4}}{\Biggl[}{}\ifthenelse{\equal{#1}{5}}{\Biggl[}{}}
\newcommand{\rb}[1]{\ifthenelse{\equal{#1}{0}}{]}{}\ifthenelse{\equal{#1}{1}}{\bigr]}{}\ifthenelse{\equal{#1}{2}}{\Bigr]}{}\ifthenelse{\equal{#1}{3}}{\biggr]}{}\ifthenelse{\equal{#1}{4}}{\Biggr]}{}\ifthenelse{\equal{#1}{5}}{\Biggr]}{}}
\newcommand{\srp}[3]{\ifthenelse{\equal{#1}{0}}{)^{#2}_{#3}}{}\ifthenelse{\equal{#1}{1}}{\bigr)^{#2}_{#3}}{}\ifthenelse{\equal{#1}{2}}{\Bigr)^{#2}_{#3}}{}
\ifthenelse{\equal{#1}{3}}{\biggr)^{#2}_{#3}}{}\ifthenelse{\equal{#1}{4}}{\Biggr)^{#2}_{#3}}{}\ifthenelse{\equal{#1}{5}}{\Biggr)^{#2}_{#3}}{}}
\newcommand{\srb}[3]{\ifthenelse{\equal{#1}{0}}{]^{#2}_{#3}}{}\ifthenelse{\equal{#1}{1}}{\bigr]^{#2}_{#3}}{}\ifthenelse{\equal{#1}{2}}{\Bigr]^{#2}_{#3}}{}
\ifthenelse{\equal{#1}{3}}{\biggr]^{#2}_{#3}}{}\ifthenelse{\equal{#1}{4}}{\Biggr]^{#2}_{#3}}{}\ifthenelse{\equal{#1}{5}}{\Biggr]^{#2}_{#3}}{}}
\newcommand{\srbc}[3]{\ifthenelse{\equal{#1}{0}}{{\}}^{#2}_{#3}}{}\ifthenelse{\equal{#1}{1}}{\bigr{\}}^{#2}_{#3}}{}\ifthenelse{\equal{#1}{2}}{\Bigr{\}}^{#2}_{#3}}{}
\ifthenelse{\equal{#1}{3}}{\biggr{\}}^{#2}_{#3}}{}\ifthenelse{\equal{#1}{4}}{\Biggr{\}}^{#2}_{#3}}{}\ifthenelse{\equal{#1}{5}}{\Biggr{\}}^{#2}_{#3}}{}}
\newcommand{\srba}[3]{\ifthenelse{\equal{#1}{0}}{\rangle^{#2}_{#3}}{}\ifthenelse{\equal{#1}{1}}{\bigr\rangle^{#2}_{#3}}{}\ifthenelse{\equal{#1}{2}}{\Bigr\rangle^{#2}_{#3}}{}
\ifthenelse{\equal{#1}{3}}{\biggr\rangle^{#2}_{#3}}{}\ifthenelse{\equal{#1}{4}}{\Biggr\rangle^{#2}_{#3}}{}\ifthenelse{\equal{#1}{5}}{\Biggr\rangle^{#2}_{#3}}{}}
\newcommand{\sve}[3]{\ifthenelse{\equal{#1}{0}}{|^{#2}_{#3}}{}\ifthenelse{\equal{#1}{1}}{\bigr|^{#2}_{#3}}{}\ifthenelse{\equal{#1}{2}}{\Bigr|^{#2}_{#3}}{}
\ifthenelse{\equal{#1}{3}}{\biggr|^{#2}_{#3}}{}\ifthenelse{\equal{#1}{4}}{\Biggr|^{#2}_{#3}}{}\ifthenelse{\equal{#1}{5}}{\Biggr|^{#2}_{#3}}{}}
\newcommand{\im}{\mathcal{I}m}
\newcommand{\re}{\mathcal{R}e}
\newcommand{\sh}{\ensuremath{\tmop{sh}}}
\newcommand{\ch}{\ensuremath{\tmop{ch}}}
\newcommand{\tgh}{\ensuremath{\tmop{th}}}

\newcommand{\paraboleb}[1]{\ensuremath{S_{#1}}}
\newcommand{\disqueb}[1]{\ensuremath{D_{#1}}}

\newcommand{\demiplanb}{\ensuremath{\mathbbm{C}^+}}

\newcommand{\parabole}[1]{{\paraboleb{#1}}}
\newcommand{\disque}[1]{{\disqueb{#1}}}

\newcommand{\demiplan}[1]{{\demiplanb}#1}

\newcommand{\ifr}[3]{I_{{#1},{#2}}^{#3}}
\newcommand{\texspace}[2]{ 
\ifthenelse{\equal{#1}{small}}{\par}{}
\ifthenelse{\equal{#1}{med}}{\par\text{}\par}{}
\ifthenelse{\equal{#1}{big}}{\par\text{}\par\text{}
\par}{}{\noindent #2}}

\begin{document}

\title{ Borel summation of the semi-classical expansion of the partition
function associated to the Schr\"odinger equation\thanks{This paper has been written using the GNU TEXMACS scientific text editor.}
\thanks{\keywords{heat
kernel, quantum mechanics, semi-classical, asymptotic expansion,
Wigner-Kirkwood expansion, Borel summation, tree graph equality, mould
calculus, hypergeometric}; \AMSclass{35K08, 30E15, 35C20,
81Q30}}}\author{Thierry Harg\'e}\date{May 15, 2013}\maketitle

\begin{abstract}
  Let $H$ be a perturbation of the semi-classical harmonic oscillator on
  $\mathbbm{R}^{\nu}$. We prove that the partition function associated to the
  Hamiltonian $H$ is the Borel sum of its $h$-expansion.
\end{abstract}

\section{Introduction}

{\texspace{small}{\tmtextbf{-1-}}} Let $\nu \geqslant 1$. Let $\omega_1,
\ldots, \omega_{\nu} > 0$ and $h > 0$. Let
\begin{equation}
  \label{introductionfred0.5} \begin{array}{lll}
    H & \assign & - h^2 \partial_x^2 + \frac{1}{4} ( \tmmathbf{\omega} x)^2 -
    c (x)\\
    &  & \\
    & \assign & - h^2 (\partial_{x_1}^2 + \cdots + \partial_{x_{\nu}}^2) +
    \frac{1}{4} (\omega_1^2 x_1^2 + \cdots + \omega_{\nu}^2 x_{\nu}^2) - c
    (x_1, \ldots, x_{\nu}) .
  \end{array}
\end{equation}
Under suitable assumptions on the potential $c$, the operator $H$ is
self-adjoint on $L^2 (\mathbbm{R}^{\nu})$ with discrete spectrum and the
following partition function
\[ \Theta_H (t, h) \assign \tmop{Tr} \lp{2} e^{- \frac{t}{h} H} \rp{2} =
   \sum_{n = 1}^{+ \infty} e^{- t \lambda_n (h) / h} \]
is well defined on $] 0, + \infty [^2$. Here
\[ \lambda_1 (h) \leqslant \lambda_2 (h) \leqslant \cdots \leqslant \lambda_n
   (h) \leqslant \cdots \]
denote the eigenvalues of the operator $H$. Let $\Theta^{\tmop{conj}}_H (t,
h)$ be defined by the identity
\[ \Theta_H (t, h) = e^{t c (0) / h} \times \Theta^{\tmop{conj}}_H (t, h) . \]
Then, under suitable assumptions on the potential $c$, for $t \in] 0, + \infty
[$, the following expansion holds
\begin{equation}
  \label{introductionfred1} \Theta^{\tmop{conj}}_H (t, h) = a_0 (t) + a_1 (t)
  h + \cdots + a_{r - 1} (t) h^{r - 1} + h^r \mathcal{O}_{h \rightarrow 0^+}
  (1) .
\end{equation}
In this paper we give conditions on the potential $c$ so that this expansion
is Borel summable (with respect to $h$) and its Borel sum is equal to
$\Theta^{\tmop{conj}}_H (t, h)$ (see Theorem \ref{theorembetty2} for a more
precise statement). \ A Borel summation statement concerning the heat kernel
is also proved (see Proposition \ref{propositionbetty1}).

We also give and use a Borel summation statement for multidimensional Gaussian
integrals (Proposition \ref{nevflora0.1}). This statement can be viewed as a
consequence of Proposition \ref{nevalclara1} which deal with so-called
hypergeometric vection transforms. These transforms map $\hat{\mathcal{N}}_{R,
K}$ (see Definition \ref{nevclara2}), the space of summable symbols, into
itself. Hypergeometric vection transforms are related to vection transforms
which satisfy a similar property. Vection transforms play a role in celeration
theory [E4]. See Section \ref{nevflora} and Appendix B for more details.

{\texspace{med}{\tmtextbf{-2-}}} Quantum mechanics gives many examples of
divergent expansions [Si]. Let us focus on semi-classical expansions related
to the Schr\"odinger equation. In this case, an important motivation is to
describe quantum quantities with the help of classical quantities [B-B]. For
instance the coefficients
\[ c (0), a_0 (t), \ldots, a_r (t), \ldots \]
are classical quantities whereas $\Theta_H (t, h)$ is a quantum quantity (see
also introduction in [Ha4]). How to recover a quantum quantity with the help
of the coefficients of its $h$-expansion? Notice that asymptotic points of
view do not provide an answer to this question. A Borel summation viewpoint is
used by Voros [V1] and Delabaere, Dillinger, Pham [D-D-P] for the study of the
one dimensional Schr\"odinger operator $- h^2 \frac{d^2}{d x^2} + V (x)$,
where the potential $V$ is polynomial. In particular, a study of tunnelling is
proposed in [D-D-P]. This involves a non-elementary Borel summation process.

A semi-classical interpretation can be proposed for the small time expansion
of the heat kernel or the partition function (see for instance [Ha4]).
Actually there are a lot of similarities between Borel summability of
$h$-expansions and small time (high temperature) expansions [Ha4].

We also find technical similarities between the study of the groundstate
energy of the massless spin-boson model [A] and our work. In this paper, the
description of the groundstate energy is viewed with the help of the heat
kernel, the interaction between the bosonic quantum field and the spin model
is viewed as a perturbation (without renormalization), and above all, the tree
graph equality plays an important role (we give more details about this
equality and its use in the sequel).

The operator $H$ describes the behaviour of a quantum particle moving in a
classical field. The interaction between the particle and the field is
understood via the deformation formula. This formula, which can be derived
from the Dyson expansion [On], is a multiple scattering expansion (see [Fe,
Figure 1]). The convergence of the $h$-expansion of the heat kernel is studied
with the use of this formula and with the help of the tree graph equality, an
identity used in statistical mechanics and quantum field theory [Br], dealing
with an infinite number of particles. The Hamiltonian $H$ does not involve,
namely, second quantization but our method, through the deformation formula,
involves virtual particles which copy the studied particle at different times
(see [Ha4, Appendix, Formula 5.1]).

We find similarities between the use of the tree graph equality and the
arborification-coarborification process, an important tool, for instance, when
linearization of vector fields or diffeomorphisms are considered [E2, E-V]. In
both cases, in a general setting, by an algebraic way, one can rearrange terms
in order to restore convergence.

{\texspace{med}{\tmtextbf{-3-}}} Let us comment the assumptions on the
function $c$. The function $c$ is viewed as the Fourier transform of some
Borel measure $\mu$:
\[ c (x) = \int_{\mathbbm{R}^{\nu}} \exp (ix \cdot \xi) d \mu (\xi) . \]
\begin{itemize}
  \item Proving a Borel summation statement for the heat kernel
  \[ \lba{1} x|e^{- \frac{t}{h} H} |y \rba{1} \]
  (see Proposition \ref{propositionbetty1}) requires the assumption
  (\ref{propositionbetty1.2}). This assumption holds for instance if $\nu =
  1$, $c (x) = k \cos x$ where $k$ is an arbitrary real number and $T$ is
  small enough. Here, $\mu \assign \frac{k}{2} (\delta_{\xi = 1} + \delta_{\xi
  = - 1})$. This allows one to build a potential
  \[ V \assign \frac{1}{4} ( \tmmathbf{\omega} x)^2 - c (x) \]
  with an arbitrary number of wells. Of course, the Borel summation statement
  for the corresponding heat kernel holds for small values of $t$.
  
  Notice that giving a Borel summation statement for the the Green function or
  resolvent (see [B-B]) with the same assumptions on the potential seems out
  of reach with our method. Since, under suitable assumptions,
  \begin{eqnarray*}
    \lba{1} x| \lp{1} H + \lambda \srp{1}{- 1}{} |y \rba{1} = &  & \int_0^{+
    \infty} e^{- \lambda t} \lba{1} x|e^{- t H} |y \rba{1} d t\\
    = &  & \int_0^{+ \infty} e^{- \lambda \frac{t}{h}} \lba{1} x|e^{-
    \frac{t}{h} H} |y \rba{1} \frac{d t}{h},
  \end{eqnarray*}
  the knowledge of the Green function through the heat kernel involves large
  values of $t$; this is consistent with the above remarks.
  
  \item Proving a Borel summation statement for the partition function
  $\Theta^{\tmop{conj}}_H$ (see Theorem \ref{theorembetty2}) requires also the
  assumption (\ref{propositionbetty1.2}) and the supplementary assumptions
  (\ref{theorembetty2.3}), (\ref{theorembetty3}), (\ref{theorembetty4}). These
  assumptions hold for instance if $\nu = 1$,
  \begin{equation}
    \label{commentamanda2} c (x) = k e^{- x^2 / 2}, \varepsilon = 1 / 4,
  \end{equation}
  $k \in \mathbbm{R}$ and $|k|$ is small enough. Here $d \mu (\xi) = k (2
  \pi)^{- 1 / 2} e^{- \xi^2 / 2} d \xi$. Notice that the potential $V$ given
  by the above choice for $c$ ($k \neq 0$) has no real critical points except
  $0$ if $|k|$ is small enough: we do not consider tunnelling.
  
  We do not give a Borel summation statement concerning the eigenvalues
  associated to the operator $H$: all complex critical times of the potential
  $V$ certainly must be taken into account for such a result. Let us comment
  this. Under suitable assumptions on $V$, the first eigenvalue (ground state
  energy) satisfies
  \[ \lambda_1 (h) = \lim_{t \longrightarrow + \infty} - \frac{1}{t} \log
     \Theta_H (t, h), \]
  thus involving the partition function for large values of $t$. Let us
  consider again the example given by (\ref{commentamanda2}). Then the
  potential $V$ has at least one complex critical point different from $0$.
  But the larger the parameter $t$ is, the smaller the parameter $k$ must be
  chosen such that the non-zero critical points are far from the real space
  and do not interfere with the Borel summation process. This qualitatively
  explains the utility of (\ref{theorembetty2.3}).
\end{itemize}
{\texspace{med}{\tmtextbf{-4-}}} Let us outline the proof of our result. We
use the deformation formula (see [Ha4, Ha6]) which gives a representation of
the conjugate heat kernel $p^{\tmop{conj}}$ defined by
\[ \lba{1} x|e^{- \frac{t}{h} H} |y \rba{1} = p^{\tmop{harm}} (t, x, y, h)
   \times p^{\tmop{conj}} (t, x, y, h) ; \]
here $p^{\tmop{harm}}$ denotes the heat kernel of the operator $- h^2
\partial_x^2 + \frac{1}{4} ( \tmmathbf{\omega} x)^2$. By the deformation
formula, $p^{\tmop{conj}}$ is viewed as a divergent expansion with respect to
$\frac{1}{h}$ and $h$. We then consider the formal logarithm of this
expansion. This yields
\begin{equation}
  \label{introductionanna2} p (t, x, y, h) = h^{- \nu / 2} \prod_{\upsilon =
  1}^{\nu} \lp{2} \frac{\omega_{\upsilon}}{4 \pi \tmop{sh} (\omega_{\upsilon}
  t)} \srp{2}{1 / 2}{} \times e^{- \phi (t, x, y) / h + w (t, x, y, h)}
\end{equation}
where the function $\phi$ does not depend of $h$ and the function $w$ is a
divergent expansion with respect to $h$. These functions are defined through
iterated integrals, inherited from the deformation formula, and a sum indexed
by trees, inherited from the tree graph equality. We prove that the
$h$-expansion of the function $e^w$ is Borel summable. The partition function
satisfies
\[ \Theta_H (t, h) = \int_{\mathbbm{R}^{\nu}} p (t, x, x, h) d x. \]
and by (\ref{introductionanna2})
\[ p (t, x, x, h) = h^{- \nu / 2} \prod_{\upsilon = 1}^{\nu} \lp{2}
   \frac{\omega_{\upsilon}}{4 \pi \tmop{sh} (\omega_{\upsilon} t)} \srp{2}{1 /
   2}{} \times e^{- \Phi_t (x) / h + \tmmathbf{w}_t (x, h)} \]
where $\Phi_t (x) \assign \phi (t, x, x)$ and $\tmmathbf{w}_t (x, h) \assign w
(t, x, x, h)$. We then establish the following results.
\begin{enumerate}
  \item The function $\Phi_t$ is analytic on a particular neighbourhood of
  $\mathbbm{R}^{\nu}$ in $\mathbbm{C}^{\nu}$.
  
  \item A Morse lemma holds for $\Phi_t$ on this neighborhood.
  
  \item The $h$-expansion of the function $e^{w_t}$ is Borel summable,
  uniformly in $(x, y)$.
\end{enumerate}
Finally a Borel summation statement for multidimensional Gaussian integrals
(Proposition \ref{nevflora0.1}) is proved. Then the proof of the Borel
summability of the expansion with respect to $h$ of the partition function can
be achieved.

Let us add the following remarks.

The function $\phi_t$ is a solution of a first order non-linear partial
differential equation, the Hamilton Jacobi (eikonal) equation. The explicit
form of the solution allows to avoid the method of characteristics, a standard
way to solve this equation. In a heuristic setting and if $\omega = 0$, a
similar formulation of the solution (without taking into account the
tree-graph equality, which restores convergence) can be found in [Fu-Os-Wi];
however, we use the formalism developped in [Ha3].

The Borel summability of the $h$-expansion of the function $e^w$ comes from
the Borel summability of the expansion of $w$. For this purpose, we proceed as
in [Ha4]: the deformation formula gives an explicit Borel transform of the
function $w$ based on Bessel functions. Here the tree graph equality plays a
central role (see the assertion \ref{treemaria11.3} in Remark
\ref{treemaria11.2}).

A global version of the Morse lemma allows to express the function $\Phi_t$ in
a suitable way for the Borel summation statement dealing with Gaussian
integrals.

We express the function $c$ as the Fourier transform of a measure $\mu$: it
is a convenient way to use the deformation formula [It, Ha4]. Our statement
involves two norms on $\mu$, denoted by $M_{\mu, \varepsilon}$ and $M'_{\mu}$
(see Definition \ref{theorembetty1}). We make the following assumptions.
\begin{itemize}
  \item The norm $M_{\mu, \varepsilon}$ is small. This implies the Borel
  summability of the function $e^w$. This assumption seems natural when a
  perturbation viewpoint and Borel summation are considered: it is also used
  in the proof of Borel summability of the small time expansion of the heat
  kernel (see [Ha4]).
  
  \item The function $c$ is $\mathbbm{R}$-valued on $\mathbbm{R}^{\nu}$, its
  first order derivatives vanish at the origin and the norm $M'_{\mu}$ is
  small (Theorem \ref{theorembetty2}). This implies the analyticity of
  $\Phi_t$ and the Morse lemma on a neighbourhood of $\mathbbm{R}^{\nu}$ in
  $\mathbbm{C}^{\nu}$. For this last step, we use the following: $\Phi_t$ is a
  global perturbation of the function
  \[ x \longmapsto \frac{1}{2} \sum_{\upsilon = 1}^{\nu}
     \text{$\frac{\omega_{\upsilon}}{\tmop{sh} (\omega_{\upsilon} t)} \lp{1}
     \tmop{ch} (\omega_{\upsilon} t) - 1 \rp{1} x_{\upsilon}^2$}, \]
  which is the trace on the diagonal of $\mathbbm{R}^{2 \nu}$ of the phase (up
  to a factor $1 / h$) of $p^{\tmop{harm}}$ (see (\ref{formulaclara5})). This
  explains why we assume that the measure $\mu$ admits a differentiable
  density with respect to Lebesgue measure (see the definition of $M'_{\mu}$).
\end{itemize}
Here the problem is non-linear and therefore assuming that the data are small
is not surprising. As we saw above, the assumptions on the perturbation $c$
are highly restrictive. However this perturbation kills the symmetries due to
the choice of the harmonic oscillator.

\section{Notation and main results}

For $z = |z|e^{i \theta} \in \mathbbm{C}$, $\theta \in] - \pi, \pi]$, we
denote $z^{1 / 2} \assign |z|^{1 / 2} e^{i \theta / 2}$. For $z \in
\mathbbm{C}$, we denote $\sh z \assign \frac{1}{2} (e^z - e^{- z})$, $\ch z
\assign \frac{1}{2} (e^z + e^{- z})$, $\tmop{th} z = \sh z / \ch z$ (if $\ch z
\neq 0$). For $z, z' \in \mathbbm{C}^{\nu}$, let $z \cdot z' \assign z_1 z'_1
+ \cdots + z_{\nu} z'_{\nu}$, $z^2 \assign z \cdot z$, $\bar{z} \assign (
\bar{z}_1, \ldots, \bar{z}_{\nu})$, $\mathcal{I} mz \assign ( \mathcal{I}
mz_1, \ldots, \mathcal{I} mz_{\nu})$, $|z| \assign (z \cdot \bar{z})^{1 / 2}$
(if $z \in \mathbbm{R}^{\nu}$, $|z| = \sqrt{z^2}$). We extend the first two
notations to operators such as $\partial_x = (\partial_{x_1}, \ldots,
\partial_{x_{\nu}})$. Let $a \in \mathbbm{C}^{\nu}$. We denote by
$\tmmathbf{a}$ (bold character) the linear transformation defined on
$\mathbbm{C}^{\nu}$ by
\[ z \longmapsto \tmmathbf{a} z = (a_1 z_1, \ldots, a_{\nu} z_{\nu}) \]
where $a_1, \ldots, a_{\nu}$ (respectively $z_1, \ldots, z_{\nu}$) denote the
coordinates of the vector $a$ (respectively $z$) in the standard basis of
$\mathbbm{C}^{\nu}$. For $R > 0$, let
\[ \disque{R} \assign \{z \in \mathbbm{C^{\nu}} | |z| < R\}. \]
If $U$ and $V$ are two subsets of $\mathbbm{C}^{\nu}$, let $U + V \assign \{u
+ v|u \in U, v \in V\}$. We denote $D_{U, R} \assign U + D_R$. For instance
\[ D_{\mathbbm{R}^{\nu}, R} \assign \lbc{1} a + i b|a \in \mathbbm{R}^{\nu}, b
   \in \mathbbm{R}^{\nu}, |b| < R \rbc{1} . \]
We also denote by $(e_1, \ldots, e_{\nu})$ the standard basis of
$\mathbbm{C}^{\nu}$ and by $(e^{\ast}_1, \ldots, e^{\ast}_{\nu})$ its dual
basis. For $A = (a_{j, k})_{1 \leqslant j, k \leqslant \nu}$ a $\nu \times
\nu$ matrix with complex entries, we denote $|A| \assign \sup_{|z| = 1} |A z|$
and $|A|_{\infty} \assign \max_{1 \leqslant j, k \leqslant \nu} |a_{j, k} |$.
Then $|A| \leqslant \nu |A|_{\infty}$. We set $\text{}^{\mathfrak{t}} A$ for
the transpose of the matrix $A$ and $\mathbbm{1}$ denotes the identity matrix.
Let $\Omega$ be an open domain in $\mathbbm{R}^m \times \mathbbm{C}^{m'}$. Let
$F$ be a finite dimensional space. We denote by $\mathcal{A}_F (\Omega)$
(respectively $\mathcal{A} (\Omega)$) the space of $F$-valued (respectively
$\mathbbm{C}$-valued) analytic functions on $\Omega$. If $\mathcal{I}$ is an
interval of $\mathbbm{R}$, $\mathcal{C}^k \lp{1} \mathcal{I},
\text{$\mathcal{A}(\Omega)$} \rp{1}$ denotes the space of functions $f$
defined on $\mathcal{I}$ with values in $\mathcal{A}(\Omega$) such that $f,
\ldots, f^{(k)}$ exist and are continuous. We always consider these spaces
with their standard Frechet structure (the semi-norms are indexed by compact
sets and eventually differentiation order).

Let $\mathfrak{B}$ be the collection of all Borel sets on $\mathbbm{R}^m$. A
$\mathbbm{C}$-valued measure $\mu$ on $\mathbbm{R}^m$ is a
$\mathbbm{C}$-valued function on $\mathfrak{B}$ satisfying the classical
countable additivity property \ [Ru]. We denote by $| \mu |$ the positive
measure defined \ by
\[ | \mu | (E) = \sup \sum_{j = 1}^{\infty} | \mu (E_j) | (E \in
   \mathfrak{B}), \]
the supremum being taken over all partition $\{E_j \}$ of $E$. In particular
$| \mu | (\mathbbm{R}^m) < \infty$. If there exists some measurable function
$\rho_{\mu}$ such that $d \mu (\xi) = \rho_{\mu} (\xi) d \xi$, then $d| \mu |
(\xi) = | \rho_{\mu} (\xi) |d \xi$.

We refer to [Ha4] for a rigorous definition of Borel and Laplace transform.
Roughly speaking, assuming that $f$ (respectively $\hat{f}$) is a function of
a complex variable $h$ (respectively $\zeta$), $f$ is the Laplace transform of
$\hat{f}$ if
\begin{equation}
  \label{venus7} \text{$f (h) = \int_0^{+ \infty} \hat{f} (\zeta) e^{-
  \frac{\zeta}{h}} \frac{d \zeta}{h}$},
\end{equation}
whereas the (formal) Borel transform of the formal power series $\tilde{f} =
\sum_{r \geqslant 0} a_r h^r$ \ is $\hat{f}$ defined by
\[ \text{$\hat{f} (\zeta) = \sum_{r = 0}^{\infty} \frac{a_r}{r!} \zeta^r$} .
\]
With suitable assumptions, these two transforms are inverse of each other. We
say that a formal power series $\tilde{f} = \sum_{r \geqslant 0} a_r h^r$ is
Borel summable if its Borel transform $\hat{f}$ has a non-vanishing radius of
convergence near 0, has an analytic continuation (still denoted by $\hat{f}$)
on a domain $D_{\mathbbm{R}^+, \kappa}$ with $\kappa > 0$ and is exponentially
dominated on this domain. The Borel sum of $\tilde{f}$ is by definition the
Laplace transform of this analytic continuation (see [Ha4] for rigourous
definitions). In the whole paper, sums indexed by an empty set are, by
convention, equal to zero.

\begin{definition}
  \label{theorembetty1}Let $\varepsilon > 0$ and let $\mu$ be a
  $\mathbbm{C}$-valued measure on $\mathbbm{R}^{\nu}$. Let us denote
  \[ M_{\mu} \assign \int_{\mathbbm{R}^{\nu}} e^{5| \xi |} d^{\nu} | \mu |
     (\xi), \]
  \[ M_{\mu, \varepsilon} \assign \int_{\mathbbm{R}^{\nu}} e^{\varepsilon
     \xi^2 + 5| \xi |} d^{\nu} | \mu | (\xi) . \]
  If there exists a differentiable function on $\mathbbm{R}^{\nu}$ such that
  \[ d \mu (\xi) = \rho_{\mu} (\xi) d \xi, \]
  let us denote
  \[ M'_{\mu} \assign \int_{\mathbbm{R}^{\nu}} \max \lp{1} | \rho_{\mu} (\xi)
     |, | \partial_{\xi_1} \rho_{\mu} (\xi) |, \ldots, | \partial_{\xi_{\nu}}
     \rho_{\mu} (\xi) | \rp{1} e^{5| \xi |} d^{\nu} \xi . \]
\end{definition}

If $M_{\mu} < \infty$, we shall associate to $\mu$ the operator $H$ defined by
(\ref{introductionfred0.5}) where
\[ c (x) = \int_{\mathbbm{R}^{\nu}} \exp (ix \cdot \xi) d \mu (\xi) . \]
Notice that the function $c$ is $\mathbbm{C}$-valued analytic and bounded on
$\mathbbm{R}^{\nu}$.

In the following proposition, we give a Borel summation statement concerning
the heat kernel associated to the operator $H$.

Let $\kappa > 0$. Let
\[ \parabole{\kappa} \assign \lbc{1} z \in \mathbbm{C} | | \mathcal{I} mz^{1 /
   2} |^2 < \kappa \rbc{1} = \lbc{1} z \in \mathbbm{C} | \mathcal{R} ez >
   \frac{1}{4 \kappa} \mathcal{I} m^2 z - \kappa \rbc{1} . \]
$\parabole{\kappa}$ is the interior of a parabola (see [Ha4, fig 2.2.]) and
\[ D_{\mathbbm{R}^+, \kappa} \subset \parabole{\kappa} . \]
We also denote
\[ \demiplan{\assign \{z \in \mathbbm{C} | \re z > 0\}.} \]
\begin{proposition}
  \label{propositionbetty1}Let $T, \varepsilon, \omega_1, \ldots, \omega_{\nu}
  > 0$. Let $\mu$ be a $\mathbbm{C}$-valued measure on $\mathbbm{R}^{\nu}$
  such that
  \begin{equation}
    \label{propositionbetty1.2} 4 T^2 e^T M_{\mu, \varepsilon} < 1.
  \end{equation}
  Then there exist $\kappa, K, K_1 > 0$, $\phi \in \mathcal{C}^0 \lp{1}] 0, T
  [, \mathcal{A} \lp{0} \mathbbm{R}^{2 \nu} \rp{0} \rp{1}$ and $\hat{W} \in
  \mathcal{C}^0 \lp{1}] 0, T [, \mathcal{A} \lp{0} \mathbbm{R}^{2 \nu} \times
  S_{\kappa} \rp{0} \rp{1}$ such that, for every $(t, x, y) \in] 0, T [\times
  \mathbbm{R}^{2 \nu}$,
  \begin{itemize}
    \item for every $\sigma \in S_{\kappa}$,
    \begin{equation}
      \label{propositionbetty1.4} | \hat{W} (t, x, y, \sigma) | \leqslant K_1
      e^{K| \sigma |^{1 / 2}},
    \end{equation}
    \item for every $h \in \demiplan{}$
    \begin{equation}
      \label{propositionbetty1.6} \lba{1} x|e^{- \frac{t}{h} H} |y \rba{1} =
      h^{- \nu / 2} e^{- \frac{1}{h} \phi (t, x, y)} \int_0^{+ \infty} \hat{W}
      (t, x, y, \sigma) e^{- \frac{\sigma}{h}} \frac{d \sigma}{h} .
    \end{equation}
  \end{itemize}
\end{proposition}

Let us assume that the function $c$ takes its values in $\mathbbm{R}$. Then
the operator $H$ defined by (\ref{introductionfred0.5}) is self-adjoint on
$L^2 (\mathbbm{R}^{\nu})$, its spectrum is discrete and its partition function
$\Theta_H$ is well defined on $] 0, + \infty [^2$.

For $\omega_1, \ldots, \omega_{\nu} > 0$, we denote
\[ \omega_{\flat} \assign \min (\omega_1, \ldots, \omega_{\nu}) \text{ , \ }
   \omega_{\sharp} \assign \max (\omega_1, \ldots, \omega_{\nu}) . \]
The next theorem is the main goal of this paper.

\begin{theorem}
  \label{theorembetty2}Let $T, \varepsilon, \omega_1, \ldots, \omega_{\nu} >
  0$. Let $\mu$ be a $\mathbbm{C}$-valued measure on $\mathbbm{R}^{\nu}$ such
  that
  \begin{equation}
    \label{theorembetty2.2} 4 T^2 e^T M_{\mu, \varepsilon} < 1,
  \end{equation}
  
  \begin{equation}
    \label{theorembetty2.3} \frac{\omega_{\sharp} (1 + \omega_{\flat} T) \ch
    \lp{1} \frac{\omega_{\sharp} T}{2} \rp{1} M_{\mu}'}{\omega_{\flat}^3
    \lp{1} 1 - 4 T^2 M_{\mu} \rp{1}} < \alpha_{\nu},
  \end{equation}
  where the constant $\alpha_{\nu}$ only depends{\footnote{See Proposition
  \ref{preparationclara4}.}} on the dimension $\nu$. Moreover, let us assume
  that
  \begin{equation}
    \label{theorembetty3} \bar{\rho}_{\mu} (\xi) = \rho_{\mu} (- \xi),
  \end{equation}
  \begin{equation}
    \label{theorembetty4} \int_{\mathbbm{R}^{\nu}} \xi_1 d \mu (\xi) = \cdots
    = \int_{\mathbbm{R}^{\nu}} \xi_{\nu} d \mu (\xi) = 0,
  \end{equation}
  Let $T_0 \in] 0, T [$. Then there exist $\kappa, K, K_1 > 0$ and
  $\hat{\theta} \in \mathcal{C}^0 \lp{1}] T_0, T [, \mathcal{A} \lp{0}
  S_{\kappa} \rp{0} \rp{1}$ such that, for every $t \in] T_0, T [$,
  \begin{itemize}
    \item for every $\sigma \in S_{\kappa}$,
    \begin{equation}
      \label{theorembetty5.05} | \hat{\theta} (t, \sigma) | \leqslant K_1
      e^{K| \sigma |^{1 / 2}},
    \end{equation}
    \item for every $h \in \mathbbm{C}^+$,
    \begin{equation}
      \label{theorembetty5.1} \Theta^{\tmop{conj}}_H (t, h) = \int_0^{+
      \infty} \hat{\theta} (t, \sigma) e^{- \frac{\sigma}{h}} \frac{d
      \sigma}{h} .
    \end{equation}
  \end{itemize}
  Therefore the expansion (\ref{introductionfred1}) is Borel summable with
  respect to $h$ and its Borel sum is equal to $\Theta^{\tmop{conj}}_H (t,
  h)$.
\end{theorem}

\begin{remark}
  {\tmdummy}
  
  \begin{itemize}
    \item Let $V$ be the potential defined by
    \[ V (x) \assign \frac{( \tmmathbf{\omega} x)^2}{4} - c (x) . \]
    (\ref{theorembetty3}) means that $V|_{\mathbbm{R}^{\nu}}$ takes real
    values and (\ref{theorembetty4}) means that
    \[ \partial_{x_1} V (0) = \cdots = \partial_{x_{\nu}} V (0) = 0. \]
    \item If $c$ is the null function then
    \[ \Theta^{\tmop{conj}}_H (t, h) = \Theta_H (t, h) = \prod_{\upsilon =
       1}^{\nu} \frac{1}{e^{\omega_{\upsilon} t / 2} - e^{- \omega_{\upsilon}
       t / 2}} . \]
    \item Both of the proofs of Proposition \ref{propositionbetty1} and
    Theorem \ref{theorembetty2} use the tree graph equality. However, proving
    Proposition \ref{propositionbetty1} is easier than proving Theorem
    \ref{theorembetty2}. For the proof of the theorem, the Gaussian Borel
    summation statement (see section \ref{nevflora}) is necessary and a Morse
    lemma is needed (see Proposition \ref{preparationclara4}).
  \end{itemize}
\end{remark}

\section{Mould formalism and combinatorics related to graphs}

Let $\Omega$ be an arbitrary set. Let us denote by $\tmop{seq} (\Omega)$
(respectively $\mathcal{P}_0 (\Omega)$) the set of finite (eventually empty)
ordered sequences of elements of $\Omega$ (respectively finite subsets of
$\Omega$). Let $\mathcal{A}$ be a commutative $\mathbbm{C}$-algebra. Let
$\mathcal{M}^{\tmop{cl}} (\Omega) =\mathcal{A}^{\tmop{seq} (\Omega)}$
(respectively $\mathcal{M}^{\tmop{ab}} (\Omega) =\mathcal{A}^{\mathcal{P}_0
(\Omega)}$). Equipped with the following sum and product
\[ C = A + B \Longleftrightarrow C^{\varsigma_1, \ldots, \varsigma_r} =
   A^{\varsigma_1, \ldots, \varsigma_r} + B^{\varsigma_1, \ldots, \varsigma_r}
\]
\[ C = A \times B \Longleftrightarrow C^{\varsigma_1, \ldots, \varsigma_r} =
   \sum_{i = 0}^r A^{\varsigma_1, \ldots, \varsigma_i} B^{\varsigma_{i + 1},
   \ldots, \varsigma_r}, \]
respectively
\[ C = A + B \Longleftrightarrow C^I = A^I + B^I \]
\[ C = A \times_{\tmop{sym}} B \Longleftrightarrow C^I =
   \sum_{\tmscript{\begin{array}{c}
     J \cup K = I\\
     J \cap K = \varnothing
   \end{array}}} A^J B^K, \]
$\mathcal{M}^{\tmop{cl}} (\Omega)$ and $\mathcal{M}^{\tmop{ab}} (\Omega)$ are
algebras. For instance
\[ C^{\{1, 2\}} = A^{\{1, 2\}} B^{\varnothing} + A^{\{1\}} B^{\{2\}} +
   A^{\{2\}} B^{\{1\}} + A^{\varnothing} B^{\{1, 2\}} . \]
The algebra $\mathcal{M}^{\tmop{ab}} (\Omega)$ is commutative. Elements of
$\mathcal{M}^{\tmop{cl}} (\Omega)$ or $\mathcal{M}^{\tmop{ab}} (\Omega)$ are
called moulds. We denote by $1$ the identity element of
$\mathcal{M}^{\tmop{ab}} (\Omega)$. Then $1^{\varnothing} = 1$ and $1^I = 0$
if $|I| \geqslant 1$. Let $A \in \mathcal{M}^{\tmop{ab}} (\Omega)$ be such
that $A^{\varnothing} = 0$ and let $\varphi \in \mathbbm{C} \lb{0} [H]
\rb{0}$. Then the mould $\varphi (A)$ is well defined (only finite sums occur
in its definition). The following elementary fact will be useful. Let $A^I =
\lp{1} \prod_{j \in I} \gamma_j \rp{1} B^I$ where $\gamma_j \in
\mathbbm{C}_{}$ and $B \in \mathcal{M}^{\tmop{ab}} (\Omega), B^{\varnothing} =
0$. Then
\[ \lp{1} \varphi (A) \srp{1}{I}{} = \lp{1} \prod_{j \in I} \gamma_j \rp{1}
   \lp{1} \varphi (B) \srp{1}{I}{} . \]
\begin{remark}
  A lot of important symmetries occur in $\mathcal{M}^{\tmop{cl}} (\Omega)$.
  In particular, a mould $A$ is said to be symmetral if for every sequence
  $\varsigma^1$ and $\varsigma^2$
  \[ A^{\varsigma^1} A^{\varsigma^2} = \sum_{\varsigma \in \tmop{sh}
     (\varsigma^1, \varsigma^2)} A^{\varsigma} \]
  where $\tmop{sh} (\varsigma^1, \varsigma^2)$ denotes the set of all
  sequences $\varsigma$ obtaining from $\varsigma^1$ and $\varsigma^2$ under
  shuffling. The following lemma (Lemma \ref{mouldofamanda2}) is related to
  this fundamental notion (see also [Ha3, Prop.3.1]). See [E1, E2, E-V] for
  general aspects of the mould formalism.
\end{remark}

For $m = 1, 2, \ldots$, let
\[ T_m \assign \lbc{1} (s_1, \ldots, s_m) \in [0, 1] \ve{1} 0 < s_1 < \cdots <
   s_m < 1\} \times \mathbbm{R}^{\nu m} \]
and let us denote $T_{\infty} \assign \{\varnothing\} \sqcup T_1 \sqcup T_2
\sqcup \cdots$. Let $\mathcal{M}^{\tmop{pre}}$ be the algebra of
$\mathcal{A}$-valued functions $f$ defined on $T_{\infty}$, such that
$f|_{T_m}$ is measurable for every $m \geqslant 1$, equipped with the trivial
sum and the following commutative product
\[ h = f \times_{\tmop{pre}} g \Leftrightarrow h (s_1, \ldots, s_m ; \xi_1,
   \ldots, \xi_m) = \]
\[ \sum_{\tmscript{\begin{array}{c}
     J \cup K =\{1, \ldots, m\}\\
     J \cap K = \varnothing
   \end{array}}} f (s_{j_1}, \ldots, s_{j_p} ; \xi_{j_1}, \ldots, \xi_{j_p}) g
   (s_{k_1}, \ldots, s_{k_q} ; \xi_{k_1}, \ldots, \xi_{k_q}) . \]
Here $j_1 < \cdots < j_p$ (respectively $k_1 < \cdots < k_q$) denote the
elements of $J$ (respectively $K$). Let $\lambda$ be a $\mathbbm{C}$-valued
Borel measure defined on $[0, 1] \times \mathbbm{R}^{\nu}$ and let us denote
by $\lambda^{\otimes_m}$ (respectively $| \lambda |^{\otimes_m}$) the Borel
measure defined on $T_m$ by
\[ d \lambda^{\otimes_m} (s, \xi) = d \lambda (s_1, \xi_1) \cdots d \lambda
   (s_m, \xi_m), \]
\[ d| \lambda |^{\otimes_m} (s, \xi) = d| \lambda | (s_1, \xi_1) \cdots d|
   \lambda | (s_m, \xi_m) \text{ \ \ \ \ (respectively)} . \]
Let $\mathcal{M}^{\tmop{pre}}_1$ be the subalgebra of
$\mathcal{M}^{\tmop{pre}}$ of all functions $f$ such that for every $m
\geqslant 1$ $f|_{T_m}$ is integrable on $T_m$ with respect to $| \lambda
|^{\otimes_m}$.

\begin{lemma}
  \label{mouldofamanda2}The mapping $\Phi$
  \[ \begin{array}{lll}
       \mathcal{M}^{\tmop{pre}}_1 & \longrightarrow & \mathcal{A} \lb{1} [H]
       \rb{1}\\
       &  & \\
       f & \longmapsto & \sum_{m \geqslant 0} H^m \int_{T^m} f (s_1, \ldots,
       s_m ; \xi_1, \ldots, \xi_m) d \lambda^{\otimes_m} (s, \xi)
     \end{array} \]
  is an algebra morphism.
\end{lemma}

\begin{proof}
  Let $m \geqslant 1$. One gets
  \[ \sum_{p + q = m} \int_{T_p} f (s_1, \ldots, s_p ; \xi_1, \ldots, \xi_p) d
     \lambda^{\otimes_p} (s, \xi) \times \int_{T_q} g (s_1, \ldots, s_q ;
     \xi_1, \ldots, \xi_q) d \lambda^{\otimes_q} (s, \xi) \]
  \[ = \int_{T_m} (f \times_{\tmop{pre}} g) (s_1, \ldots, s_m ; \xi_1, \ldots,
     \xi_m) d \lambda^{\otimes_m} (s, \xi) . \]
  It is a consequence of the shuffling property concerning the following
  characteristic functions:
  \[ 1_{0 < u_1 < \cdots < u_p < 1} \times 1_{0 < v_1 < \cdots < v_q < 1} = \]
  \[ \sum_{(w_1, \ldots, w_{p + q}) \in \tmop{sh} \lp{1} (u_1, \ldots, u_p),
     (v_1, \ldots, v_q) \rp{1}} 1_{0 < w_1 < \cdots < w_{p + q} < 1} . \]
  Then the mapping defined in Lemma \ref{mouldofamanda2} is a multiplicative
  morphism.
\end{proof}

\subsection{Some identities}\label{identities}

Let $I$ be a subset of $\mathbbm{N}$ such that $2 \leqslant |I| < \infty$. We
denote by $\mathcal{G}_I$ the set of (unordered) connected graphs on $I$. A
connected graph with no cycles is called a (unordered) tree and we denote by
$\mathcal{T}_I$ the set of trees on $I$. For instance
\[ \mathcal{T}_{\{1, 2, 4\}} = \lbc{2} \lbc{1} [1, 2], [1, 4] \rbc{1}, \lbc{1}
   [1, 2], [2, 4] \rbc{1}, \lbc{1} [1, 4], [2, 4] \rbc{1} \rbc{2} . \]
If $|I| = n$ then $|\mathcal{T}_I | = n^{n - 2}$. Let $g \in \mathcal{G}_I$.
An element $\ell$ of $g$ is called a edge and $\ell = [j, k]$ where $j, k \in
I, j \neq k$ and we always assume that $j < k$ by convention.

\begin{proposition}
  \label{treemaria10} (tree graph equality) For $1 \leqslant j < k < \infty$,
  let $\tilde{z}_{j, k} \in \mathbbm{C}$. Let $A$ and $B$ be the moulds
  defined by
  \[ A^{\varnothing} = 1, A^{\{j\}} = 1, B^{\varnothing} = 0, B^{\{j\}} = 1 \]
  and for $I \subset \mathbbm{N}^{\ast}$, $2 \leqslant |I| < \infty$,
  \[ A^I = \exp \lp{2} \sum_{j, k \in I, j < k} \tilde{z}_{j, k} \rp{2}, \]
  \begin{equation}
    \label{treemaria10.5} (B)^I = \sum_{g \in \mathcal{\mathcal{T}}_I} \lp{2}
    \prod_{\tmscript{\begin{array}{c}
      {}[j, k] \in g
    \end{array}}} \tilde{z}_{j, k} \rp{2} \int_{\theta \in [0, 1]^g}
    e^{\sum_{j, k \in I, j < k} \theta_{j, k, g} \tilde{z}_{j, k}} d^{|I| - 1}
    \theta
  \end{equation}
  (the tree $g$ contains $|I| - 1$ elements). Here
  \[ d^{|I| - 1} \theta \assign \prod_{[ \bar{j}, \bar{k}] \in g} d \theta_{[
     \bar{j}, \bar{k}]} \]
  and
  \begin{equation}
    \label{treemaria10.6} \theta_{j, k, g} = \min_{\tmscript{\begin{array}{c}
      {}[p, q]
    \end{array}}} \theta_{[p, q]}
  \end{equation}
  where $[p, q]$ runs over all edges belonging to the unique path joining $j$
  and $k$ in the tree $g$. Then
  \begin{equation}
    \label{treemaria11} A = e^B .
  \end{equation}
\end{proposition}

\begin{remark}
  \label{treemaria11.2}{\tmdummy}
  
  \begin{enumeratenumeric}
    \item Here is an example illustrating the definition of $\theta_{j, k,
    g}$. Let
    \[ g = \lbc{1} [1, 2], [2, 3], [2, 8], [5, 8], [5, 6], [5, 7], [4, 6]
       \rbc{1} . \]
    Then $\theta_{2, 6, g} = \min \lp{1} \theta_{[2, 8]}, \theta_{[5, 8]},
    \theta_{[5, 6]} \rp{1}$.
    
    \item The exponential in (\ref{treemaria11}) is defined by
    \[ e^A \assign \sum_{m \geqslant 0} \frac{1}{m!} \underbrace{A
       \times_{\tmop{sym}} \cdots \times_{\tmop{sym}} A}_{m \tmop{times}} . \]
    \item \label{treemaria11.3}The mould $B$ has a simpler expression
    \begin{equation}
      \label{treemaria11.5} B^I = \sum_{g \in \mathcal{G}_I} 
      \prod_{\tmscript{\begin{array}{c}
        {}[j, k] \in g
      \end{array}}} (e^{\tilde{z}_{j, k}} - 1) .
    \end{equation}
    But $\mathcal{|G}_I | \approx 2^{|I| (|I| - 1) / 2}$ whereas
    $|\mathcal{T}_I | = |I|^{|I| - 2} = e^{(|I| - 2) \ln (|I|)}$. Therefore
    (\ref{treemaria10.5}) is more efficient than (\ref{treemaria11.5}) to
    prove the convergence of series containing terms like $B^I$. However,
    (\ref{treemaria11.5}) can be useful for explicit computations [Fu-Os-Wi,
    Ha3].
  \end{enumeratenumeric}
\end{remark}

\begin{remark}
  Equality (\ref{treemaria11}) is proved in [B-K] or [Br] (set $\dot{u}_{k, l}
  (s) = - \tilde{z}_{k, l} 1_{[0, 1]} (s)$ and $t = 1$ in [B-K, Th.3.1]). Let
  us give the idea of this proof translated in the combinatorial language of
  moulds (see also Section 7 in [Br] where the definition of the product
  $\times_{\tmop{sym}}$ is given). Let us consider the moulds defined by
  \[ A_t^{\varnothing} = 1, A_t^{\{j\}} = 1, B_t^{\varnothing} = 0,
     B_t^{\{j\}} = 1, \]
  \[ A_t^I \assign \exp \lp{2} t \sum_{j, k \in I} \tilde{z}_{j, k} \rp{2}, \]
  \begin{equation}
    \label{trimaria11.7} (B_t)^I \assign \sum_{g \in \mathcal{\mathcal{T}}_I}
    \lp{2} \int_0^t \srp{2}{|I| - 1}{} \lp{2}
    \prod_{\tmscript{\begin{array}{c}
      {}[ \bar{j}, \bar{k}] \in g
    \end{array}}} \tilde{z}_{\bar{j}, \bar{k}} e^{\sum_{j, k \in I, j < k} (t
    - t_{j, k, g}) \tilde{z}_{j, k}} d t_{[ \bar{j}, \bar{k}]} \rp{2},
  \end{equation}
  \[ t_{j, k, g} = \max_{\tmscript{\begin{array}{c}
       {}[p, q]
     \end{array}}} t_{[p, q]} . \]
  Here $[p, q]$ runs over all edges belonging to the unique path joining $j$
  and $k$ in the tree $g$. Let us notice that $e^{B_t} |_{t = 0} = A_t |_{t =
  0}$ and that the equality (\ref{treemaria11}) is equivalent to $e^{B_t} |_{t
  = 1} = A_t |_{t = 1}$. The mould $A_t$ satisfies the differential equation
  \[ \partial_t A_t = \lp{2} \sum_{j < k} \tilde{z}_{j, k} \Delta_j \Delta_k
     \rp{2} A_t \]
  where $\Delta_j$ is the elementary mould's derivation defined by $(\Delta_j
  A)^I = 1_{j \in I} A^I$. Therefore, since $\Delta_j$ is a derivation and
  $\times_{\tmop{sym}}$ is commutative, (\ref{treemaria11}) is satisfied if
  \begin{equation}
    \label{trimaria11.8} \partial_t B_t = \lp{2} \sum_{j < k} \tilde{z}_{j, k}
    \Delta_j \Delta_k \rp{2} B_t + \lp{2} \sum_{j < k} \tilde{z}_{j, k}
    \Delta_j B_t \times_{\tmop{sym}} \Delta_k B_t \rp{2} .
  \end{equation}
  Differentiating the integrand in (\ref{trimaria11.7}) with respect to $t$
  explains the first term of the right hand side of (\ref{trimaria11.8}). The
  core of the proof is to check that the non-linear term in
  (\ref{trimaria11.8}) comes from the differentiation of the upper limit of
  the integrals in (\ref{trimaria11.7}).
\end{remark}

Let $h \in \mathbbm{C}^{\ast}$ and for $j, k = 1, 2, \ldots$, $j \leqslant k$,
let $z_{j, k} \in \mathbbm{C}_{}$. Let $C_h \in \mathcal{M}^{\tmop{ab}}
(\mathbbm{N}^{\ast})$ be defined by
\[ C_h^{\varnothing} = 1, C_h^I = \frac{1}{h^{|I|}} \times e^{h \sum_{j, k \in
   I} z_{j \wedge k, j \vee k}} . \]
Later we shall consider the behaviour of quantities involving the mould $C_h$
when $h$ goes to $0$. The next proposition simplifies their study. Let $E, R_h
\in \mathcal{M}^{\tmop{ab}} (\mathbbm{N}^{\ast})$ be defined by
\[ E^{\varnothing} = 0, E^{\{j\}} = 1, E^I = 2^{|I| - 1} \sum_{g \in
   \mathcal{\mathcal{T}}_I} \prod_{\tmscript{\begin{array}{c}
     {}[j, k] \in g
   \end{array}}} z_{j, k}, \]
\[ R_h^{\varnothing} = 0, R_h^{\{j\}} = z_{j, j} \int_0^1 e^{h \vartheta z_{j,
   j}} d \vartheta, \]
\[ R_h^I = 2^{|I| - 1} \sum_{g \in \mathcal{\mathcal{T}}_I} \lp{2}
   \prod_{\tmscript{\begin{array}{c}
     {}[j, k] \in g
   \end{array}}} z_{j, k} \rp{2} \times \int_{[0, 1]^g \times [0, 1]} \]
\[ \lp{2} \sum_{j \in I} z_{j, j} + 2 \sum_{j, k \in I, j < k} \theta_{j, k,
   g} z_{j, k} \rp{2} e^{h \vartheta \lp{1} \sum_{j \in I} z_{j, j} + 2
   \sum_{j, k \in I, j < k} \theta_{j, k, g} z_{j, k} \rp{1}} d^{|I| - 1}
   \theta d \vartheta . \]
\begin{remark}
  \label{treemaria11.99}The mould $E$ does not depend on $h$ and the mould
  $R_h$ (unlike the mould $C_h$) is not singular when $h$ goes to $0$.
\end{remark}

\begin{proposition}
  \label{treemaria12}{\tmdummy}
  
  \[ C_h = \exp \lp{2} \frac{1}{h} E + R_h \rp{2} . \]
\end{proposition}

\begin{proof}
  One gets
  \[ (\log C_h)^{\varnothing} = 0, (\log C_h)^{\{j\}} = h^{- 1} e^{h z_{j, j}}
     . \]
  Then the identity $(\log C_h)^I = \lp{1} \frac{1}{h} E + R_h \srp{1}{I}{}$
  for $|I| \leqslant 1$ is straightforward. Let $D_h$ be the mould defined by
  \[ D_h^{\varnothing} = 0, D_h^{\{j\}} = 1, D_h^I = e^{2 h \sum_{j, k \in I,
     j < k} z_{j, k}} . \]
  Let $I \subset \mathbbm{N}^{\ast}$ \ be such that $2 \leqslant |I| <
  \infty$. Then
  \begin{eqnarray*}
    (\log C_h)^I = &  & h^{- |I|} \prod_{j \in I} e^{h z_{j, j}} \times \lp{1}
    \log (1 + D_h) \srp{1}{I}{}\\
    = &  & 2^{|I| - 1} h^{- 1} \sum_{g \in \mathcal{\mathcal{T}}_I}\\
    &  & \lp{2} \prod_{\tmscript{\begin{array}{c}
      {}[j, k] \in g
    \end{array}}} z_{j, k} \rp{2} \int_{[0, 1]^g} e^{h \lp{1} \sum_{j \in I}
    z_{j, j} + 2 \sum_{j, k \in I, j < k} \theta_{j, k, g} z_{j, k} \rp{1}}
    d^{|I| - 1} \theta .
  \end{eqnarray*}
  The last equality is obtained by choosing $\tilde{z}_{j, k} = 2 h z_{j, k}$
  in Proposition \ref{treemaria10}. Then choosing $\varphi = h \lp{1} \sum_{j
  \in I} z_{j, j} + 2 \sum_{j, k \in I, j < k} \theta_{j, k, g} z_{j, k}
  \rp{1}$ in the identity
  \[ e^{\varphi} = 1 + \varphi \int_0^1 e^{\varphi \vartheta} d \vartheta \]
  yields the decomposition $(\log C_h)^I = \lp{1} \frac{1}{h} E + R_h
  \srp{1}{I}{}$. This proves Proposition \ref{treemaria12}.
\end{proof}

\begin{corollary}
  \label{treemaria14}Let $\alpha$ be a $\mathbbm{C}$-valued function on $[0,
  1]^2 \times \mathbbm{R}^{2 \nu}$. Let $f, g \in \mathcal{M}^{\tmop{pre}}$ be
  defined by
  \[ f (s_1, \ldots, s_m, \xi_1, \ldots, \xi_m) = \lp{1} \frac{1}{h} E + R_h
     \srp{1}{\{1, \ldots, m\}}{} \sve{2}{}{z_{j, k} = \alpha (s_j, s_k, \xi_j,
     \xi_k)} \]
  and
  \[ g (s_1, \ldots, s_m, \xi_1, \ldots, \xi_m) \assign C_h^{\{1, \ldots, m\}}
     \sve{2}{}{z_{j, k} = \alpha (s_j, s_k, \xi_j, \xi_k)} . \]
  Then
  \[ g = \exp_{\tmop{pre}} (f) \]
  ($\exp_{\tmop{pre}}$ is defined according to the product
  $\times_{\tmop{pre}}$).
\end{corollary}

\section{ Deformation formulas}

Let $(x, y) \in \mathbbm{C}^{2 \nu}$, $t \in \mathbbm{C}^{\ast}$, $h \in
\mathbbm{C}^{\ast}$ and $\omega_1, \ldots, \omega_{\nu} > 0$. Let
$p^{\tmop{harm}}$ be defined by
\[ \text{$p^{\tmop{harm}}$} = \left( 4 \pi h \right)^{- \nu / 2}
   \prod_{\upsilon = 1}^{\nu} \lp{2} \frac{\omega_{\upsilon}}{\tmop{sh}
   (\omega_{\upsilon} t)} \srp{2}{1 / 2}{} \times \]
\begin{equation}
  \label{formulaclara5} \exp \lp{2} - \frac{1}{4 h} \sum_{\upsilon = 1}^{\nu}
  \frac{\omega_{\upsilon}}{\tmop{sh} (\omega_{\upsilon} t)} \lp{1} \tmop{ch}
  (\omega_{\upsilon} t) (x_{\upsilon}^2 + y_{\upsilon}^2) - 2 x_{\upsilon}
  y_{\upsilon} \rp{1} \rp{2} .
\end{equation}
Then, by a variant of Mehler's formula,
\[ \text{$h \partial_t p^{\tmop{harm}} = \lp{1} h^2 \partial^2_x - \frac{(
   \tmmathbf{\omega} x)^2}{4} \rp{1} p^{\tmop{harm}}$} \text{ , \ }
   p^{\tmop{harm}} |_{t = 0^+} =_{} \delta_{x = y} . \]
Let $\xi = (\xi_1, \ldots, \xi_m) \in \mathbbm{R}^{\nu m}$. Let
\[ | \xi |_1 \assign | \xi_1 | + \cdots + | \xi_m |, \]
and for $s \in [0, 1]$, $\upsilon \in \{1, \ldots \nu\}$,
\[ q_{t, \upsilon} (s) \assign \frac{1}{\tmop{sh} (\omega_{\upsilon} t)}
   \lp{2} \tmop{sh} (\omega_{\upsilon} t s) x_{\upsilon} + \tmop{sh} \lp{1}
   \omega_{\upsilon} t (1 - s) \rp{1} y_{\upsilon} \rp{2}, \]
\[ q_t (s) \assign \lp{1} q_{t, 1} (s), \ldots, q_{t, \nu} (s) \rp{1}, \]
and for $s = (s_1, \ldots, s_m) \in [0, 1]^m$
\[ q_t^m (s) \cdot \xi \assign q_t (s_1) \cdot \xi_1 + \cdots + q_t (s_m)
   \cdot \xi_m . \]
Let $\Omega_t . \xi \otimes_n \xi$ be defined by
\[ \Omega_t . \xi \otimes_m \xi \assign \sum_{j, k = 1}^m \sum_{\upsilon =
   1}^{\nu} \frac{\tmop{sh} (\omega_{\upsilon} ts_{j \wedge k}) \tmop{sh}
   \lp{1} \omega_{\upsilon} t (1 - s_{j \vee k}) \rp{1}}{\omega_{\upsilon}
   \tmop{sh} (\omega_{\upsilon} t)} \xi_{j, \upsilon} \xi_{k, \upsilon} . \]
We also need the following generalization. Let $m \geqslant 2$. Let $g \in
\mathcal{T}_m$ and $\theta \in [0, 1]^g$. Let $\theta_{j, k, g}$ be defined by
(\ref{treemaria10.6}). Let $\Omega_t^{g, \theta} . \xi \otimes_m \xi$ be
defined by
\[ \Omega_t^{g, \theta} . \xi \otimes_m \xi \assign \sum_{j = 1}^m
   \sum_{\upsilon = 1}^{\nu} \frac{\sh (\omega_{\upsilon} t s_j) \sh \lp{1}
   \omega_{\upsilon} t (1 - s_j) \rp{1}}{\omega_{\upsilon} \sh
   (\omega_{\upsilon} t)} \xi_{j, \upsilon}^2 + \]
\[ 2 \sum_{1 \leqslant j < k \leqslant m} \theta_{j, k, g} \sum_{\upsilon =
   1}^{\nu} \frac{\tmop{sh} (\omega_{\upsilon} ts_j) \tmop{sh} \lp{1}
   \omega_{\upsilon} t (1 - s_k) \rp{1}}{\omega_{\upsilon} \tmop{sh}
   (\omega_{\upsilon} t)} \xi_{j, \upsilon} \xi_{k, \upsilon} . \]
\begin{lemma}
  \label{formulaclara10}Let $\omega_1, \ldots, \omega_{\nu}, t \geqslant 0$.
  Let $g$ and $\theta$ be as above. Let $(s_1, \ldots, s_m) \in [0, 1]^m$ such
  that $0 < s_1 < \cdots < s_m < 1$ and let $\xi = (\xi_1, \ldots, \xi_m) \in
  \mathbbm{R}^{\nu m}$. Then
  \begin{equation}
    \label{formulaclara12a} 0 \leqslant \Omega_t . \xi \otimes_m \xi \leqslant
    \frac{m t}{4} \lp{1} \xi_1^2 + \cdots + \xi_m^2 \rp{1},
  \end{equation}
  \begin{equation}
    \label{formulaclara12b} 0 \leqslant \Omega_t^{g, \theta} . \xi \otimes_m
    \xi \leqslant \frac{m t}{4} \lp{1} \xi_1^2 + \cdots + \xi_m^2 \rp{1} .
  \end{equation}
\end{lemma}

\begin{proof}
  Let $\dot{\omega} \geqslant 0$. We claim that the $m \times m$ matrix $M$
  defined by
  \[ M_{j, k} = \frac{\tmop{sh} ( \dot{\omega} ts_{j \wedge k}) \tmop{sh}
     \lp{1} \dot{\omega} t (1 - s_{j \vee k}) \rp{1}}{\dot{\omega} t \tmop{sh}
     ( \dot{\omega} t)} \]
  is symmetric non-negative. Let $\zeta_1, \ldots, \zeta_m \in \mathbbm{R}$
  and let $u : [0, 1] \longrightarrow \mathbbm{R}$ defined by
  \[ u (s) \assign \sum_{j = 1}^m \frac{\sh ( \dot{\omega} t s_j \wedge s) \sh
     \lp{1} \dot{\omega} t (1 - s_j \vee s) \rp{1}}{\dot{\omega} t \sh (
     \dot{\omega} t)} \zeta_j . \]
  Let $\mu$ be the $\mathbbm{R}$-valued Borel measure on $[0, 1]$
  \[ \mu \assign \sum_{j = 1}^m \zeta_j \delta_{s_j} . \]
  Then $u$ is continuous, piecewise differentiable on $[0, 1]$ and (see also
  [Ha4])
  \[ \left\{ \begin{array}{l}
       - \frac{d^2 u}{d s^2} + ( \dot{\omega} t)^2 u = \mu\\
       \\
       u (0) = u (1) = 0
     \end{array} . \right. \]
  Then
  \begin{eqnarray*}
    \sum_{j, k = 1}^m M_{j, k} \zeta_j \zeta_k = &  & \int_0^1 u d \mu\\
    = &  & \int_0^1 \lbc{3} \lp{2} \frac{d u}{d s} \srp{2}{2}{} + (
    \dot{\omega} t)^2 u^2 \rbc{3},
  \end{eqnarray*}
  which proves the claim. Now, by choosing $\dot{\omega} = \omega_1, \ldots,
  \dot{\omega} = \omega_{\nu}$,
  \[ \Omega_t . \xi \otimes_m \xi \geqslant 0. \]
  Let us prove the non-negativity of $\Omega_t^{g, \theta} . \xi \otimes_m
  \xi$. By the same argument, without loss of generality, we may choose $\nu =
  1$. We shall use Lemma \ref{trimaria16} (see Appendix). Let $(u_{j, k})_{j,
  k \in \{1, \ldots, m\}}$ be the real symmetric matrix defined by $u_{j, j} =
  1$ and, for $j \neq k$, $u_{j, k} \assign \theta_{j \wedge k, j \vee k, g}$.
  For $j, k = 1, \ldots, m$, let us denote by $\lba{0} j, k \srba{0}{}{g}$ the
  unique path in the tree $g$ joining $j$ and $k$. Let $q = 1, \ldots, m$. Let
  us denote by $q_0$ the unique element belonging to $\lba{0} j, k
  \srba{0}{}{g}$ such that $\lba{0} j, k \srba{0}{}{g} \cap \lba{0} q_0, q
  \srba{0}{}{g} =\{q_0 \}$ ($q_0 = q$ if $q \in \lba{0} j, k \srba{0}{}{g}$).
  The path $\lba{0} j, k \srba{0}{}{g}$ is covered by the union of $\lba{0} j,
  q \srba{0}{}{g}$ and $\lba{0} k, q \srba{0}{}{g}$. Therefore, by
  (\ref{treemaria10.6}), (\ref{trimaria17}) holds and, by Lemma
  \ref{trimaria16},
  \[ \Omega_t^{g, \theta} . \xi \otimes_m \xi \geqslant 0. \]
  Let us prove the upper bound in (\ref{formulaclara12a}) and
  (\ref{formulaclara12b}). Since $s_{j \wedge k} \leqslant s_{j \vee k}$ and
  $\frac{\sh (x A) \sh \lp{1} x (1 - B) \rp{1}}{x \sh (x)} \leqslant
  \frac{1}{4}$ for arbitrary $x \in \mathbbm{R}$ and $0 \leqslant A < B
  \leqslant 1$, one gets
  \begin{equation}
    \label{formulaclara13} M_{j, k} \leqslant \frac{1}{4} .
  \end{equation}
  Then (\ref{formulaclara12a}) holds. Since the parameters $\theta_{j, k, g}$
  are bounded by 1, (\ref{formulaclara12b}) also holds.
\end{proof}

\begin{proposition}
  \label{formulaclara14}Let $\omega_1, \ldots, \omega_{\nu} > 0$. Let $h \in
  \mathbbm{C}$ such that $\re h > 0$. Let $\mu$ be a $\mathbbm{C}$-valued
  measure on $\mathbbm{R}^{\nu}$. Let us assume that for every $R > 0$
  \begin{equation}
    \label{formulaclara14.1} \int_{\mathbbm{R}^{\nu}} \exp (R| \xi |) d| \mu |
    (\xi) < \infty .
  \end{equation}
  Let
  \[ c (x) = \int_{\mathbbm{R}^{\nu}} \exp (ix \cdot \xi) d \mu (\xi) . \]
  Let $v$ be defined by
  \begin{equation}
    \label{formulaclara14.2} \left\{ \begin{array}{l}
      \text{$v = 1 + \sum_{m \geqslant 1} v_m$},\\
      \\
      v_m (t, x, y, h) = \lp{2} \frac{t}{h} \srp{2}{m}{} \int_{0 < s_1 <
      \cdots < s_m < 1} \int_{\mathbbm{R}^{\nu m}} e^{- h \Omega_t . \xi
      \otimes_m \xi} e^{i q_t^m (s) \cdot \xi} d^{\nu m} \mu^{\otimes} (\xi)
      d^m s.
    \end{array} \right.
  \end{equation}
  In (\ref{formulaclara14.2}), $d^m s$ denotes $ds_1 \cdots d \nonesep s_m$
  and $d^{\nu m} \mu^{\otimes} (\xi)$ denotes $d \mu (\xi_m) \cdots d \mu
  (\xi_1)$. Then, denoting $v_0 = 1$, for every $m \geqslant 1$,
  \begin{equation}
    \label{formulaclara14.3} \left\{ \begin{array}{l}
      h \lp{1} \partial_t - \frac{2}{p^{\tmop{harm}}} \partial_x
      p^{\tmop{harm}} \cdot \partial_x \rp{1} v_m \text{$= h^2 \partial^2_x
      v_m$} + c (x) v_{m - 1}\\
      \\
      v_m |_{t = 0^+} =_{} 0
    \end{array} \right. .
  \end{equation}
  Moreover $v \in \mathcal{C}^1 \lp{1} [0, + \infty [, \mathcal{A} \lp{0}
  \mathbbm{C}^{2 \nu} \times \demiplan{} \rp{0} \rp{1}$ and the function $u :
  = p^{\tmop{harm}} v$ is the solution of
  \begin{equation}
    \label{theorembetty6} \left\{ \begin{array}{l}
      h \partial_t u = h^2 \partial_x^2 u - \frac{( \tmmathbf{\omega} x)^2}{4}
      u + c (x) u\\
      \\
      u|_{t = 0^+} = \delta_{x = y}
    \end{array} . \right.
  \end{equation}
\end{proposition}

\begin{proof}
  For small values of $t$, this proposition can be viewed as a consequence of
  Theorem 2.1 in [Ha6]. The proof is very similar to the one of [ Ha4,
  Proposition 4.7] and is left to the reader. The convergence of the integrals
  defining the function $v$ for arbitrary values of $t$ uses the
  non-negativity of $\Omega_t . \xi \otimes_n \xi$ (Lemma
  \ref{formulaclara10}) and the estimate (\ref{formulaclara16.b}).
\end{proof}

\begin{remark}
  Considering simultaneously complex values of $h$ ($\re h > 0$) and complex
  values of $t$ (in a small conical neighbourhood of $\mathbbm{R}^+$)
  certainly needs an extra assumption (as in [Ha7, Prop. $4.5$, case $2$]).
  Therefore we only consider smoothness (and not analyticity) of the function
  $v$ with respect to $t$.
\end{remark}

\begin{remark}
  \label{formulaclara14.7}Let us assume that the function $c$ is
  $\mathbbm{R}$-valued. Let $t, h \in] 0, + \infty [$ and $x, y \in
  \mathbbm{R}^{\nu}$. Then $q_t^m (s) \in \mathbbm{R}^{\nu m}$ and $\Omega_t .
  \xi \otimes_m \xi \geqslant 0$. Therefore
  \[ |v_m (t, x, y, h) | \leqslant \frac{1}{m!} \lp{2} \frac{t}{h}
     \srp{2}{m}{} \lp{2} \int_{\mathbbm{R}^{\nu}} d| \mu | (\xi) \srp{2}{m}{}
  \]
  and $v (t, \cdot, \cdot, h) \in L^{\infty} (\mathbbm{R}^{2 \nu})$.
  Therefore, since \text{$p^{\tmop{harm}} (t, \cdot, \cdot, h) \in L^2
  (\mathbbm{R}^{2 \nu})$}, $u (t, \cdot, \cdot, h)$ belongs to $L^2
  (\mathbbm{R}^{2 \nu})$ and the operator $e^{- \frac{t}{h} H}$ is
  Hilbert-Schmidt. Then the operator $e^{- \frac{t}{h} H} = e^{- \frac{t}{2 h}
  H} \times e^{- \frac{t}{2 h} H}$ is trace class and
  \[ \Theta_H (t, h) = \int_{\mathbbm{R}^{\nu}} u (t, x, x, h) d x. \]
\end{remark}

In formula (\ref{formulaclara14.2}), the function $v$ looks very singular with
respect to $h$. Considering $v$ as the exponential of a new function will
allow us to deal with this singularity (see Remark \ref{treemaria11.99}). We
need some definitions.

\begin{definition}
  \label{formulaclara15}Let $\omega \in \mathbbm{R}^{\nu}$ and let $\mu$ be as
  in Proposition \ref{formulaclara14}. For every $(t, x, y) \in \mathbbm{R}
  \times \mathbbm{C}^{2 \nu}$, let us denote
  \[ Q_1 (t, x, y) = \int_{0 < s_1 < 1} \int_{\mathbbm{R}^{\nu}} e^{i q_t
     (s_1) \cdot \xi} d \mu (\xi_1) d s_1 \]
  and, for $m = 2, 3, \ldots$
  \[ Q_m (t, x, y) = (- 2)^{m - 1} \sum_{g \in \mathcal{\mathcal{T}}_m}
     \int_{0 < s_1 < \ldots < s_m < 1} \int_{\mathbbm{R}^{\nu m}} \Upsilon^{g,
     \xi}_t e^{i q_t^m (s) \cdot \xi} d^{\nu m} \mu^{\otimes} (\xi) d^m s \]
  where
  \[ \Upsilon^{g, \xi}_t \assign \prod_{[j, k] \in g} \lp{3} \sum_{\upsilon =
     1}^{\nu} \frac{\sh (\omega_{\upsilon} t s_j) \sh \lp{1} \omega_{\upsilon}
     t (1 - s_k) \rp{1}}{\omega_{\upsilon} t \sh (\omega_{\upsilon} t)}
     \xi_{j, \omega_{\upsilon}} \xi_{k, \omega_{\upsilon}} \rp{3} . \]
\end{definition}

For every $t \in \mathbbm{R}$, $s \in [0, 1]$ and $\upsilon \in \{1, \ldots,
\nu\}$, let
\[ \varpi_{t, \upsilon} (s) \assign \frac{1}{\tmop{sh} (\omega_{\upsilon} t)}
   \lp{2} \tmop{sh} (\omega_{\upsilon} t s) + \tmop{sh} \lp{1}
   \omega_{\upsilon} t (1 - s) \rp{1} \rp{2}, \]
\[ \varpi_t (s) \assign \lp{1} \varpi_{t, 1} (s), \ldots, \varpi_{t, \nu} (s)
   \rp{1} \]
Then
\begin{equation}
  \label{formulaclara15.5} \frac{1}{\ch \lp{1} \frac{\omega_{\upsilon} t}{2}
  \rp{1}} \leqslant \varpi_{t, \upsilon} (s) \leqslant 1
\end{equation}
and, for every $R > 0$,
\begin{equation}
  \label{formulaclara16.a} \forall x \in D_{\mathbbm{R}^{\nu}, R} \text{, }
  \ve{1} \im (q_t^m |_{y = x} (s) \cdot \xi) \ve{1} \leqslant | \im x| \times
  | \xi |_1 \leqslant R| \xi |_1 .
\end{equation}
\begin{equation}
  \label{formulaclara16.b} \forall (x, y) \in D^2_{\mathbbm{R}^{\nu}, R}
  \text{, } \ve{1} \im (q_t^m (s) \cdot \xi) \ve{1} \leqslant (| \im x| + |
  \im y|) \times | \xi |_1 \leqslant 2 R| \xi |_1 .
\end{equation}
Therefore the functions $Q_m$ are well defined on $\mathbbm{R} \times
\mathbbm{C}^{2 \nu}$. In the case $t \omega = 0$, these quantities are studied
in [Ha1, Ha3]. We need upper bounds for some quantities which depend on the
functions $Q_m$, for large values of $m$. We often use the following
elementary inequalities. Let $\alpha = (\alpha_1, \ldots, \alpha_m) \in
\mathbbm{N}^m$ and let $a_1, \ldots, a_m \geqslant 0$. Let us denote $| \alpha
| \assign \alpha_1 + \cdots + \alpha_m$. Then
\begin{equation}
  \label{formulaclara17.2} \frac{a_1^{\alpha_1}}{\alpha_1 !} \cdots
  \frac{a_m^{\alpha_m}}{\alpha_m !} \leqslant e^{a_1 + \cdots + a_m}
\end{equation}
and, for every $\lambda = 1, 2, \ldots$
\begin{equation}
  \label{formulaclara17.4} \frac{a_1^{\alpha_1}}{\alpha_1 !} \cdots
  \frac{a_m^{\alpha_m}}{\alpha_m !} (a_1^{\lambda} + \cdots + a_m^{\lambda})
  \leqslant \lp{1} | \alpha | + m \rp{1} \lp{1} | \alpha | + \lambda
  \srp{1}{\lambda - 1}{} e^{a_1 + \cdots + a_m} .
\end{equation}
\begin{lemma}
  \label{formulaclara18}Let $\omega_1, \ldots, \omega_{\nu} \geqslant 0$ and
  let $\mu$ be as in Proposition \ref{formulaclara14}. Then, for every $m
  \geqslant 1$, $Q_m \in \mathcal{C}^1 \lp{1} \mathbbm{R}, \mathcal{A} \lp{0}
  \mathbbm{C}^{2 \nu} \rp{0} \rp{1}$. Moreover, for every $R > 0$ and every
  $(t, x, y) \in \mathbbm{R} \times D^2_{\mathbbm{R}^{\nu}, R}$,
  \begin{equation}
    \label{formulaclara19.a} |Q_m (t, x, x) | \leqslant \lp{2} 4
    \int_{\mathbbm{R}^{\nu}} e^{(1 + R) | \xi |} d^{\nu} | \mu | (\xi)
    \srp{2}{m}{},
  \end{equation}
  \begin{equation}
    \label{formulaclara19.b} |Q_m (t, x, y) | \leqslant \lp{2} 4
    \int_{\mathbbm{R}^{\nu}} e^{(1 + 2 R) | \xi |} d^{\nu} | \mu | (\xi)
    \srp{2}{m}{},
  \end{equation}
  \[ \ve{1} \partial_t \lp{1} t^{2 m - 1} Q_m (t, x, y) \rp{1} \ve{1}
     \leqslant |t|^{2 m - 2} \times \]
  \begin{equation}
    \label{formulaclara20} \lp{1} 1 + \omega_{\sharp} |t| \max (|x|, |y|)
    \rp{1} \lp{2} 4 \int_{\mathbbm{R}^{\nu}} e^{(1 + 2 R) | \xi |} d^{\nu} |
    \mu | (\xi) \srp{2}{m}{} .
  \end{equation}
\end{lemma}

\begin{proof}

  {\texspace{small}{\tmtextbf{-1-}}} The proof of (\ref{formulaclara19.a}),
  (\ref{formulaclara19.b}) and (\ref{formulaclara20}) (using
  (\ref{formulaclara21}) and (\ref{formulaclara22})) are straightforward if $m
  = 1$ and we now assume $m \geqslant 2$.

  {\texspace{med}{\tmtextbf{-2-}}} Let us check (\ref{formulaclara19.b}) and
  (\ref{formulaclara19.a}). By (\ref{formulaclara13}) and
  (\ref{formulaclara16.b})
  \[ |Q_m (t, x, y) | \leqslant \frac{2^{m - 1}}{4^{m - 1} m!} \sum_{g \in
     \mathcal{\mathcal{T}}_m} \int_{\mathbbm{R}^{\nu m}} | \xi_{j_1} ||
     \xi_{k_1} | \cdots | \xi_{j_{m - 1}} || \xi_{k_{m - 1}} |e^{2 R| \xi |_1}
     d^{\nu m} | \mu |^{\otimes} (\xi) \]
  where $[j_1, k_1], \ldots, [j_{m - 1}, k_{m - 1}]$ denote the $m - 1$ edges
  of the graph $g$. For $g \in \mathcal{\mathcal{T}}_m$, let us denote by
  $d_1, \ldots, d_m$ the degrees of the vertices $1, \ldots, m$. Then $d_1 +
  \cdots + d_m = 2 (m - 1)$ and, if $d_1, \ldots, d_m$ satisfy the previous
  equality, the number of trees with vertices $1, \ldots, m$ such that the
  $d^{\circ} (1) = d_1, \ldots, d^{\circ} (m) = d_m$ is equal to $\frac{(m -
  2) !}{(d_1 - 1) ! \cdots (d_m - 1) !}$ [Co]. Therefore
  \begin{eqnarray*}
    |Q_m (t, x, y) | \leqslant &  & \frac{2^{1 - m}}{m (m - 1)} \sum_{d_1 +
    \cdots + d_m = 2 m - 2}\\
    &  & \int_{\mathbbm{R}^{\nu m}} \frac{| \xi_1 |^{d_1}}{(d_1 - 1) !}
    \cdots \frac{| \xi_m |^{d_m}}{(d_m - 1) !} e^{2 R| \xi |_1} d^{\nu m} |
    \mu |^{\otimes} (\xi)\\
    \leqslant &  & \frac{2}{m (m - 1)} \sum_{d_1 + \cdots + d_m = 2 m - 2}\\
    &  & \int_{\mathbbm{R}^{\nu m}} \frac{| \xi_1 |^{d_1}}{d_1 !} \cdots
    \frac{| \xi_m |^{d_m}}{d_m !} e^{2 R| \xi |_1} d^{\nu m} | \mu |^{\otimes}
    (\xi),
  \end{eqnarray*}
  since, under the assumption $d_1 + \cdots + d_m = 2 m - 2$,
  \begin{equation}
    \label{formulaclara20.2} d_1 \times \cdots \times d_m \leqslant \lp{2}
    \frac{d_1 + \cdots + d_m}{m} \srp{2}{m}{} \leqslant 2^m .
  \end{equation}
  Then, by (\ref{formulaclara17.2}) and since
  \begin{equation}
    \label{formulaclara20.3} \ve{1} \lbc{1} (d_1, \ldots, d_m) \in
    (\mathbbm{N}^{\ast})^m |d_1 + \cdots + d_m = 2 m - 2 \rbc{1} \ve{1} =
    \left( \begin{array}{c}
      2 m - 3\\
      m - 1
    \end{array} \right) \leqslant 2^{2 m - 3},
  \end{equation}
  one gets
  \begin{equation}
    \label{formulaclara20.4} |Q_m | \leqslant \frac{2^{2 m - 2}}{m (m - 1)}
    \int_{\mathbbm{R}^{\nu m}} e^{(1 + 2 R) | \xi |_1} d^{\nu m} | \mu
    |^{\otimes} (\xi) \leqslant \frac{1}{4 m (m - 1)} (4 M)^m
  \end{equation}
  where
  \begin{equation}
    \label{formulaclara20.8} M \assign \int_{\mathbbm{R}^{\nu}} e^{(1 + 2 R) |
    \xi |} d^{\nu} | \mu | (\xi) .
  \end{equation}
  This proves (\ref{formulaclara19.b}). By (\ref{formulaclara16.a}), one gets
  (\ref{formulaclara19.a}).

  {\texspace{med}{\tmtextbf{-3-}}} For proving (\ref{formulaclara20}), we need
  the following inequalities. Let $s \in [0, 1]$ and let $A, B \in [0, 1]$
  such that $A < B$. Then, for every $\dot{\omega} \geqslant 0$ and every $t
  \in \mathbbm{R}$,
  \begin{equation}
    \label{formulaclara21} \ve{2} \frac{d}{d t} \lp{2} \frac{\sh (
    \dot{\omega} t s)}{\sh ( \dot{\omega} t)} \rp{2} \ve{2} \leqslant
    \frac{1}{2} \dot{\omega},
  \end{equation}
  \begin{equation}
    \label{formulaclara22} \ve{2} \frac{d}{d t} \lp{2} \frac{\sh (
    \dot{\omega} t A) \sh \lp{1} \dot{\omega} t (1 - B) \rp{1}}{\dot{\omega}
    \sh ( \dot{\omega} t)} \rp{2} \ve{2} \leqslant \frac{1}{2} .
  \end{equation}
  One has
  \begin{equation}
    \label{formulaclara22.4} \partial_t \lp{1} t^{2 m - 1} Q_m \rp{1} = m t^{m
    - 1} \times t^{m - 1} Q_m + R_m
  \end{equation}
  where
  \[ R_m \assign (- 2)^{m - 1} t^m \sum_{g \in \mathcal{\mathcal{T}}_m}
     \int_{0 < s_1 < \ldots < s_m < 1} \int_{\mathbbm{R}^{\nu m}} \]
  \[ \lp{1} \partial_t (t^{m - 1} \Upsilon^{g, \xi}_t) + i t^{m - 1}
     \Upsilon^{g, \xi}_t \partial_t (q_t^m (s) \cdot \xi) \rp{1} e^{i q_t^m
     (s) \cdot \xi} d^{\nu m} \mu^{\otimes} (\xi) d^m s. \]
  By (\ref{formulaclara13}) and (\ref{formulaclara22})
  \[ \text{$\ve{1} \partial_t (t^{m - 1} \Upsilon^{g, \xi}_t) \ve{1} \leqslant
     (m - 1) \times \lp{2} \frac{|t|}{4} \srp{2}{m - 2}{} \times \frac{1}{2}
     \times | \xi_{j_1} || \xi_{k_1} | \cdots | \xi_{j_{m - 1}} || \xi_{k_{m -
     1}} |$} . \]
  By (\ref{formulaclara13}) and (\ref{formulaclara21})
  \[ \ve{1} i \Upsilon^{g, \xi}_t \partial_t (q_t^m (s) \cdot \xi) \ve{1}
     \leqslant \lp{2} \frac{1}{4} \srp{2}{m - 1}{} | \xi_{j_1} || \xi_{k_1} |
     \cdots | \xi_{j_{m - 1}} || \xi_{k_{m - 1}} | \times \omega_{\sharp}
     \tilde{R} | \xi |_1 . \]
  where $\tilde{R} = \frac{1}{2} (|x| + |y|)$. Then
  \[ | \partial_t (t^{m - 1} \Upsilon^{g, \xi}_t) + i t^{m - 1} \Upsilon^{g,
     \xi}_t \partial_t (q_t^m (s) \cdot \xi) | \leqslant \]
  \[ \lp{2} \frac{|t|}{4} \srp{2}{m - 2}{} \lp{2} \frac{m - 1}{2} +
     \frac{\omega_{\sharp} |t| \tilde{R} | \xi |_1}{4} \rp{2} | \xi_{j_1} ||
     \xi_{k_1} | \cdots | \xi_{j_{m - 1}} || \xi_{k_{m - 1}} | \]
  and, by (\ref{formulaclara20.2})
  \begin{eqnarray*}
    |R_m | \leqslant &  & \frac{2^{2 m - 1} |t|^m}{m (m - 1)} \lp{2}
    \frac{|t|}{4} \srp{2}{m - 2}{} \sum_{d_1 + \cdots + d_m = 2 m - 2}
    \int_{\mathbbm{R}^{\nu m}} \frac{| \xi_1 |^{d_1}}{d_1 !} \cdots \frac{|
    \xi_m |^{d_m}}{d_m !} \times\\
    &  & \lp{2} \frac{m - 1}{2} + \frac{\omega_{\sharp} |t| \tilde{R} | \xi
    |_1}{4} \rp{2} e^{2 R| \xi |_1} d^{\nu m} | \mu |^{\otimes} (\xi)\\
    \leqslant &  & \frac{2^{2 m - 1} |t|^m}{m (m - 1)} \lp{2} \frac{|t|}{4}
    \srp{2}{m - 2}{} \left( \begin{array}{c}
      2 m - 3\\
      m - 1
    \end{array} \right) \lp{1} A 1 + A 2 \rp{1}
  \end{eqnarray*}
  where $A 1 \assign \frac{m - 1}{2} \times M^m$ by (\ref{formulaclara17.2})
  and $A 2 \assign \frac{\omega_{\sharp} |t|R}{4} \times (3 m - 2) \times M^m$
  by using (\ref{formulaclara17.4}). Then
  \[ |R_m | \leqslant \frac{(4 M)^m |t|^{2 m - 2}}{m} \lp{2} \frac{1}{2} +
     \omega_{\sharp} |t| \tilde{R} \rp{2} \]
  and by (\ref{formulaclara22.4}) and (\ref{formulaclara20.4})
  \[ | \partial_t \lp{1} t^{2 m - 1} Q_m \rp{1} | \leqslant (4 M)^m |t|^{2 m -
     2} (1 + \omega_{\sharp} |t| \tilde{R}) . \]
  This proves (\ref{formulaclara20}).
\end{proof}

\begin{definition}
  Let $\omega_1, \ldots, \omega_{\nu} \geqslant 0$ and let $\mu$ be as in
  Proposition \ref{formulaclara14}. For every $(t, x, y, h) \in \mathbbm{R}^+
  \times \mathbbm{C}^{2 \nu} \times \demiplan{}$, let us denote
  \[ w_1 (t, x, y, h) \assign - t \int_{0 < s_1 < 1} \int_{\mathbbm{R}^{\nu}}
     \int_0^1 \Omega_t . \xi \otimes_1 \xi e^{- h \vartheta \Omega_t . \xi
     \otimes_1 \xi} e^{i q_t (s) \cdot \xi_1} d \vartheta d \mu (\xi_1) d s \]
  and for $m \geqslant 2$
  \[ w_m (t, x, y, h) \assign 2^{m - 1} (- 1)^m t^{2 m - 1} \sum_{g \in
     \mathcal{\mathcal{T}}_m} \int_{0 < s_1 < \ldots < s_m < 1}
     \int_{\mathbbm{R}^{\nu m}} \int_{[0, 1]^g \times [0, 1]} \Upsilon^{g,
     \xi}_t \times \]
  \[ \Omega_t^{g, \theta} . \xi \otimes_m \xi e^{- h \vartheta \Omega_t^{g,
     \theta} . \xi \otimes_m \xi} e^{i q_t^m (s) \cdot \xi} d \vartheta d^{m -
     1} \theta d^{\nu m} \mu^{\otimes} (\xi) d^m s. \]
\end{definition}

By Lemma \ref{formulaclara10} and (\ref{formulaclara16.b}) the functions $w_m$
are well defined on $[0, + \infty [\times \mathbbm{C}^{2 \nu} \times
\demiplan{}$.

\begin{lemma}
  \label{formulaclara25}Let $\omega_1, \ldots, \omega_{\nu} \geqslant 0$ and
  let $\mu$ be as in Proposition \ref{formulaclara14}. Then, for every $m
  \geqslant 1$,
  \[ w_m \in \mathcal{C}^1 \lp{1} [0, + \infty [, \mathcal{A} \lp{0}
     \mathbbm{C}^{2 \nu} \times \demiplan{} \rp{0} \rp{1} . \]
  Moreover, for every $(t, x, y, h) \in [0, + \infty [\times
  D^2_{\mathbbm{R}^{\nu}, R} \times \demiplan{}$,
  \begin{equation}
    \label{formulaclara26} |w_m (t, x, y, h) | \leqslant m \lp{2} 4 t^2
    \int_{\mathbbm{R}^{\nu}} e^{(1 + 2 R) | \xi |} d^{\nu} | \mu | (\xi)
    \srp{2}{m}{} .
  \end{equation}
\end{lemma}

\begin{proof}
  The proof is similar to the one of Lemma \ref{formulaclara18} and we only
  focus on the differences between the two proofs. By (\ref{formulaclara12a}),
  (\ref{formulaclara26}) holds for $m = 1$. Let $m \geqslant 2$. By Lemma
  \ref{formulaclara10}
  \begin{eqnarray*}
    |w_m | \leqslant &  & \frac{2^{m - 1} t^{2 m - 1}}{4^{m - 1} m!} \sum_{g
    \in \mathcal{\mathcal{T}}_m} \int_{\mathbbm{R}^{\nu m}}\\
    &  & \frac{m t}{4} \lp{1} \xi_1^2 + \cdots + \xi_m^2 \rp{1} | \xi_{j_1}
    || \xi_{k_1} | \cdots | \xi_{j_{m - 1}} || \xi_{k_{m - 1}} |e^{2 R| \xi
    |_1} d^{\nu m} | \mu |^{\otimes} (\xi)\\
    \leqslant &  & \frac{t^{2 m}}{2 (m - 1)} \sum_{d_1 + \cdots + d_m = 2 m -
    2} \int_{\mathbbm{R}^{\nu m}}\\
    &  & \frac{| \xi_1 |^{d_1}}{d_1 !} \cdots \frac{| \xi_m |^{d_m}}{d_m !}
    \lp{1} \xi_1^2 + \cdots + \xi_m^2 \rp{1} e^{2 R| \xi |_1} d^{\nu m} | \mu
    |^{\otimes} (\xi) .
  \end{eqnarray*}
  By (\ref{formulaclara17.4})
  \begin{equation}
    \label{formulaclara27} \frac{| \xi_1 |^{d_1}}{d_1 !} \cdots \frac{| \xi_m
    |^{d_m}}{d_m !} \lp{1} \xi_1^2 + \cdots + \xi_m^2 \rp{1} \leqslant (3 m -
    2) (2 m) e^{| \xi |_1} \leqslant 6 m^2 e^{| \xi |_1} .
  \end{equation}
  Then
  \[ |w_m | \leqslant \frac{3 m^2}{m - 1} \left( \begin{array}{c}
       2 m - 3\\
       m - 1
     \end{array} \right) M^m t^{2 m} \leqslant m (4 t^2 M)^m \]
  where $M$ is given by (\ref{formulaclara20.8}). This proves
  (\ref{formulaclara26}). By dominated convergence theorem
  \[ w_m \in \mathcal{C}^1 \lp{1} [0, + \infty [, \mathcal{A} \lp{0}
     \mathbbm{C}^{2 \nu} \times \demiplan{} \rp{0} \rp{1} . \]
\end{proof}

\begin{lemma}
  \label{formulaclara30}Let $\omega_1, \ldots, \omega_{\nu} \geqslant 0$ and
  let $\mu$ be as in Proposition \ref{formulaclara14}. For $m \geqslant 1$,
  let us define
  \[ v_m^{\diamond} (t, x, y, h) \assign \frac{t^{2 m - 1}}{h} Q_m (t, x, y) +
     w_m (t, x, y, h) . \]
  Then
  \begin{equation}
    \label{formulaclara32} \lp{2} \text{$\partial_t -
    \frac{2}{p^{\tmop{harm}}} \partial_x p^{\tmop{harm}} \cdot \partial_x
    \rp{2} v^{\diamond}_m = \partial_x^2 v^{\diamond}_m$} +
    \sum_{\tmscript{\begin{array}{c}
      p + q = m\\
      p, q \geqslant 1
    \end{array}}} \partial_x v^{\diamond}_p \cdot \partial_x v^{\diamond}_q +
    c (x) \delta_{m = 1}
  \end{equation}
  and $v^{\diamond}_m |_{t = 0} = 0$.
\end{lemma}

\begin{proof}
  Let $E$ and $R_h$ be the moulds defined in subsection \ref{identities}. Let
  $f \in \mathcal{M}^{\tmop{pre}}$ be defined by
  \[ f^{\varnothing} = 0, \]
  \[ f (s_1, \ldots, s_m ; \xi_1, \ldots, \xi_m) = \lp{2} \frac{1}{h} E + R_h
     \srp{2}{\{1, \ldots, m\}}{} \]
  where, for $1 \leqslant j \leqslant k \leqslant m$,
  \[ z_{j, k} \assign - \sum_{\upsilon = 1}^{\nu} \frac{\sh (\omega_{\upsilon}
     t s_j) \sh \lp{1} \omega_{\upsilon} t (1 - s_k) \rp{1}}{\omega_{\upsilon}
     \sh (\omega_{\upsilon} t)} \xi_{j, \upsilon} \xi_{k, \upsilon} . \]
  Let $\lambda$ be the Borel measure defined on $[0, 1] \times
  \mathbbm{R}^{\nu}$ by
  \[ d \lambda (s, \xi) = t e^{i q_t (s) \cdot \xi} d^{\nu} \mu (\xi) d s. \]
  Then
  \[ \frac{t^{2 m - 1}}{h} Q_m + w_m = \int_{0 < s_1 < \ldots < s_m < 1}
     \int_{\mathbbm{R}^{\nu m}} f (s_1, \ldots, s_m ; \xi_1, \ldots, \xi_m)
     d^{\nu m} \lambda^{\otimes} (s, \xi) . \]
  By Corollary \ref{treemaria14}
  \[ \exp_{\tmop{pre}} (f)^{}_{} (s_1, \ldots, s_m ; \xi_1, \ldots, \xi_m) =
     \frac{1}{h^m} \exp (- h \Omega_t . \xi \otimes_m \xi) . \]
  Let $\tilde{v}$ be the formal series with respect to $H$ defined by
  \[ \tilde{v} \assign \exp \lp{2} \sum^{+ \infty}_{m = 1} v^{\diamond}_m H^m
     \rp{2} . \]
  Then
  \[ \tilde{v} = \exp \lp{1} \Phi (f) \rp{1} = \Phi (\exp_{\tmop{pre}} f) =
     \sum_{m = 0}^{+ \infty} v_m H^m . \]
  The second equality holds since $\Phi$ is a morphism (Lemma
  \ref{mouldofamanda2}) and the third one uses the definition
  (\ref{formulaclara14.2}) of $v_m$. Then (\ref{formulaclara14.3}) implies
  (\ref{formulaclara32}).
\end{proof}

\begin{proposition}
  \label{formulaclara34}Let $\omega_1, \ldots, \omega_{\nu}, T > 0$. Let $\mu$
  be as in Proposition \ref{formulaclara14} and let us assume that
  \[ 4 T^2 M_{\mu} < 1. \]
  Let
  \[ \varphi (t, x, y) \assign \sum^{+ \infty}_{m = 1} t^{2 m - 2} Q_m (t, x,
     y), \]
  \[ w (t, x, y, h) \assign \sum^{+ \infty}_{m = 1} w_m (t, x, y, h), \]
  and
  \[ \phi \assign \frac{1}{4} \lp{2} \sum_{\upsilon = 1}^{\nu}
     \frac{\omega_{\upsilon}}{\tmop{sh} (\omega_{\upsilon} t)} \lp{1}
     \tmop{ch} (\omega_{\upsilon} t) (x_{\upsilon}^2 + y_{\upsilon}^2) - 2
     x_{\upsilon} y_{\upsilon} \rp{1} \rp{2} - t \varphi . \]
  Then
  \[ \varphi \in \mathcal{C}^1 \lp{1}] - T, T [, \mathcal{A} \lp{0}
     D^2_{\mathbbm{R}^{\nu}, 2} \rp{0} \rp{1} \text{ and \ } w \in
     \mathcal{C}^1 \lp{1} [0, T [, \mathcal{A} \lp{0} D^2_{\mathbbm{R}^{\nu},
     2} \times \demiplan{} \rp{0} \rp{1} . \]
  Moreover
  \begin{equation}
    \label{formulaclara36} u (t, x, y, h) \assign \left( 4 \pi h \right)^{-
    \nu / 2} \prod_{\upsilon = 1}^{\nu} \frac{\omega_{\upsilon}^{1 /
    2}}{\tmop{sh}^{1 / 2} (\omega_{\upsilon} t)} \times e^{- \phi / h + w}
  \end{equation}
  satisfies (\ref{theorembetty6}).
\end{proposition}

\begin{proof}
  By dominated convergence theorem, by choosing $R = 2$ in Lemma
  \ref{formulaclara18} and Lemma \ref{formulaclara25},
  \[ \varphi \in \mathcal{C}^1 \lp{1}] - T, T [, \mathcal{A} \lp{0}
     D^2_{\mathbbm{R}^{\nu}, 2} \rp{0} \rp{1} \text{ and \ } w \in
     \mathcal{C}^0 \lp{1} [0, T [, \mathcal{A} \lp{0} D^2_{\mathbbm{R}^{\nu},
     2} \times \demiplan{} \rp{0} \rp{1} . \]
  The quantity $| \partial_t w_m |$ is bounded by
  \[ P (m, |h|, t, |x|, |y|) t^{2 m - 1} \lp{2} 4 \int_{\mathbbm{R}^{\nu}}
     e^{5| \xi |} d^{\nu} | \mu | (\xi) \srp{2}{m}{}, \]
  where $P$ is a polynomial with respect to its arguments (this step is left
  to the reader). This implies the $\mathcal{C}^1$-regularity of $w$.
  
  By Lemma \ref{formulaclara30}, $v^{\diamond}_m$ satisfies
  (\ref{formulaclara32}). Therefore
  \[ v^{\diamond} \assign \frac{t}{h} \varphi + w = \sum_{m \geqslant 1}
     v^{\diamond}_m \]
  satisfies
  \[ \lp{2} \partial_t - \frac{2}{p^{\tmop{harm}}} \partial_x p^{\tmop{harm}}
     \cdot \partial_x \rp{2} v^{\diamond} = \partial_x^2 v^{\diamond} +
     \partial_x v^{\diamond} \cdot \partial_x v^{\diamond} + c (x) . \]
  Then
  \[ \lp{2} \partial_t - \frac{2}{p^{\tmop{harm}}} \partial_x p^{\tmop{harm}}
     \cdot \partial_x \rp{2} e^{v^{\diamond}} = \partial_x^2 e^{v^{\diamond}}
     + c (x) e^{v^{\diamond}} \]
  and the function $u$ satisfies (\ref{theorembetty6}).
\end{proof}

\begin{remark}
  \label{formulaclara40}Let $y \in D_{\mathbbm{R}^{\nu}, 2}$. Then the
  function $(t, x) \longmapsto \phi (t, x, y)$ satisfies
  \begin{equation}
    \label{formulaclara41} \left\{ \begin{array}{l}
      \partial_t \phi + (\partial_x \phi)^2 = \frac{( \tmmathbf{\omega}
      x)^2}{4} - c (x)\\
      \\
      \lp{2} \phi - \frac{(x - y)^2}{4 t} \rp{2} \sve{2}{}{t = 0} = 0
    \end{array} \right.
  \end{equation}
  for $(t, x) \in \lp{1}] - T, T [-\{0\} \rp{1} \times D_{\mathbbm{R}^{\nu},
  2}$.
\end{remark}

\section{Borel summation preliminary statements}

The following lemma will be useful for the proof of a Morse Lemma concerning
the function $\phi |_{y = x}$ (see Proposition \ref{formulaclara34}).

\begin{lemma}
  Let $\mu$ be as in Definition \ref{theorembetty1} and Proposition
  \ref{formulaclara14}. Let us assume that $M_{\mu}' < \infty$ and $4 T^2
  M_{\mu} < 1$. Then, for every $(t, x) \in] 0, T [\times
  D_{\mathbbm{R}^{\nu}, 4}$,
  \begin{equation}
    \label{preparationsandra1.2} | \varphi (t, x, x) | \leqslant \frac{4
    M_{\mu}}{1 - 4 T^2 M_{\mu}},
  \end{equation}
  \begin{equation}
    \label{preparationsandra1.4} \max_{1 \leqslant \delta \leqslant \nu}
    |x_{\delta} \varphi (t, x, x) | \leqslant \ch \lp{1} \frac{\omega_{\sharp}
    T}{2} \rp{1} \times \frac{M_{\mu}'}{1 - 4 T^2 M_{\mu}}
  \end{equation}
\end{lemma}

\begin{proof}
  Let $(t, x) \in] 0, T [\times D_{\mathbbm{R}^{\nu}, 4}$.

  {\texspace{med}{\tmtextbf{-1-}}} Let us check (\ref{preparationsandra1.2}).
  Since
  \[ \varphi (t, x, x) = \sum^{+ \infty}_{m = 1} t^{2 m - 2} Q_m (t, x, x) \]
  and by choosing $R = 4$ in (\ref{formulaclara19.a}),
  \[ | \varphi (t, x, x) | \leqslant \frac{4 M_{\mu}}{1 - 4 T^2 M_{\mu}} . \]
  This proves (\ref{preparationsandra1.2}).

  {\texspace{med}{\tmtextbf{-2-}}} Let us check (\ref{preparationsandra1.4}).
  By Definition \ref{formulaclara15}
  \[ x_{\delta} Q_1 (t, x, x) = \int_{0 < s_1 < 1} \int_{\mathbbm{R}^{\nu}}
     x_{\delta} e^{i x \cdot \lp{1} \tmmathbf{\varpi}_t (s_1) \xi_1 \rp{1}} F
     (s_1, \xi_1) d \xi_1 d s_1 \]
  where $F (s_1, \xi_1) = \rho_{\mu} (\xi_1)$. If $m \geqslant 2$,
  \[ x_{\delta} Q_m (t, x, x) = (- 2)^{m - 1} \sum_{g \in
     \mathcal{\mathcal{T}}_m} \int_{0 < s_1 < \ldots < s_m < 1}
     \int_{\mathbbm{R}^{\nu m}} \]
  \[ x_{\delta} F (s, \xi) e^{i x \cdot \lp{1} \tmmathbf{\varpi}_t (s_1) \xi_1
     + \cdots + \tmmathbf{\varpi}_t (s_m) \xi_m \rp{1}} d^{\nu m} \xi d^m s \]
  where
  \[ F (s, \xi) \assign \Upsilon^{g, \xi}_t \rho_{\mu} (\xi_1) \cdots
     \rho_{\mu} (\xi_m) . \]
  Let
  \[ D \assign \frac{1}{m} \sum_{p = 1}^m \frac{1}{\varpi_{t, \delta} (s_p)}
     \partial_{\xi_{p, \delta}} . \]
  By integration by parts,
  \[ x_{\delta} \int_{\mathbbm{R}^m} F \times e^{i x_{\delta} \lp{1}
     \varpi_{t, \delta} (s_1) \xi_{1, \delta} + \cdots + \varpi_{t, \delta}
     (s_m) \xi_{m, \delta} \rp{1}} d \xi_{1, \delta} \cdots d \xi_{m, \delta}
     = \]
  \[ i \int_{\mathbbm{R}^m} D F \times e^{i x_{\delta} \lp{1} \varpi_{t,
     \delta} (s_1) \xi_{1, \delta} + \cdots + \varpi_{t, \delta} (s_m) \xi_{m,
     \delta} \rp{1}} d \xi_{1, \delta} \cdots d \xi_{m, \delta} . \]
  Therefore
  \[ x_{\delta} Q_1 (t, x, x) = i \int_{0 < s_1 < 1} \int_{\mathbbm{R}^{\nu}}
     e^{i x \cdot \lp{1} \tmmathbf{\varpi}_t (s_1) \xi_1 \rp{1}}
     (\partial_{\xi_{1, \delta}} \rho_{\mu}) (\xi_1) d \xi_1 \frac{d
     s_1}{\varpi_{t, \delta} (s_1)} \]
  and, if $m \geqslant 2$, denoting by $F_g$ the quantity $F$ in order to
  emphasize its dependence on the tree $g$,
  \[ x_{\delta} Q_m (t, x, x) = i (- 2)^{m - 1} \sum_{g \in
     \mathcal{\mathcal{T}}_m} \int_{0 < s_1 < \ldots < s_m < 1}
     \int_{\mathbbm{R}^{\nu m}} \]
  \[ D F_g \times e^{i x \cdot \lp{1} \tmmathbf{\varpi}_t (s_1) \xi_1 + \cdots
     + \tmmathbf{\varpi}_t (s_m) \xi_m \rp{1}} d^{\nu m} \xi d^m s. \]
  By (\ref{formulaclara15.5}) and (\ref{formulaclara16.a})
  \[ |x_{\delta} Q_1 (t, x, x) | \leqslant \ch \lp{1} \frac{\omega_{\sharp}
     t}{2} \rp{1} M'_{\mu} . \]
  Let us assume that $m \geqslant 2$. Using (\ref{formulaclara13}) and
  (\ref{formulaclara15.5}), one can show
  \[ \ve{1} \Upsilon^{g, \xi}_t \ve{1} \leqslant \frac{1}{4^{m - 1}} |
     \xi_{j_1} || \xi_{k_1} | \cdots | \xi_{j_{m - 1}} || \xi_{k_{m - 1}} |,
  \]
  \[ \ve{1} D \Upsilon^{g, \xi}_t \ve{1} \leqslant \frac{\ch \lp{1}
     \frac{\omega_{\sharp} t}{2} \rp{1}}{4^{m - 1} m} \lp{1} \partial_{u_1} +
     \cdots + \partial_{u_m} \rp{1} (u_{j_1} u_{k_1} \cdots u_{j_{m - 1}}
     u_{k_{m - 1}}) \sve{1}{}{u_1 = | \xi_1 |, \ldots, u_m = | \xi_m |} \]
  where $(j_1, k_1), \ldots, (j_{m - 1}, k_{m - 1})$ denote the $m - 1$ edges
  of the graph $g$ ($j_p < k_p$). Let $d_1, \ldots, d_m$ be the degrees of the
  vertices $1, \ldots, m$ of the graph $g$. Then
  \[ \frac{|D F_g |}{d_1 ! \cdots d_m !} \leqslant \frac{1}{4^{m - 1}} (A 1 +
     A 2) \]
  where
  \[ A 1 \assign \frac{\ch \lp{1} \frac{\omega_{\sharp} t}{2} \rp{1}}{m}
     \frac{| \xi_1 |^{d_1}}{d_1 !} \cdots \frac{| \xi_m |^{d_m}}{d_m !} \times
  \]
  \[ \lp{1} | \partial_{\xi_{1, \delta}} \rho_{\mu} (\xi_1) || \rho_{\mu}
     (\xi_2) | \cdots | \rho_{\mu} (\xi_m) | + \cdots + | \rho_{\mu} (\xi_1) |
     \cdots | \rho_{\mu} (\xi_{m - 1}) || \partial_{\xi_{m, \delta}}
     \rho_{\mu} (\xi_m) | \rp{1}, \]
  \[ A 2 \assign \frac{\ch \lp{1} \frac{\omega_{\sharp} t}{2} \rp{1}}{m}
     \lp{1} \partial_{u_1} + \cdots + \partial_{u_m} \rp{1} \lp{2}
     \frac{u_1^{d_1}}{d_1 !} \cdots \frac{u_m^{d_m}}{d_m !} \rp{2}
     \sve{2}{}{u_1 = | \xi_1 |, \ldots, u_m = | \xi_m |} \times \]
  \[ | \rho_{\mu} (\xi_1) | \cdots | \rho_{\mu} (\xi_m) |. \]
  Using (\ref{formulaclara17.2}), one gets
  \[ A 1 \leqslant \frac{1}{m} \ch \lp{1} \frac{\omega_{\sharp} t}{2} \rp{1}
     e^{| \xi |_1} \times \]
  \[ \lp{1} | \partial_{\xi_{1, \delta}} \rho_{\mu} (\xi_1) || \rho_{\mu}
     (\xi_2) | \cdots | \rho_{\mu} (\xi_m) | + \cdots + | \rho_{\mu} (\xi_1) |
     \cdots | \rho_{\mu} (\xi_{m - 1}) || \partial_{\xi_{m, \delta}}
     \rho_{\mu} (\xi_m) | \rp{1}, \]
  \[ A 2 \leqslant \ch \lp{1} \frac{\omega_{\sharp} t}{2} \rp{1} e^{| \xi |_1}
     | \rho_{\mu} (\xi_1) | \cdots | \rho_{\mu} (\xi_m) | \]
  since
  \[ \lp{1} \partial_{u_1} + \cdots + \partial_{u_m} \rp{1} \lp{2}
     \frac{u_k^{d_k}}{d_k !} \rp{2} = \frac{u_k^{d_k - 1}}{(d_k - 1) !} . \]
  Then
  \[ \frac{|D F_g |}{d_1 ! \cdots d_m !} \leqslant \frac{2}{4^{m - 1}} \ch
     \lp{1} \frac{\omega_{\sharp} t}{2} \rp{1} e^{| \xi |_1} \times \]
  \begin{equation}
    \label{preparationclara2} \frac{1}{m} \lp{1} | \rho_{\mu} (\xi_1) |_{\ast}
    | \rho_{\mu} (\xi_2) | \cdots | \rho_{\mu} (\xi_m) | + \cdots + |
    \rho_{\mu} (\xi_1) | \cdots | \rho_{\mu} (\xi_{m - 1}) || \rho_{\mu}
    (\xi_m) |_{\ast} \rp{1}
  \end{equation}
  where
  \[ | \rho_{\mu} (\eta) |_{\ast} \assign \max \lp{1} | \rho_{\mu} (\eta) |,
     | \partial_{\eta_1} \rho_{\mu} (\eta) |, \ldots, | \partial_{\eta_{\nu}}
     \rho_{\mu} (\eta) | \rp{1} \]
  for $\eta \in \mathbbm{R}^{\nu}$. Then, by (\ref{formulaclara16.a}) and
  (\ref{preparationclara2}),
  \begin{eqnarray*}
    |x_{\delta} Q_m (t, x, x) | \leqslant &  & 2^{m - 1} \sum_{g \in
    \mathcal{\mathcal{T}}_m} \int_{0 < s_1 < \ldots < s_m < 1}
    \int_{\mathbbm{R}^{\nu m}} |D F_g | \times e^{4| \xi |_1} d^{\nu m} \xi
    d^m s\\
    \leqslant &  & \frac{2^{m - 1}}{m!} \times \lp{2} \sum_{g \in
    \mathcal{\mathcal{T}}_m} d_1 ! \cdots d_m ! \rp{2} \times \frac{2}{4^{m -
    1}} \ch \lp{1} \frac{\omega_{\sharp} t}{2} \rp{1} M_{\mu}' M_{\mu}^{m - 1}
    .
  \end{eqnarray*}
  But, by (\ref{formulaclara20.2}) and (\ref{formulaclara20.3})
  \begin{eqnarray*}
    \sum_{g \in \mathcal{\mathcal{T}}_m} d_1 ! \cdots d_m ! = &  & \sum_{d_1 +
    \cdots + d_m = 2 m - 2} \frac{(m - 2) !}{(d_1 - 1) ! \cdots (d_m - 1) !}
    \times d_1 ! \cdots d_m !\\
    \leqslant &  & (m - 2) ! \times 2^m \times 2^{2 m - 3} .
  \end{eqnarray*}
  Then
  \[ |x_{\delta} Q_m (t, x, x) | \leqslant 2 \frac{4^{m - 1}}{m (m - 1)} \ch
     \lp{1} \frac{\omega_{\sharp} T}{2} \rp{1} M_{\mu}' M_{\mu}^{m - 1} . \]
  Then, for every $m \geqslant 1$,
  \[ |x_{\delta} Q_m (t, x, x) | \leqslant \ch \lp{1} \frac{\omega_{\sharp}
     T}{2} \rp{1} M_{\mu}' (4 M_{\mu})^{m - 1} . \]
  This proves (\ref{preparationsandra1.4}).
\end{proof}

The following Morse lemma gives a convenient expression for $\phi \sve{1}{}{y
= x}$ (see Proposition \ref{formulaclara34} for the definition of the function
$\phi$). We denote by $D_{\omega}$ be the following $\nu \times \nu$ diagonal
matrix
\begin{equation}
  \label{preparationclara3.14} D_{\omega} \assign \frac{1}{\sqrt{2}}
  \left(\begin{array}{ccc}
    \omega_1^{1 / 2} \tgh^{1 / 2} \lp{2} \frac{\omega_1 t}{2} \rp{2} &  & 0\\
    \vdots & \ddots & \vdots\\
    0 & \ldots & \omega_{\nu}^{1 / 2} \tgh^{1 / 2} \lp{2} \frac{\omega_{\nu}
    t}{2} \rp{2}
  \end{array}\right) .
\end{equation}
Moreover, if $W$ is a $\mathbbm{C}^{\nu}$-valued function analytic function
defined on some open set of $\mathbbm{C}^{\nu}$, we denote by $\partial_x W
(x)$ the $\nu \times \nu$ matrix $\lp{1} \partial_{x_1} W (x), \ldots,
\partial_{x_{\nu}} W (x) \rp{1}$.

\begin{proposition}
  \label{preparationclara4} Let $\alpha_{\nu} > 0$ such that $\alpha =
  \alpha_{\nu}$ satisfies (\ref{preparationclara10.5}) and
  (\ref{preparationclara16}) below (the number $\alpha_{\nu}$ only depends on
  $\nu$). Let $\mu$ be as in Definition \ref{theorembetty1} and Proposition
  \ref{formulaclara14}. Let us assume that $M_{\mu}' < \infty$, that
  (\ref{theorembetty3}), (\ref{theorembetty4}) are satisfied and
  \[ \int_{\mathbbm{R}^{\mu}} d \mu (\xi) = 0. \]
  Let $\omega_1, \ldots, \omega_{\nu}, T > 0$ such that
  \begin{equation}
    \label{preparationclara4.2a} 4 T^2 M_{\mu} < 1,
  \end{equation}
  \[ \ch \lp{1} \frac{\omega_{\sharp} T}{2} \rp{1} \frac{\omega_{\sharp} (1 +
     \omega_{\flat} T) M'_{\mu}}{\omega_{\flat}^3 (1 - 4 T^2 M_{\mu})} <
     \alpha_{\nu} . \]
  Then there exists $\Lambda \in \mathcal{C}^0 \lp{1}] 0, T [,
  \mathcal{A}_{\mathbbm{C}^{\nu}} (D_{\mathbbm{R}^{\nu}, 1}) \rp{1}$ such that
  \begin{enumerate}
    \item \label{preparationclara4.5}$\Lambda |_{] 0, T [\times
    \mathbbm{R}^{\nu}}$ is $\mathbbm{R}^{\nu}$-valued and
    \begin{equation}
      \label{preparationclara4.7} \sup_{(t, x) \in] 0, T [\times
      D_{\mathbbm{R}^{\nu}, 1}} | \partial_x \Lambda (t, x) | \leqslant
      \frac{1}{2},
    \end{equation}
    \item \label{preparationclara5}for every $(t, x) \in] 0, T [\times
    D_{\mathbbm{R}^{\nu}, 1}$, $\phi \sve{1}{}{y = x} = \theta^2 (t, x)$ where
    \begin{equation}
      \label{preparationclara5.1} \theta (t, x) \assign D_{\omega} \lp{1} x +
      \Lambda (t, x) \rp{1} .
    \end{equation}
  \end{enumerate}
\end{proposition}

\begin{proof}
  We assume that (\ref{preparationclara4.2a}) holds and that, for some
  arbitrary $\alpha$,
  \begin{equation}
    \label{preparationclara5.31} \ch \lp{1} \frac{\omega_{\sharp} T}{2} \rp{1}
    \frac{\omega_{\sharp} (1 + \omega_{\flat} T) M'_{\mu}}{\omega_{\flat}^3 (1
    - 4 T^2 M_{\mu})} < \alpha .
  \end{equation}
  Let $(t, x) \in] 0, T [\times D_{\mathbbm{R}^{\nu}, 1}$.

  {\texspace{med}{\tmtextbf{-1-}}} First, we look for a condition on $\alpha$
  providing the existence of the function $\Lambda$ such that
  (\ref{preparationclara5.1}) holds. Let us denote $\Phi (t, x) \assign \phi
  \sve{1}{}{y = x}$ (the function $\Phi$ is well defined by Proposition
  \ref{formulaclara34}). We claim that
  \begin{equation}
    \label{preparationclara6} \Phi (t, 0) = \partial_{x_1} \Phi (t, 0) =
    \cdots = \partial_{x_{\nu}} \Phi (t, 0) = 0 \text{.}
  \end{equation}
  This fact can be proved either by checking that, for $m \geqslant 1$,
  \[ Q_m (t, 0, 0) = \partial_{x_1} Q_m (t, 0, 0) = \cdots =
     \partial_{x_{\nu}} Q_m (t, 0, 0) = 0 \]
  (notice that each tree of $\mathcal{\mathcal{T}}_m$ has at least two
  vertices of degree $1$ if $m \geqslant 2$) or by using
  \[ c (0) = \partial_{x_1} c (0) = \cdots = \partial_{x_{\nu}} c (0) = 0 \]
  and that $\phi$ satisfies (\ref{formulaclara41}).
  
  Since $\frac{\dot{\omega}}{\tmop{sh} ( \dot{\omega} t)} \lp{1} \tmop{ch} (
  \dot{\omega} t) - 1 \rp{1} = \dot{\omega} \tgh \lp{2} \frac{\dot{\omega}
  t}{2} \rp{2}$ for $\dot{\omega} \in \mathbbm{R}$ and by the Taylor formula,
  \[ \Phi (t, x) = x \cdot \lp{1} B (t, x) \cdot x \rp{1} \]
  where
  \begin{eqnarray*}
    B (t, x) \assign &  & \int_0^1 (1 - u) (\partial_z \otimes \partial_z)
    \Phi (t, z) |_{z = u x} d u\\
    = &  & D^2_{\omega} - t \int_0^1 (1 - u) (\partial_z \otimes \partial_z)
    \varphi (t, z, z) |_{z = u x} d u\\
    = &  & D_{\omega} \lp{1} \mathbbm{1} + R (t, x) \rp{1} D_{\omega}
  \end{eqnarray*}
  and
  \begin{equation}
    \label{preparationclara7} R (t, x) \assign - t D_{\omega}^{- 1} \lp{2}
    \int_0^1 (1 - u) (\partial_z \otimes \partial_z) \varphi (t, z, z) |_{z =
    u x} d u \rp{2} D_{\omega}^{- 1} .
  \end{equation}
  Let $A > 0$, let $\beta = 1, \ldots, \nu$ and let $\psi$ be an analytic
  function on $D_{\mathbbm{R}^{\nu}, A + 1}$. By the Cauchy formula
  \begin{equation}
    \label{preparationclara7.1} \underset{z \in D_{\mathbbm{R}^{\nu},
    A}}{\sup} | \partial_{z_{\beta}} \psi (z) | \leqslant \underset{z \in
    D_{\mathbbm{R}^{\nu}, A + 1}}{\sup} | \psi (z) |.
  \end{equation}
  Let $z \in D_{\mathbbm{R}^{\nu}, 1}$. Let $\beta, \gamma = 1, \ldots, \nu$.
  Then, by (\ref{preparationsandra1.2}) and (\ref{preparationclara7.1}),
  \[ | \partial_{z_{\beta}} \partial_{z_{\gamma}} \varphi (t, z, z) |
     \leqslant \frac{4 M_{\mu}}{1 - 4 T^2 M_{\mu}} . \]
  Then
  \[ \ve{1} (\partial_z \otimes \partial_z) \varphi (t, z, z) |_{z = u x}
     \sve{1}{}{\infty} \leqslant \frac{4 M_{\mu}}{1 - 4 T^2 M_{\mu}} . \]
  Since the matrix $D_{\omega}^{- 1}$ is diagonal,
  \begin{eqnarray*}
    \ve{1} R (t, x) \sve{1}{}{\infty} \leqslant &  & t \ve{1} D_{\omega}^{- 1}
    \sve{1}{2}{\infty} \times \frac{2 M_{\mu}}{1 - 4 T^2 M_{\mu}}\\
    \leqslant &  & \frac{t}{\omega_{\flat} \tgh \lp{2} \frac{\omega_{\flat}
    t}{2} \rp{2}} \times \frac{4 M_{\mu}}{1 - 4 T^2 M_{\mu}} .
  \end{eqnarray*}
  For every $\theta \in [0, + \infty [$,
  \begin{equation}
    \label{preparationclara7.5} \frac{\theta}{\tgh \theta} \leqslant 1 + 2
    \theta .
  \end{equation}
  Then
  \begin{equation}
    \label{preparationclara8} \ve{1} R (t, x) \sve{1}{}{\infty} \leqslant
    \frac{8 M_{\mu} (1 + \omega_{\flat} T)}{\omega_{\flat}^2 (1 - 4 T^2
    M_{\mu})}
  \end{equation}
  and by (\ref{preparationclara5.31})
  \begin{equation}
    \label{preparationclara8.1} \ve{1} R (t, x) \sve{1}{}{\infty} < 8 \alpha .
  \end{equation}
  For $\mathbbm{K}=\mathbbm{R}$ or $\mathbbm{C}$, let us denote by \
  $\mathcal{M}^{\mathbbm{K}}_{\tmop{sym}}$ the space of $\nu \times \nu$ \
  symmetric matrices with entries in $\mathbbm{K}$. Since the map
  $\mathcal{M}^{\mathbbm{R}}_{\tmop{sym}} \longrightarrow
  \mathcal{M}^{\mathbbm{R}}_{\tmop{sym}}, C \longmapsto C^2$ is a real
  analytic local diffeomorphism near $C = \mathbbm{1}$, one can find \
  $\mathcal{U} \subset \mathcal{M}_{\tmop{sym}}^{\mathbbm{R}}$ a neighbourhood
  of $\mathbbm{1}$ and a real analytic local diffeomorphism
  \[ \Xi : \mathcal{U} \longrightarrow \mathcal{M}^{\mathbbm{R}}_{\tmop{sym}}
  \]
  such that, denoting $S^{1 / 2} \assign \Xi (S)$, \ $\mathbbm{1}^{1 / 2} =
  \mathbbm{1}$ and
  \begin{equation}
    \label{preparationclara8.5} \text{}^{} S^{1 / 2} S^{1 / 2} = S
  \end{equation}
  for $S \in \mathcal{U}$. Let $B^{\mathbbm{C}}_{\infty} ( \mathbbm{1}, \rho)
  \subset \mathcal{M}^{\mathbbm{C}}_{\tmop{sym}}$ be the open ball of center
  $\mathbbm{1}$ and radius $\rho$ (with respect to the norm $| \cdot
  |_{\infty}$). By analytic continuation, there exists $\rho_{\sqrt{}} > 0$
  such that $\Xi$ is analytic and (\ref{preparationclara8.5}) is satisfied on
  $B^{\mathbbm{C}}_{\infty} ( \mathbbm{1}, \rho_{\sqrt{}})$. Let $\partial_C
  \Xi$ be the tangent map associated to $\Xi$ at the point $C$. Since
  $\partial_{\mathbbm{1}} \Xi \cdot H = \frac{1}{2} H$, we can choose
  $\rho_{\sqrt{}}$ small enough such that for every $S \in
  B^{\mathbbm{C}}_{\infty} ( \mathbbm{1}, \rho_{\sqrt{}})$ and every $H \in
  \mathcal{M}^{\mathbbm{C}}_{\tmop{sym}}$,
  \begin{equation}
    \label{preparationclara9} |S^{1 / 2} - \mathbbm{1} | \leqslant \nu |S -
    \mathbbm{1} |_{\infty},
  \end{equation}
  \begin{equation}
    \label{preparationclara10} | \partial_S \Xi \cdot H| \leqslant \nu
    |H|_{\infty} .
  \end{equation}
  Let us assume
  \begin{equation}
    \label{preparationclara10.5} 8 \alpha < \rho_{\sqrt{}} .
  \end{equation}
  By (\ref{preparationclara7}) and Proposition \ref{formulaclara34}, \ $R \in
  \mathcal{C}^0 \lp{1}] 0, T [,
  \mathcal{A}_{\mathcal{M}^{\mathbbm{C}}_{\tmop{sym}}} \lp{0}
  D_{\mathbbm{R}^{\nu}, 1} \rp{0} \rp{1}$ and by (\ref{preparationclara8.1})
  \[ |R (t, x) |_{\infty} < \rho_{\sqrt{}} . \]
  Then
  \[ B (t, x) = D_{\omega} \lp{1} \mathbbm{1} + R (t, x) \srp{1}{1 / 2}{}
     \lp{1} \mathbbm{1} + R (t, x) \srp{1}{1 / 2}{} D_{\omega} . \]
  Let
  \begin{equation}
    \label{preparationclara11} \theta (t, x) \assign \lp{1} \mathbbm{1} + R
    (t, x) \srp{1}{1 / 2}{} D_{\omega} x.
  \end{equation}
  Then
  \[ \phi \sve{1}{}{y = x} = x \cdot \lp{1} B (t, x) x \rp{1} = \theta^2 (t,
     x) . \]
  Let
  \[ \Lambda (t, x) \assign D_{\omega}^{- 1} \lp{1} \lp{1} \mathbbm{1} + R (t,
     x) \srp{1}{1 / 2}{} - \mathbbm{1} \rp{1} D_{\omega} x. \]
  Then (\ref{preparationclara5.1}) holds. Notice also that $\Lambda \in
  \mathcal{C}^0 \lp{1}] 0, T [, \mathcal{A}_{\mathbbm{C}^{\nu}}
  (D_{\mathbbm{R}^{\nu}, 1}) \rp{1}$.

  {\texspace{med}{\tmtextbf{-2-}}} We want to prove that
  (\ref{preparationclara4.7}) holds. We shall need an additional condition on
  $\alpha$. One has
  \[ \partial_{x_{\gamma}} \Lambda (t, x) = D_{\omega}^{- 1} \lp{1} \lp{1}
     \mathbbm{1} + R (t, x) \srp{1}{1 / 2}{} - \mathbbm{1} \rp{1} D_{\omega}
     e_{\gamma} + D_{\omega}^{- 1} \lp{2} \partial_{x_{\gamma}} \lp{1}
     \mathbbm{1} + R (t, x) \srp{1}{1 / 2}{} \rp{2} D_{\omega} x. \]
  By (\ref{preparationclara9}) and (\ref{preparationclara8})
  \begin{eqnarray*}
    \ve{1} D_{\omega}^{- 1} \lp{1} \lp{1} \mathbbm{1} + R (t, x) \srp{1}{1 /
    2}{} - \mathbbm{1} \rp{1} D_{\omega} e_{\gamma} \ve{1} \leqslant &  & \nu
    \ve{1} D_{\omega}^{- 1} \ve{1} \ve{1} R (t, x) \sve{1}{}{\infty} \ve{1}
    D_{\omega} \ve{1} |e_{\gamma} |\\
    \leqslant &  & \nu \frac{8 M_{\mu} (1 + \omega_{\flat}
    T)}{\omega_{\flat}^2 (1 - 4 T^2 M_{\mu})} \times
    \frac{\omega_{\sharp}}{\omega_{\flat}}
  \end{eqnarray*}
  since
  \[ \ve{1} D_{\omega}^{- 1} \ve{1} \ve{1} D_{\omega} \ve{1} =
     \frac{\omega_{\sharp}^{1 / 2} \tgh^{1 / 2} \lp{2} \frac{\omega_{\sharp}
     t}{2} \rp{2}}{\omega_{\flat}^{1 / 2} \tgh^{1 / 2} \lp{2}
     \frac{\omega_{\flat} t}{2} \rp{2}} \leqslant
     \frac{\omega_{\sharp}}{\omega_{\flat}} . \]
  By (\ref{preparationclara10})
  \begin{eqnarray*}
    \ve{1} D_{\omega}^{- 1} \lp{2} \partial_{x_{\gamma}} \lp{1} \mathbbm{1} +
    R (t, x) \srp{1}{1 / 2}{} \rp{2} D_{\omega} x \ve{1} \leqslant &  & \nu
    \frac{\omega_{\sharp}}{\omega_{\flat}} \ve{1} \partial_{x_{\gamma}} R (t,
    x) \sve{1}{}{\infty} |x|\\
    \leqslant &  & \nu^{3 / 2} \frac{\omega_{\sharp}}{\omega_{\flat}}
    \underset{1 \leqslant \delta \leqslant \nu}{\max} |x_{\delta}
    \partial_{x_{\gamma}} R (t, x) |_{\infty} .
  \end{eqnarray*}
  By (\ref{preparationclara7})
  \[ x_{\delta} \partial_{x_{\gamma}} R (t, x) \assign - 2 t D_{\omega}^{- 1}
     \lp{2} \int_0^1 (1 - u) z_{\delta} \partial_{z_{\gamma}} (\partial_z
     \otimes \partial_z) \varphi (t, z, z) |_{z = u x} d u \rp{2}
     D_{\omega}^{- 1} \]
  Let $z \in D_{\mathbbm{R}^{\nu}, 1}$. Since
  \[ z_{\delta} \partial_{z_{\gamma}} \partial_{z_{\alpha}}
     \partial_{z_{\beta}} \psi = \partial_{z_{\gamma}} \partial_{z_{\alpha}}
     \partial_{z_{\beta}} z_{\delta} \psi - (\delta_{\alpha = \delta}
     \partial_{z_{\beta}} \partial_{z_{\gamma}} + \delta_{\beta = \delta}
     \partial_{z_{\alpha}} \partial_{z_{\gamma}} + \delta_{\gamma = \delta}
     \partial_{z_{\alpha}} \partial_{z_{\beta}}) \psi \]
  and by (\ref{preparationsandra1.2}), (\ref{preparationsandra1.4}),
  (\ref{preparationclara7.1}), one gets
  \begin{eqnarray*}
    |z_{\delta} \partial_{z_{\gamma}} \partial_{z_{\alpha}}
    \partial_{z_{\beta}} \varphi (t, z, z) | \leqslant &  & \frac{M_{\mu}'}{1
    - 4 T^2 M_{\mu}} \ch \lp{1} \frac{\omega_{\sharp} T}{2} \rp{1} + 3 \times
    \frac{4 M_{\mu}}{1 - 4 T^2 M_{\mu}}\\
    \leqslant &  & \frac{13 M_{\mu}'}{1 - 4 T^2 M_{\mu}} \ch \lp{1}
    \frac{\omega_{\sharp} T}{2} \rp{1} .
  \end{eqnarray*}
  Then

  \[ \ve{1} x_{\delta} \partial_{x_{\gamma}} R (t, x) \sve{1}{}{\infty}
     \leqslant \frac{26 M'_{\mu} (1 + \omega_{\flat} T)}{\omega_{\flat}^2 (1 -
     4 T^2 M_{\mu})} \ch \lp{1} \frac{\omega_{\sharp} T}{2} \rp{1} . \]
  Then
  \[ \ve{1} \partial_{x_{\gamma}} \Lambda (t, x) \ve{1} \leqslant 34 \nu^{3 /
     2} \ch \lp{1} \frac{\omega_{\sharp} T}{2} \rp{1}
     \frac{\omega_{\sharp}}{\omega_{\flat}} \frac{(1 + \omega_{\flat} T)
     M'_{\mu}}{\omega_{\flat}^2 (1 - 4 T^2 M_{\mu})} \]
  and
  \[ \ve{1} \partial_x \Lambda (t, x) \ve{1} \leqslant 34 \nu^{5 / 2} \ch
     \lp{1} \frac{\omega_{\sharp} T}{2} \rp{1} \frac{\omega_{\sharp} (1 +
     \omega_{\flat} T) M'_{\mu}}{\omega_{\flat}^3 (1 - 4 T^2 M_{\mu})} . \]
  Let us assume that
  \begin{equation}
    \label{preparationclara16} 34 \nu^{5 / 2} \alpha < \frac{1}{2} .
  \end{equation}
  Then, by (\ref{preparationclara5.31}), $\ve{1} \partial_x \Lambda (t, x)
  \ve{1} < \frac{1}{2}$. This proves (\ref{preparationclara4.7}).
\end{proof}

The following proposition gives a Borel summability property concerning
$\tmop{the} \tmop{function} w$ (see Proposition \ref{formulaclara34} for the
definition of this function).

\begin{proposition}
  \label{preparationclara20}Let $\omega_1, \ldots, \omega_{\nu}, T, R,
  \varepsilon > 0$. Let $\mu$ be a $\mathbbm{C}$-valued Borel measure defined
  on $\mathbbm{R}^{\nu}$ such that
  \begin{equation}
    \label{preparationclara21} 4 T^2 e^T M_{\mu, \varepsilon} < 1.
  \end{equation}
  Then there exist $\kappa, K, K_1 > 0$ and a function $\hat{W} \in
  \mathcal{C}^0 \lp{1}] 0, T [, \mathcal{A}(D^2_{\mathbbm{R}^{\nu}, 1 / 2}
  \times S_{\kappa}) \rp{1}$ satisfying, \ for every $(t, x, y) \in] 0, T
  [\times D^2_{\mathbbm{R}^{\nu}, 1 / 2}$,
  \begin{equation}
    \label{preparationclara21.2} \forall \sigma \in S_{\kappa} \text{, } |
    \text{$\hat{W}$} (t, x, y, \sigma) | \leqslant K_1 e^{K| \sigma |^{1 / 2}}
  \end{equation}
  and
  \begin{equation}
    \label{preparationclara21.4} \forall h \in \demiplan{} \text{, } \exp
    \lp{1} w (t, x, y, h) \rp{1} = \int_0^{+ \infty} e^{- \sigma / h} \hat{W}
    (t, x, y, \sigma) \frac{d \sigma}{h} .
  \end{equation}
\end{proposition}

\begin{proof}

  {\texspace{small}{\tmtextbf{-1a-}}} We first prove that there exist $\kappa,
  K_2 > 0$ and a function
  \[ \hat{w} \in \mathcal{C}^0 \lp{1}] 0, T [,
     \mathcal{A}(D^2_{\mathbbm{R}^{\nu}, 1 / 2} \times S_{2 \kappa}) \rp{1} \]
  such that, for $(t, x, y, \sigma) \in] 0, T [\times D_{\mathbbm{R}^{\nu}, 1
  / 2}^2 \times S_{2 \kappa}$,
  \begin{equation}
    \label{preparationclara22} | \text{$\hat{w}$} (t, x, y, \sigma) |
    \leqslant K_2
  \end{equation}
  and, for $h \in \demiplan{}$,
  \begin{equation}
    \label{preparationclara23} w (t, x, y, h) = \int_0^{+ \infty} e^{- \sigma
    / h} \widehat{w_{}} (t, x, y, \sigma) \frac{d \sigma}{h} .
  \end{equation}
  We proceed as in [Ha4]: for $B \geqslant 0$, the Borel transform of the
  function $h \longrightarrow e^{- B h}$ is $\sigma \longrightarrow J (B
  \sigma)$ where
  \[ J (z) \assign \sum_{n \geqslant 0} (- 1)^n \frac{z^n}{(n!)^2} =
     \int_0^{\pi} \cos \lp{1} 2 z^{1 / 2} \sin (\varphi) \rp{1} \frac{d
     \varphi}{\pi} . \]
  Therefore, let
  \[ \hat{w}_1 (t, x, y, \sigma) \assign - t \int_{0 < s_1 < 1}
     \int_{\mathbbm{R}^{\nu}} \int_0^1 \Omega_t . \xi \otimes_1 \xi J \lp{1}
     (\vartheta \Omega_t . \xi \otimes_1 \xi) \sigma \rp{1} e^{i q_t (s_1)
     \cdot \xi_1} d \vartheta d \mu (\xi_1) d s_1 \]
  and for $m \geqslant 2$
  \[ \hat{F}_m (t, x, y, \sigma) \assign 2^{m - 1} t^{2 m - 1} \Upsilon^{g,
     \xi}_t \Omega_t^{g, \theta} . \xi \otimes_m \xi J \lp{1} (\vartheta
     \Omega_t^{g, \theta} . \xi \otimes_m \xi) \sigma \rp{1} e^{i q_t^m (s)
     \cdot \xi}, \]
  \[ \hat{w}_m (t, x, y, \sigma) \assign (- 1)^m \sum_{g \in
     \mathcal{\mathcal{T}}_m} \int_{0 < s_1 < \ldots < s_m < 1}
     \int_{\mathbbm{R}^{\nu m}} \int_{[0, 1]^g \times [0, 1]} \]
  \[ \hat{F}_m (t, x, y, \sigma) d \vartheta d^{m - 1} \theta d^{\nu m}
     \mu^{\otimes} (\xi) d^m s \]
  (we shall check later that these integrals are convergent). Let us choose
  \begin{equation}
    \label{preparationclara23.5} \kappa = \frac{\varepsilon}{2} .
  \end{equation}

  {\texspace{med}{\tmtextbf{-1b-}}} We claim that, for $m \geqslant 1$,
  $\hat{w}_m \in \mathcal{C}^0 \lp{1}] 0, T [, \mathcal{A} \lp{0}
  D^2_{\mathbbm{R}^{\nu}, 1 / 2} \times S_{2 \kappa} \rp{0} \rp{1}$,
  \begin{equation}
    \label{preparationclara24} | \hat{w}_m (t, x, y, \sigma) | \leqslant m (4
    T^2 e^T M_{\mu, \varepsilon})^m,
  \end{equation}
  and for $h \in \demiplan{}$
  \begin{equation}
    \label{preparationclara25} w_m (t, x, y, h) = \int_0^{+ \infty} e^{-
    \sigma / h} \hat{w}_m (t, x, y, \sigma) \frac{d \sigma}{h} .
  \end{equation}
  We shall use the following estimate (see [Ha4]). For every $B \geqslant 0$
  \begin{equation}
    \label{preparationclara26} |J (B \sigma) | \leqslant \exp \lp{1} 2 B^{1 /
    2} | \mathcal{I} m \sigma^{1 / 2} | \rp{1} .
  \end{equation}
  In the sequel, we assume that $x, y \in D_{\mathbbm{R}^{\nu}, 1 / 2}$, $t
  \in] 0, T]$ and $\sigma \in S_{2 \kappa}$.
  
  \tmtextbf{{\texspace{med}{\tmtextbf{-1c-}}}} Let $m = 1$. Since $\sigma \in
  S_{2 \kappa}$ and by (\ref{preparationclara23.5}),
  \begin{equation}
    \label{preparationclara28} | \im \sigma^{1 / 2} | < \varepsilon^{1 / 2} .
  \end{equation}
  By (\ref{formulaclara12a}), $\Omega_t . \xi \otimes_1 \xi \geqslant 0$. By
  (\ref{formulaclara16.b}), (\ref{preparationclara26}) and
  (\ref{formulaclara12a}) again
  \begin{eqnarray*}
    \ve{1} \Omega_t . \xi \otimes_1 \xi J \lp{1} (\vartheta \Omega_t . \xi
    \otimes_1 \xi) \sigma \rp{1} e^{i q_t (s_1) \cdot \xi_1} \ve{1} \leqslant
    &  & \frac{t}{4} \xi_1^2 e^{\varepsilon^{1 / 2} t^{1 / 2} | \xi_1 |} e^{|
    \xi_1 |}\\
    \leqslant &  & \frac{T}{2} e^T e^{\varepsilon \xi_1^2 + 2| \xi_1 |}
  \end{eqnarray*}
  since $\frac{\xi_1^2}{2} \leqslant e^{| \xi_1 |}$ and $\varepsilon^{1 / 2}
  t^{1 / 2} | \xi_1 | \leqslant t + \varepsilon \xi_1^2$ . Then \ $\hat{w}_1
  (t, x, y, \sigma)$ is well defined, the regularity claim holds by the
  convergence dominated theorem and (\ref{preparationclara24}) is satisfied
  for $m = 1$. Since the Laplace transform of the function $\sigma \longmapsto
  J (B \sigma)$ is the function $h \longmapsto e^{- B h}$, one gets
  (\ref{preparationclara25}) for $m = 1$ by Fubini theorem.
  
  {\texspace{med}{\tmtextbf{-1d-}}} Let $m \geqslant 2$. By
  (\ref{formulaclara12b}), (\ref{preparationclara26}) and
  (\ref{preparationclara28})
  \begin{eqnarray*}
    \ve{1} J \lp{1} (\vartheta \Omega_t^{g, \theta} . \xi \otimes_m \xi)
    \sigma \rp{1} \ve{1} \leqslant &  & \exp \lp{1} \vartheta^{1 / 2} (m t)^{1
    / 2} \varepsilon^{1 / 2} \lp{1} \xi_1^2 + \cdots + \xi_m^2 \rp{1}^{1 / 2}
    \rp{1}\\
    \leqslant &  & e^{m t} e^{\varepsilon \lp{1} \xi_1^2 + \cdots + \xi_m^2
    \rp{1}} .
  \end{eqnarray*}
  Therefore, by (\ref{formulaclara16.b}), (\ref{formulaclara13}) and
  (\ref{formulaclara12b}),
  \[ | \hat{F}_m (t, x, y, \sigma) | \leqslant \frac{2^{m - 1} t^{2 m -
     1}}{4^{m - 1}}  \prod_{[ \bar{j}, \bar{k}] \in g} | \xi_{\bar{j}} | |
     \xi_{\bar{k}} | \times \frac{m t}{4} \lp{1} \xi_1^2 + \cdots + \xi_m^2
     \rp{1} \times \]
  \[ e^{m t} e^{\varepsilon \lp{1} \xi_1^2 + \cdots + \xi_m^2 \rp{1}} \times
     e^{| \xi_1 | + \cdots + | \xi_m |} . \]
  Then, by the dominated convergence theorem, our regularity claim holds.
  Moreover, by mimicking the proof of Lemma \ref{formulaclara18},
  \[ | \hat{w}_m (t, x, y, \sigma) | \leqslant \frac{2 t^{2 m - 1}}{m (m - 1)}
     \sum_{d_1 + \cdots + d_m = 2 m - 2} \int_{\mathbbm{R}^{\nu m}} \int_{[0,
     1]^g \times [0, 1]} \frac{| \xi_1 |^{d_1}}{d_1 !} \cdots \frac{| \xi_m
     |^{d_m}}{d_m !} \times \]
  \[ \frac{m t}{4} \lp{1} \xi_1^2 + \cdots + \xi_m^2 \rp{1} e^{m t}
     e^{\varepsilon \lp{1} \xi_1^2 + \cdots + \xi_m^2 \rp{1}} e^{| \xi |_1} d
     \vartheta d^{m - 1} \theta d^{\nu m} | \mu |^{\otimes} (\xi), \]
  which yields, by (\ref{formulaclara17.4}) and (\ref{formulaclara20.3}),
  \[ | \hat{w}_m (t, x, y, \sigma) | \leqslant \frac{2^{2 m - 3}}{2 (m - 1)}
     (3 m - 2) (2 m) t^{2 m} e^{m t} M^m_{\mu, \varepsilon} . \]
  This proves (\ref{preparationclara24}). Since the Laplace transform of the
  function $\sigma \longmapsto J (B \sigma)$ is the function $h \longmapsto
  e^{- B h}$, one gets (\ref{preparationclara25}).
  
  {\texspace{med}{\tmtextbf{-1e-}}} Let us assume that
  (\ref{preparationclara21}) holds. By (\ref{preparationclara24}), the
  function $\hat{w}$ defined by
  \[ \hat{w} (t, x, y, \sigma) = \sum_{m \geqslant 1} \hat{w}_m (t, x, y,
     \sigma) \]
  belongs to the space $\mathcal{C}^0 \lp{1}] 0, T [, \mathcal{A} \lp{0}
  D_{\mathbbm{R}^{\nu}, 1 / 2}^2 \times S_{2 \kappa} \rp{0} \rp{1}$ and
  \[ | \hat{w} (t, x, y, \sigma) | \leqslant \frac{4 T^2 e^T M_{\mu,
     \varepsilon}}{\lp{1} 1 - 4 T^2 e^T M_{\mu, \varepsilon} \srp{1}{2}{}} .
  \]
  Therefore (\ref{preparationclara22}) holds and, by
  (\ref{preparationclara25}), (\ref{preparationclara23}) is also satisfied.
  
  {\texspace{med}{\tmtextbf{-2-}}} Since there exists $\rho > 0$ such that
  $S_{\kappa} + D_{\rho} \subset S_{2 \kappa}$ and by the Cauchy formula, the
  function $\partial_{\sigma} \hat{w}$ is bounded on $] 0, T [\times
  D_{\mathbbm{R}^{\nu}, 1 / 2}^2 \times S_{\kappa}$. Then, by a
  parameter-dependent version of Proposition \ref{expconvmaria2} (see
  Appendix), there exists a function $\hat{W} \in \mathcal{C}^0 \lp{1}] 0, T
  [, \mathcal{A}(D^2_{\mathbbm{R}^{\nu}, 1 / 2} \times S_{\kappa}) \rp{1}$
  satisfying (\ref{preparationclara21.2}) and (\ref{preparationclara21.4}).
\end{proof}

\section{\label{nevflora}A Gaussian Borel summation statement}

\subsection{A multidimensional statement}

\begin{definition}
  Let $R, K > 0$. Let $\mathcal{N}_{R, K}^{\nu}$ be the set of analytic
  functions on $D_{\mathbbm{R}^{\nu}, \sqrt{R}} \times D_{\mathbbm{R}^+, R}$
  satisfying
  \[ a \in \mathcal{N}^{\nu}_{R, K} \Leftrightarrow \]
  \[ \forall (R', K') \in] 0, R [\times] K, + \infty [ \text{, } \exists C > 0
     \text{, } \forall (x, \sigma) \in D_{\mathbbm{R}^{\nu}, \sqrt{R'}} \times
     D_{\mathbbm{R}^+, R'}, \]
  \begin{equation}
    \label{nevflora0.41} |a (x, \sigma) | \leqslant C e^{K' (|x|^2 + | \sigma
    |)} .
  \end{equation}
\end{definition}

\begin{definition}
  \label{nevclara2}Let $R, K > 0$. Let $\hat{\mathcal{N}}_{R, K}$ be the set
  of analytic functions on $D_{\mathbbm{R}^+, R}$ satisfying
  \[ \hat{c} \in \hat{\mathcal{N}}_{R, K} \Leftrightarrow \forall (R', K')
     \in] 0, R [\times] K, + \infty [ \text{, } \exists C > 0 \text{, }
     \forall \tau \in D_{\mathbbm{R}^+, R'}, \]
  \begin{equation}
    \label{nevflora0.42} | \hat{c} (\tau) | \leqslant C e^{K' | \tau |} .
  \end{equation}
\end{definition}

Notice that the space $\hat{\mathcal{N}}_{R, K}$ is invariant under the
transformations $\hat{c} \mapsto \partial_{\tau} \hat{c}, \tau \hat{c}$.

\begin{proposition}
  \label{nevflora0.1}{\tmdummy}
  
  \begin{enumerate}
    \item Let $R, K > 0$. For every $a \in \mathcal{N}^{\nu}_{R, K}$, there
    exists $\hat{b} \in \hat{\mathcal{N}}_{R, K}$ such that, for $h \in
    \mathbbm{C}$, $\re \lp{1} \frac{1}{h} \rp{1} > K$,
    \begin{equation}
      \label{nevflora0.4} \int_{\mathbbm{R}^{\nu}} \int_0^{+ \infty} e^{- (x^2
      + \sigma) / h} a (x, \sigma) \frac{d^{\nu} x d \sigma}{h^{1 + \nu / 2}}
      = \int_0^{+ \infty} \hat{b} (\tau) e^{- \tau / h} \frac{d \tau}{h} .
    \end{equation}
    \item Let us assume that there exist $\kappa', C', K' > 0$ such that the
    function $a$ is analytic on $D_{\mathbbm{R}^{\nu}, \sqrt{\kappa'}} \times
    S_{\kappa'}$ and
    \begin{equation}
      \label{nevflora0.53} \forall (x, \sigma) \in D_{\mathbbm{R}^{\nu},
      \sqrt{\kappa'}} \times S_{\kappa'} \text{, } |a (x, \sigma) | \leqslant
      C' e^{K' | \sigma |^{1 / 2}} .
    \end{equation}
    Then the function $\hat{b}$ is analytic on $S_{\kappa'}$, for every $K'' >
    K'$ and $\kappa'' < \kappa'$, there exists $C'' > 0$ such that
    \begin{equation}
      \label{nevflora0.54} \forall \tau \in S_{\kappa''} \text{, } | \hat{b}
      (\tau) | \leqslant C'' e^{K'' | \tau |^{1 / 2}}
    \end{equation}
    and (\ref{nevflora0.4}) holds for $h \in \mathbbm{C}^+$.
  \end{enumerate}
\end{proposition}

\subsection{Hypergeometric vection transforms}

In this subsection, we introduce transforms exhibiting the analytic content of
Proposition \ref{nevflora0.1} and useful for its proof.

\begin{definition}
  Let $R, K > 0$. Let $\mathcal{H}\mathcal{V}_{1 / 2 \rightarrow 1}$ and
  $\mathcal{H}\mathcal{V}_{1 \rightarrow 1 / 2}$ be the operators defined on
  $\hat{\mathcal{N}}_{R, K}$ by
  \[ \mathcal{H}\mathcal{V}_{1 / 2 \rightarrow 1} ( \hat{c}) (\tau) \assign
     \frac{2}{\pi} \int_0^{\pi / 2} \hat{c} \lp{1} \tau \sin^2 \theta \rp{1} d
     \theta, \]
  \[ \mathcal{H}\mathcal{V}_{1 \rightarrow 1 / 2} ( \hat{c}) (\tau) \assign
     \int_0^{\pi / 2} \hat{c} (\tau \sin^2 \theta) \sin \theta d \theta + 2
     \tau \int_0^{\pi / 2} \hat{c}' (\tau \sin^2 \theta) \sin^3 \theta d
     \theta . \]
\end{definition}

\begin{proposition}
  \label{nevflora0.44}Let $R, K > 0$. For every $\hat{a} \in
  \hat{\mathcal{N}}_{R, K}$ (respectively $\hat{b} \in \hat{\mathcal{N}}_{R,
  K}$) there exists a unique $\hat{b} \in \hat{\mathcal{N}}_{R, K}$
  (respectively $\hat{a} \in \hat{\mathcal{N}}_{R, K}$) such that, for $h \in
  \mathbbm{C}$, $\re \lp{1} \frac{1}{h} \rp{1} > K$,
  \[ 2 \int_0^{+ \infty} \hat{a} (\tau^2) e^{- \tau^2 / h} \frac{d \tau}{(\pi
     h)^{1 / 2}} = \int_0^{+ \infty} \hat{b} (\tau) e^{- \tau / h} \frac{d
     \tau}{h} . \]
  Moreover $\hat{b} (\tau) =\mathcal{H}\mathcal{V}_{1 / 2 \rightarrow 1} (
  \hat{a}) (\tau)$ and $\hat{a} (\tau) =\mathcal{H}\mathcal{V}_{1 \rightarrow
  1 / 2} ( \hat{b}) (\tau)$.
\end{proposition}

\begin{proof}
  Since
  \[ 2 \int_0^{+ \infty} \hat{a} (\tau^2) e^{- \tau^2 / h} \frac{d \tau}{(\pi
     h)^{1 / 2}} = \frac{1}{\Gamma (1 / 2)} \int_0^{+ \infty} \hat{a} (\zeta)
     e^{- \zeta / h} \frac{d \zeta}{h^{1 / 2} \zeta^{1 / 2}}, \]
  Proposition \ref{nevflora0.44} is a consequence of the following proposition
  with the choice $\gamma = \frac{1}{2}$.
\end{proof}

\begin{proposition}
  \label{nevalclara1}Let $\gamma \in] 0, 1 [$. Let $R, K > 0$. For every
  $\hat{a} \in \hat{\mathcal{N}}_{R, K}$ (respectively $\hat{b} \in
  \hat{\mathcal{N}}_{R, K}$) there exists a unique $\hat{b} \in
  \hat{\mathcal{N}}_{R, K}$ (respectively $\hat{a} \in \hat{\mathcal{N}}_{R,
  K}$) such that, for $h \in \mathbbm{C}$, $\re \lp{1} \frac{1}{h} \rp{1} >
  K$,
  \begin{equation}
    \label{nevalclara4} \frac{1}{\Gamma (\gamma)} \int_0^{+ \infty} \hat{a}
    (\tau) e^{- \tau / h} \lp{2} \frac{\tau}{h} \srp{2}{\gamma - 1}{} \frac{d
    \tau}{h} = \int_0^{+ \infty} \hat{b} (\tau) e^{- \tau / h} \frac{d
    \tau}{h} .
  \end{equation}
  Moreover $\hat{b} (\tau) =\mathcal{H}\mathcal{V}_{\gamma \rightarrow 1} (
  \hat{a}) (\tau)$ and $\hat{a} (\tau) =\mathcal{H}\mathcal{V}_{1 \rightarrow
  \gamma} ( \hat{b}) (\tau)$ where
  \[ \mathcal{H}\mathcal{V}_{\gamma \rightarrow 1} ( \hat{a}) (\tau) \assign
     \frac{\sin \lp{1} \pi \gamma \rp{1}}{\pi} \int_0^1 \hat{a} (\tau u)
     u^{\gamma - 1} (1 - u)^{- \gamma} d u, \]
  \[ \mathcal{H}\mathcal{V}_{1 \rightarrow \gamma} ( \hat{b}) (\tau) \assign
     \gamma \int_0^1 \hat{b} (\tau u) (1 - u)^{\gamma - 1} d u + \tau \int_0^1
     \hat{b}' (\tau u) u (1 - u)^{\gamma - 1} d u \]
  (hypergeometric vection transforms).
\end{proposition}

\begin{proof}
  By the definition of the operators $\mathcal{H}\mathcal{V}_{\gamma
  \rightarrow 1}$ and $\mathcal{H}\mathcal{V}_{1 \rightarrow \gamma}$, the
  space $\hat{\mathcal{N}}_{R, K}$ is invariant under these operators. See
  also Section \ref{interpretationkarina} (Appendix). For two continuous
  functions $\widehat{f_1}, \hat{f}_2$ on $] 0, + \infty [$ such that
  $\tau^{\rho} \hat{f}_k (\tau) \overset{\tau \rightarrow 0}{\longrightarrow}
  0$ for some $\rho \in] 0, 1 [$, let $\ast$ be the convolution product
  defined by
  \begin{equation}
    \label{nevclara5.1} \widehat{f_1} \ast \hat{f}_2 (\tau) \assign
    \int_0^{\tau} \widehat{f_1} (\tau_1) \hat{f}_2 (\tau - \tau_1) d \tau_1 .
  \end{equation}
  Let $\hat{a} \in \hat{\mathcal{N}}_{R, K}$. One has
  \begin{eqnarray*}
    \int_0^{+ \infty} \hat{a} (\tau) e^{- \tau / h} \lp{2} \frac{\tau}{h}
    \srp{2}{\gamma - 1}{} d \tau = &  & h^{1 - \gamma} \times \int_0^{+
    \infty} \hat{a} (\tau) e^{- \tau / h} \tau^{\gamma - 1} d \tau\\
    = &  & \frac{1}{\Gamma (1 - \gamma)} \int_0^{+ \infty} \tau^{- \gamma}
    e^{- \tau / h} d \tau \times \int_0^{+ \infty} \tau^{\gamma - 1} \hat{a}
    (\tau) e^{- \tau / h} d \tau\\
    = &  & \frac{1}{\Gamma (1 - \gamma)} \int_0^{+ \infty} \hat{A} (\tau) e^{-
    \tau / h} d \tau
  \end{eqnarray*}
  where the function $\hat{A}$ denotes the convolution product of the
  functions $\tau \longmapsto \tau^{\gamma - 1} \hat{a} (\tau), \tau^{-
  \gamma}$. But
  \begin{eqnarray*}
    \hat{A} (\tau) = &  & \int_0^{\tau} \tau_1^{\gamma - 1} \hat{a} (\tau_1)
    (\tau - \tau_1)^{- \gamma} d \tau_1\\
    = &  & \int_0^1 \hat{a} (\tau u) u^{\gamma - 1} (1 - u)^{- \gamma} d u.
  \end{eqnarray*}
  This proves the existence of a function $\hat{b} \in \hat{\mathcal{N}}_{R,
  K}$ satisfying (\ref{nevalclara4}) and justifies the definition of
  $\mathcal{H}\mathcal{V}_{\gamma \rightarrow 1}$.
  
  Let $\hat{b} \in \hat{\mathcal{N}}_{R, K}$. A function $\hat{a}$ satisfies
  (\ref{nevalclara4}) if and only if
  \begin{equation}
    \label{nevalclara8} \int_0^{+ \infty} \hat{a} (\tau) e^{- \tau / h}
    \tau^{\gamma - 1} d \tau = h^{- 1} \times \Gamma (\gamma) h^{\gamma}
    \times \int_0^{+ \infty} \hat{b} (\tau) e^{- \tau / h} d \tau .
  \end{equation}
  Since
  \[ \Gamma (\gamma) h^{\gamma} = \int_0^{+ \infty} \tau^{\gamma - 1} e^{-
     \tau / h} d \tau, \]
  the right hand side of (\ref{nevalclara8}) is equal to
  \[ \frac{1}{h} \text{$\int_0^{+ \infty} \widehat{B} (\tau) e^{- \tau / h} d
     \tau$} = \text{$\int_0^{+ \infty} \partial_{\tau} \lp{1} \widehat{B}
     (\tau) \rp{1} e^{- \tau / h} d \tau$} \]
  where the function $\hat{B}$ denotes the convolution product of the
  functions $\tau \longmapsto \hat{b} (\tau), \tau^{\gamma - 1}$. But
  \begin{eqnarray*}
    \tau^{1 - \gamma} \partial_{\tau} \lp{1} \widehat{B} (\tau) \rp{1} = &  &
    \tau^{1 - \gamma} \partial_{\tau} \lp{2} \int_0^{\tau} \hat{b} (\tau_1)
    (\tau - \tau_1)^{\gamma - 1} d \tau_1 \rp{2}\\
    = &  & \tau^{1 - \gamma} \partial_{\tau} \lp{2} \tau^{\gamma} \int_0^1
    \hat{b} (\tau u) (1 - u)^{\gamma - 1} d u \rp{2}\\
    = &  & \gamma \int_0^1 \hat{b} (\tau u) (1 - u)^{\gamma - 1} d u + \tau
    \int_0^1 \hat{b}' (\tau u) u (1 - u)^{\gamma - 1} d u.
  \end{eqnarray*}
  This proves the existence of a function $\hat{a} \in \hat{\mathcal{N}}_{R,
  K}$ satisfying (\ref{nevalclara4}) and justifies the definition of
  $\mathcal{H}\mathcal{V}_{1 \rightarrow \gamma}$.
\end{proof}

\begin{remark}
  Proposition \ref{nevalclara1} is also related to fractional calculus. Let
  $\alpha > 0$. Let $f$ be a function continuous on $] 0, + \infty [$ and
  integrable on $] 0, 1 [$. For every $x \in] 0, + \infty [$, let
  \[ \ifr{0}{x}{\alpha} f \assign \frac{1}{\Gamma (\alpha)} \int_0^x (x -
     x_1)^{\alpha - 1} f (x_1) d x_1 \]
  (see also [M-R]). Then, if $\gamma \in] 0, 1 [$ and $\hat{c} \in
  \hat{\mathcal{N}}_{R, K}$,
  \[ \mathcal{H}\mathcal{V}_{\gamma \rightarrow 1} ( \hat{c}) (\tau) =
     \frac{1}{\Gamma (\gamma)} \times \ifr{0}{\tau}{1 - \gamma} \lp{1}
     \tau^{\gamma - 1} \hat{c} \rp{1}, \]
  \[ \mathcal{H}\mathcal{V}_{1 \rightarrow \gamma} ( \hat{c}) (\tau) = \Gamma
     (\gamma) \tau^{1 - \gamma} \partial_{\tau} ( \ifr{0}{\tau}{\gamma}
     \hat{c} \rp{1} . \]
  The fact that the transforms $\mathcal{H}\mathcal{V}_{\gamma \rightarrow 1}$
  and $\mathcal{H}\mathcal{V}_{1 \rightarrow \gamma}$ are inverse each to
  other is straightforward with the above formulas. One can also prove
  Proposition \ref{nevalclara1} by using a fractional integration by parts.
\end{remark}

\subsection{Proof of Proposition \ref{nevflora0.1}}

Let $R, K > 0$ and $a \in \mathcal{N}^{\nu}_{R, K}$. Let
\[ \mathcal{I}(h) \assign \int_{\mathbbm{R}^{\nu}} \int_0^{+ \infty} e^{- (x^2
   + \sigma) / h} a (x, \sigma) \frac{d^{\nu} x d \sigma}{h^{1 + \nu / 2}} .
\]
By the definition of $\mathcal{N}_{R, K}^{\nu}$, the above integral is well
defined if $\re \lp{1} \frac{1}{h} \rp{1} > K$. Moreover
\[ \mathcal{I}(h) = \int_{\mathbbm{R}^{\nu - 1}} \int_0^{+ \infty} e^{- (x'^2
   + \sigma) / h} A (x', \sigma) \frac{d^{\nu - 1} x' d \sigma}{h^{1 + (\nu -
   1) / 2}} \]
where
\begin{eqnarray*}
  A (x', \sigma) \assign &  & \int_{\mathbbm{R}} e^{- x_1^2 / h} a (x_1, x',
  \sigma) \frac{d x_1}{h^{1 / 2}}\\
  = &  & \frac{\pi^{1 / 2}}{2} \times 2 \int_0^{+ \infty} e^{- x_1^2 / h}
  \lp{1} a (x_1, x', \sigma) + a (- x_1, x', \sigma) \rp{1} \frac{d x_1}{(\pi
  h)^{1 / 2}} .
\end{eqnarray*}
By Proposition \ref{nevflora0.44}
\[ A (x', \sigma) = \int_0^{+ \infty} e^{- \tau_1 / h} B (\tau_1, x', \sigma)
   \frac{d \tau_1}{h} \]
where
\[ B (\tau_1, x', \sigma) \assign \frac{1}{\pi^{1 / 2}} \int_0^{\pi / 2}
   \lp{2} a \lp{1} \tau_1^{1 / 2} \sin \theta, x', \sigma \rp{1} + a \lp{1} -
   \tau_1^{1 / 2} \sin \theta, x', \sigma \rp{1} \rp{2} d \theta . \]
Then, by iterating this argument,
\[ \mathcal{I}(h) = \int_{] 0, + \infty [^{\nu + 1}} e^{- (\tau_1 + \cdots +
   \tau_{\nu} + \sigma) / h} d (\tau_1, \ldots, \tau_{\nu}, \sigma) \frac{d
   \tau_1 \cdots d \tau_{\nu} d \sigma}{h^{\nu + 1}} \]
where
\[ d (\tau_1, \ldots, \tau_{\nu}, \sigma) \assign \frac{1}{\pi^{\nu / 2}}
   \int_{[0, \pi / 2]^{\nu}} \sum_{\varepsilon_1 = \pm, \ldots,
   \varepsilon_{\nu} = \pm} \]
\[ a \lp{1} \varepsilon_1 \tau_1^{1 / 2} \sin \theta_1, \ldots,
   \varepsilon_{\nu} \tau_{\nu}^{1 / 2} \sin \theta_{\nu}, \sigma \rp{1} d
   \theta_1 \cdots d \theta_{\nu} . \]
Then
\[ \mathcal{I}(h) = \int_0^{+ \infty} \hat{b} (\tau) e^{- \tau / h} \frac{d
   \tau}{h} \]
where
\[ \hat{b} (\tau) \assign \partial_{\tau}^{\nu} \lp{3} \tau^{\nu} \int_{u_1 +
   \cdots + u_{\nu + 1} = 1} \frac{1}{\pi^{\nu / 2}} \int_{[0, \pi / 2]^{\nu}}
   \sum_{\varepsilon_1 = \pm, \ldots, \varepsilon_{\nu} = \pm} \]
\[ a \lp{1} \varepsilon_1 u_1^{1 / 2} \sin \theta_1 \tau^{1 / 2}, \ldots,
   \varepsilon_{\nu} u_{\nu}^{1 / 2} \sin \theta_{\nu} \tau^{1 / 2}, u_{\nu +
   1} \tau \rp{1} d \theta_1 \cdots d \theta_{\nu} d u_1 \cdots d u_{\nu}
   \rp{3} . \]
Let us remind that $z \in D_{\mathbbm{R}^+, R} \Rightarrow | \im z^{1 / 2} | <
R^{1 / 2} \Rightarrow \pm z^{1 / 2} \in D_{\mathbbm{R}, \sqrt{R}}$. Then,
since $u_1 + \cdots + u_{\nu} \leqslant 1$,
\[ \tau \in D_{\mathbbm{R}^+, R} \Rightarrow \lp{1} \varepsilon_1 u_1^{1 / 2}
   \sin \theta_1 \tau^{1 / 2}, \ldots, \varepsilon_{\nu} u_{\nu}^{1 / 2} \sin
   \theta_{\nu} \tau^{1 / 2}, u_{\nu + 1} \tau \rp{1} \in
   D_{\mathbbm{R}^{\nu}, \sqrt{R}} \times D_{\mathbbm{R}^+, R} . \]
Then the analyticity of the function $a$ on $D_{\mathbbm{R}^{\nu}, \sqrt{R}}
\times D_{\mathbbm{R}^+, R}$ implies the analyticity of the function $\hat{b}$
on $D_{\mathbbm{R}^+, R}$. The function $a$ satisfies (\ref{nevflora0.41}),
thus the function $\hat{b}$ satisfies (\ref{nevflora0.42}) (we use in
particular that the space $\hat{\mathcal{N}}_{R, K}$ is invariant under
$\hat{u} \mapsto \partial_{\tau} \hat{u}, \tau \hat{u}$). Let us assume that
the function $a$ satifies (\ref{nevflora0.53}). Then, since
\[ \tau \in S_{\kappa'} \Rightarrow \lp{1} \varepsilon_1 u_1^{1 / 2} \sin
   \theta_1 \tau^{1 / 2}, \ldots, \varepsilon_{\nu} u_{\nu}^{1 / 2} \sin
   \theta_{\nu} \tau^{1 / 2}, u_{\nu + 1} \tau \rp{1} \in
   D_{\mathbbm{R}^{\nu}, \sqrt{\kappa'}} \times S_{\kappa'}, \]
the function $\hat{b}$ is analytic on $S_{\kappa'}$. Moreover $\hat{b}$
satisfies (\ref{nevflora0.54}). This proves Proposition \ref{nevflora0.1}.

\section{Proofs of Proposition \ref{propositionbetty1} and Theorem
\ref{theorembetty2}}

\subsection{Proof of Proposition \ref{propositionbetty1}}

By (\ref{propositionbetty1.2}), one can use Proposition \ref{formulaclara34}.
Then, by (\ref{formulaclara36}),
\[ \lba{1} y|e^{- \frac{t}{h} H} |x \rba{1} = h^{- \nu / 2} e^{- \phi (t, x,
   y) / h} \times \lp{2} \prod_{\upsilon = 1}^{\nu} \frac{\omega_{\upsilon}^{1
   / 2}}{\sqrt{4 \pi} \tmop{sh}^{1 / 2} (\omega_{\upsilon} t)} \rp{2} e^{w (t,
   x, y, h)} . \]
By (\ref{propositionbetty1.2}), one can use Proposition
\ref{preparationclara20}. This proves Proposition \ref{propositionbetty1}.

\subsection{Proof of Theorem \ref{theorembetty2}}

Without loss of generality, one can assume $c (0) = 0$. Then
$\Theta^{\tmop{conj}}_H = \Theta_H$. Let $u$ be the solution of
(\ref{theorembetty6}). For $t, h > 0$, one has (see Remark
$\ref{formulaclara14.7})$
\[ \Theta_H (t, h) = \int_{\mathbbm{R}^{\nu}} u (t, x, x, h) d x. \]
By Proposition \ref{formulaclara34}, there exist $\phi \in \mathcal{C}^0
\lp{1}] 0, T [, \mathcal{A} \lp{0} D^2_{\mathbbm{R}^{\nu}, 1} \rp{0} \rp{1}$
and
\[ w \in \mathcal{C}^0 \lp{1}] 0, T [, \mathcal{A} \lp{0}
   D^2_{\mathbbm{R}^{\nu}, 1} \times \mathbbm{C}^+ \rp{0} \rp{1} \]
such that
\[ u (t, x, x, h) \assign (4 \pi h)^{- \nu / 2} \prod_{\upsilon = 1}^{\nu}
   \lp{2} \frac{\omega_{\upsilon}}{\tmop{sh} (\omega_{\upsilon} t)} \srp{2}{1
   / 2}{} \times \exp \lp{2} - \frac{1}{h} \phi |_{y = x} \rp{2} \times \exp
   \lp{1} w|_{y = x} \rp{1} . \]
By Proposition \ref{preparationclara4}, there exists $\Lambda \in
\mathcal{C}^0 \lp{1}] 0, T [, \mathcal{A}_{\mathbbm{C}^{\nu}}
(D_{\mathbbm{R}^{\nu}, 1}) \rp{1}$, such that $\phi \sve{1}{}{y = x} =
\theta^2 (t, x)$ where $\theta (t, x) \assign D_{\omega} \varphi (t, x)$,
$\varphi (t, x) \assign x + \Lambda (t, x)$ (see (\ref{preparationclara3.14})
for the definition of the matrix $D_{\omega}$). By Proposition
\ref{preparationclara20}, there exist $\kappa > 0$ and a function
$\widehat{\tmmathbf{W}} \in \mathcal{C}^0 \lp{1}] 0, T [,
\mathcal{A}(D_{\mathbbm{R}^{\nu}, 1 / 2} \times S_{\kappa}) \rp{1}$ (choose
$\widehat{\tmmathbf{W}} (t, x, \sigma) \assign \text{$\hat{W}$} (t, x, x,
\sigma)$) such that
\[ \exp \lp{1} w|_{y = x} \rp{1} = \int_0^{+ \infty} e^{- \sigma / h}
   \widehat{\tmmathbf{W}} (t, x, \sigma) \frac{d \sigma}{h} . \]
Therefore
\[ \Theta_H (t, h) = \prod_{\upsilon = 1}^{\nu} \lp{2}
   \frac{\omega_{\upsilon}}{4 \pi \tmop{sh} (\omega_{\upsilon} t)} \srp{2}{1 /
   2}{} \times \int_{\mathbbm{R}^{\nu}} \int_0^{+ \infty} e^{- \lp{1} \theta^2
   (t, x) + \sigma \rp{1} / h} \widehat{\tmmathbf{W}} (t, x, \sigma)
   \frac{d^{\nu} x d \sigma}{h^{1 + \nu / 2}} . \]
Since the assertion \ref{preparationclara4.5} of Proposition
\ref{preparationclara4} is satisfied, one can use a parameter-dependent
version of Proposition \ref{globaldiffeomaria2} (see Appendix). Therefore
$\varphi (t, \cdot) |_{\mathbbm{R}^{\nu}}$ and $\theta (t, \cdot)
|_{\mathbbm{R}^{\nu}}$ are global diffeomorphisms from $\mathbbm{R}^{\nu}$
onto $\mathbbm{R}^{\nu}$. Then
\[ \Theta_H (t, h) = \int_{\mathbbm{R}^{\nu}} \int_0^{+ \infty} e^{- (y^2 +
   \sigma) / h} a (t, y, \sigma) \frac{d^{\nu} y d \sigma}{h^{1 + \nu / 2}} .
\]
where
\[ a (t, y, \sigma) \assign \epsilon \prod_{\upsilon = 1}^{\nu} \lp{2}
   \frac{\omega_{\upsilon}}{4 \pi \tmop{sh} (\omega_{\upsilon} t)} \srp{2}{1 /
   2}{} \times \frac{\widehat{\tmmathbf{W}} \lp{2} t, \lp{1} \theta (t, \cdot)
   |_{\mathbbm{R}^{\nu}} \srp{1}{- 1}{} (y), \sigma \rp{2}}{\det \lp{2} \lp{1}
   \partial_x \theta (t, \cdot) \rp{1} \circ \lp{1} \theta (t, \cdot)
   |_{\mathbbm{R}^{\nu}} \srp{1}{- 1}{} (y) \rp{2}} \]
and $\epsilon$ is the sign of the above determinant. By Proposition
\ref{globaldiffeomaria2}, $(\varphi (t, \cdot) |_{\mathbbm{R}^{\nu}})^{- 1}$
maps $D_{\mathbbm{R}^{\nu}, 1 / 2}$ into $D_{\mathbbm{R}^{\nu}, 1}$. By
(\ref{preparationclara4.7})
\[ \sup_{(t, x) \in] 0, T [\times D_{\mathbbm{R}^{\nu}, 1}} \ve{1} \det \lp{1}
   \partial_x \varphi (t, x) \rp{1} \sve{1}{- 1}{} < \infty . \]
Moreover (\ref{preparationclara21.2}) holds. Let $T_0 \in] 0, T [$. Since
$\theta (t, x) = D_{\omega} \varphi (t, x)$, there exists $\kappa', C', K' >
0$ such that $a \in \mathcal{C}^0 \lp{1}] T_0, T [, \mathcal{A} \lp{0}
D_{\mathbbm{R}^{\nu}, \sqrt{\kappa'}} \times S_{\kappa'} \rp{0} \rp{1}$ and,
for every $t \in] T_0, T [$,
\[ \forall (y, \sigma) \in D_{\mathbbm{R}^{\nu}, \sqrt{\kappa'}} \times
   S_{\kappa'} \text{, } |a (t, y, \sigma) | \leqslant C' e^{K' | \sigma |^{1
   / 2}} \]
(notice that, since the matrix $D_{\omega}$ vanishes for $t = 0$, $\kappa'$
goes to $0$ and $C'$ goes to $\infty$ when $T_0$ goes to $0$). By a
parameter-dependent version of Proposition \ref{nevflora0.1}, since the
function $a (t, \cdot)$ satisfies (\ref{nevflora0.53}), one gets
(\ref{theorembetty5.05}) and (\ref{theorembetty5.1}). This proves Theorem
\ref{theorembetty2}.

\section{Appendix $A$}

Here is a statement about a non-negativity property which is useful when the
tree graph equality is considered. Such a result is well known (see [A-R,
Th.IV-5]).

\begin{lemma}
  \label{trimaria16}Let $I$ be a non-empty finite set. Let $M \assign (M_{j,
  k})_{j, k \in I}$ be a real symmetric non-negative{\footnote{Non-negativity
  means $\sum_{j, k \in I} M_{j, k} x_j x_k \geqslant 0$ for every $(x_j)_{j
  \in I} \in \mathbbm{R}^I$.}} matrix. Let $(u_{j, k})_{j, k \in I}$ be a
  symmetric matrix with coefficients in $[0, + \infty [$ such that, for every
  \ $j, k \in I, j \neq k$,
  \begin{equation}
    \label{trimaria17} \left\{ \begin{array}{l}
      \min_{\ell \in I} u_{\ell, \ell} \geqslant u_{j, k}\\
      \\
      \forall q \in I -\{j, k\} \text{, } u_{j, k} \geqslant \min (u_{j, q},
      u_{k, q})
    \end{array} . \right.
  \end{equation}
  Let $M^u$ be the matrix defined by
  \[ M^u_{j, k} = u_{j, k} M_{j, k} . \]
  Then the matrix $M^u$ is symmetric non-negative.
\end{lemma}

\begin{proof}
  We prove the lemma by induction on $|I|$. The statement is straightforward
  if $|I| \leqslant 2$. Let us assume the lemma is proved for $|I| < n$ and
  let $I$ be a subset such that $|I| = n$. Let $u_{\min} = \min_{j, k \in I}
  u_{j, k}$. Let $\mathcal{I}$ be the set containing every pair $(I_1, I_2)$
  such that $\varnothing \neq I_1 \subset I$, $\varnothing \neq I_2 \subset
  I$, $I_1 \cap I_2 = \varnothing$ and
  \begin{equation}
    \label{trimaria18} \forall i_1 \in I_1, \forall i_2 \in I_2 \text{, }
    u_{i_1, i_2} = u_{\min} .
  \end{equation}
  Let $(I_1, I_2) \in \mathcal{I}$ \ be such that $|I_1 | + |I_2 |$ is
  maximal. Let us assume that $I_1 \cup I_2 \neq I$. Let $j \in I - I_1 \cup
  I_2$. Then there exist $i_1 \in I_1$ and $i_2 \in I_2$ such that $u_{i_1, j}
  > u_{\min}$ and $u_{i_2, j} > u_{\min}$. Then, by (\ref{trimaria17})
  \[ u_{i_1, i_2} \geqslant \min (u_{i_1, j}, u_{i_2, j}) > u_{\min} \]
  which contradicts (\ref{trimaria18}). Thus $I_1 \sqcup I_2 = I$. Let $v_{j,
  k} \assign u_{j, k} - u_{\min}$. Then
  \[ M^u_{j, k} = u_{\min} M_{j, k} + 1_{j, k \in I_1} v_{j, k} M_{j, k} +
     1_{j, k \in I_2} v_{j, k} M_{j, k} . \]
  (\ref{trimaria17}) always holds if the coefficients $u_{j, k}$ are replaced
  by the coefficients $v_{j, k}$ and the set $I$ is replaced by some subset.
  Then the matrices $(v_{j, k} M_{j, k})_{j, k \in I_1}$ and $(v_{j, k} M_{j,
  k})_{j, k \in I_2}$ are symmetric non-negative. Then, by the above
  decomposition, the matrix $M^u$ is symmetric non-negative.
\end{proof}

For the reader's convenience we now verify the Borel summability of the
exponential of a Borel summable expansion in the setting of the paper. The
fact that Borel summability is preserved by composition with analytic
functions is well known in similar settings.

\begin{proposition}
  \label{expconvmaria2}Let $U \subset \mathbbm{C}$ be an open convex
  neighbourhood of $\mathbbm{R}^+$. Let $K_2 > 0$. Then there exist $K_1, K >
  0$ satisfying the following property. For every analytic function $\hat{a}$
  on $U$ such that
  \[ \max \lp{1} | \hat{a} (0) |, \underset{\sigma \in U}{\sup} | \hat{a}'
     (\sigma) | \rp{1} \leqslant K_2 \]
  there exists an analytic function $\hat{b}$ on $U$ satisfying
  \begin{enumerate}
    \item \label{expconvmaria4}for every $\sigma \in U$
    \[ | \hat{b} (\sigma) | \leqslant K_1 e^{K| \sigma |^{1 / 2}}, \]
    \item \label{expconvmaria6}for every $h \in \demiplan{}$
    \[ \exp \lp{2}  \int_0^{+ \infty} e^{- \sigma / h} \hat{a} (\sigma)
       \frac{d \sigma}{h} \rp{2} = \int_0^{+ \infty} e^{- \sigma / h} \hat{b}
       (\sigma) \frac{d \sigma}{h} . \]
  \end{enumerate}
\end{proposition}

\begin{proof}
  Let $h \in \demiplan{}$. By integration by parts,
  \[ \int_0^{+ \infty} e^{- \sigma / h} \hat{a} (\sigma) \frac{d \sigma}{h} =
     \hat{a} (0) + \int_0^{+ \infty} e^{- \sigma / h} \hat{a}' (\sigma) d
     \sigma . \]
  Then
  \begin{equation}
    \label{expconvmaria7} \exp \lp{2}  \int_0^{+ \infty} e^{- \sigma / h}
    \hat{a} (\sigma) \frac{d \sigma}{t} \rp{2} = e^{\hat{a} (0)} \lp{2} 1 +
    \sum_{n \geqslant 1} \frac{1}{n!} A_n (h) \rp{2}
  \end{equation}
  where
  \begin{eqnarray*}
    A_n (h) \assign &  & \lp{3} \int_0^{+ \infty} e^{- \sigma / h} \hat{a}'
    (\sigma) d \sigma \srp{3}{n}{}\\
    = &  & \int_0^{+ \infty} e^{- \sigma / h} ( \hat{a}')^{\ast, n} (\sigma) d
    \sigma
  \end{eqnarray*}
  where
  \begin{eqnarray*}
    ( \hat{a}')^{\ast, n} (\sigma) \assign &  & ( \hat{a}' \ast \cdots \ast
    \hat{a}') (\sigma) \text{ \ (} n \text{ times)}\\
    = &  & \int_{\sigma_1 + \cdots + \sigma_n = \sigma} \hat{a}' (\sigma_1)
    \cdots \hat{a}' (\sigma_n) d \sigma_2 \cdots d \sigma_n
  \end{eqnarray*}
  (see (\ref{nevclara5.1}) for the definition of $\ast$). Let
  \begin{eqnarray*}
    \hat{A}_n (\sigma) \assign &  & \int_0^{\sigma} ( \hat{a}')^{\ast, n}
    (\sigma') d \sigma'\\
    = &  & \sigma^n \int_{u_1 + \cdots + u_n \leqslant 1} \hat{a}' (u_1
    \sigma) \cdots \hat{a}' (u_n \sigma) d u_1 \cdots d u_n .
  \end{eqnarray*}
  Then, for $\sigma \in U$,
  \[ \ve{1} \hat{A}_n (\sigma) \ve{1} \leqslant \frac{K_2^n}{n!} | \sigma |^n
  \]
  and, by integration by parts,
  \begin{equation}
    \label{expconvmaria9} A_n (h) = \int_0^{+ \infty} e^{- \sigma / h}
    \hat{A}_n (\sigma) \frac{d \sigma}{h} .
  \end{equation}
  Let $\hat{b}$ be the analytic function on $U$ defined by
  \[ \hat{b} (\sigma) = e^{\hat{a} (0)} \lp{2} 1 + \sum_{n \geqslant 1}
     \frac{1}{n!} \hat{A}_n (\sigma) \rp{2} . \]
  Then
  \begin{eqnarray*}
    | \hat{b} (\sigma) | \leqslant &  & \ve{1} e^{\hat{a} (0)} \ve{1} \sum_{n
    \geqslant 0} \frac{1}{n!} \frac{1}{n!} K_2^n | \sigma |^n\\
    \leqslant &  & e^{K_2} \sum_{n \geqslant 0} \frac{1}{(2 n) !} \lp{2} 2
    K_2^{1 / 2} | \sigma |^{1 / 2} \srp{2}{2 n}{}\\
    \leqslant &  & e^{K_2} \exp \lp{1} 2 K_2^{1 / 2} | \sigma |^{1 / 2} \rp{1}
    .
  \end{eqnarray*}
  Hence the function $\hat{b}$ satisfies the assertion \ref{expconvmaria4}. By
  (\ref{expconvmaria7}) and (\ref{expconvmaria9}),
  \begin{eqnarray*}
    \exp \lp{2}  \int_0^{+ \infty} e^{- \sigma / h} \hat{a} (\sigma) \frac{d
    \sigma}{h} \rp{2} = &  & e^{\hat{a} (0)} \int_0^{+ \infty} e^{- \sigma /
    h} \lp{2} 1 + \sum_{n \geqslant 1} \frac{1}{n!} \hat{A}_n (\sigma) \rp{2}
    \frac{d \sigma}{h}\\
    = &  & \int_0^{+ \infty} e^{- \sigma / h} \hat{b} (\sigma) \frac{d
    \sigma}{h} .
  \end{eqnarray*}
  and the assertion \ref{expconvmaria6} is also satisfied.
\end{proof}

We also use in the paper the following result.

\begin{proposition}
  \label{globaldiffeomaria2}Let $r > 0$. Let $\Lambda : D_{\mathbbm{R}^{\nu},
  r} \longrightarrow \mathbbm{C}^{\nu}$ be an analytic function such that
  \begin{enumerate}
    \item the function $\Lambda |_{\mathbbm{R}^{\nu}}$ is
    $\mathbbm{R}^{\nu}$-valued,
    
    \item the matrix $\partial_x \Lambda (x) \assign \lp{1} \partial_{x_1}
    \Lambda (x), \ldots, \partial_{x_{\nu}} \Lambda (x) \rp{1}$ satisfies
    \begin{equation}
      \label{continuationelisabeth1} M \assign \underset{x \in
      D_{\mathbbm{R}^{\nu}, r}}{\sup} | \partial_x \Lambda (x) | < 1.
    \end{equation}
    Let $\varphi \in \mathcal{A}_{\mathbbm{C}^{\nu}} \lp{1}
    D_{\mathbbm{R}^{\nu}, r} \rp{1}$ be defined by $\varphi (x) = x + \Lambda
    (x)$. Then $\varphi |_{\mathbbm{R}^{\nu}}$ is a global diffeomorphism from
    $\mathbbm{R}^{\nu}$ onto $\mathbbm{R}^{\nu}$. Moreover $(\varphi
    |_{\mathbbm{R}^{\nu}})^{- 1}$ admits an analytic continuation{\footnote{we
    choose the same notation for this continuation.}} on
    $D_{\mathbbm{R}^{\nu}, r (1 - M)}$ and
    \[ (\varphi |_{\mathbbm{R}^{\nu}})^{- 1} \lp{1} D_{\mathbbm{R}^{\nu}, r (1
       - M)} \rp{1} \subset D_{\mathbbm{R}^{\nu}, r} . \]
  \end{enumerate}
\end{proposition}

\begin{proof}
  By (\ref{continuationelisabeth1}), $\varphi$ is a local diffeomorphism. Let
  us prove that $\varphi$ is injective. Let $x_1, x_2 \in
  D_{\mathbbm{R}^{\nu}, r}$. Then
  \begin{equation}
    \label{continuationelisabeth2} \varphi (x_2) - \varphi (x_1) = \lp{2}
    \tmop{Id} + A (x_1, x_2) \rp{2} (x_2 - x_1)
  \end{equation}
  where
  \[ A (x_1, x_2) \assign \int_0^1 (\partial_z \Lambda) (z) |_{z = x_1 + u
     (x_2 - x_1)} d u. \]
  But the matrix $\tmop{Id} + A (x_1, x_2)$ is invertible since
  \[ |A (x_1, x_2) | \leqslant \int_0^1 | \partial_z \Lambda (z) ||_{z = x_1 +
     u (x_2 - x_1)} d u \leqslant M < 1. \]
  Then, by (\ref{continuationelisabeth2}), $\varphi$ is injective.
  
  Let $x_0 \in \mathbbm{R}^{\nu}$ and let $y_0 = \varphi (x_0)$. Let $R > 0$
  and let $(x, y) \in B (x_0, r) \times B (y_0, R)$. One has
  \begin{equation}
    \label{continuationelisabeth4} y = \varphi (x) \Leftrightarrow y - y_0 =
    \varphi (x) - \varphi (x_0) \Leftrightarrow (\tmop{Id} + A (x_0, x))^{- 1}
    (y - y_0) = (x - x_0) .
  \end{equation}
  Let us assume that
  \[ \frac{1}{1 - M} \times R < r. \]
  Let $\rho \in] \frac{R}{1 - M}, r [$. Then $x \longmapsto x_0 + (\tmop{Id} +
  A (x_0, x))^{- 1} (y - y_0)$ maps $\overline{B (x_0, \rho)}$ into
  $\overline{B (x_0, \rho)}$ and has a fixed point by the Brouwer theorem
  since $\overline{B (x_0, \rho)}$ is compact and convex. Therefore, for every
  $y \in \text{$B (y_0, R)$}$, there exists $x \in \text{$\overline{B (x_0,
  \rho)}$}$ satisfying (\ref{continuationelisabeth4}). Then
  \[ \forall x_0 \in \mathbbm{R}^{\nu}, B \lp{1} \varphi (x_0), r (1 - M)
     \rp{1} \subset \varphi \lp{1} B (x_0, r) \text{$\rp{1}$} . \]
  The above inclusion holds if $B (z, \varepsilon)$ denotes the real or
  complex ball of center $z \in \mathbbm{R}^{\nu}$ and radius $\varepsilon >
  0$. In particular, $\varphi |_{\mathbbm{R}^{\nu}}$ is an open and closed
  map. Therefore $\varphi |_{\mathbbm{R}^{\nu}} (\mathbbm{R}^{\nu})
  =\mathbbm{R}^{\nu}$. Since $\varphi |_{\mathbbm{R}^{\nu}}$ is injective,
  $\varphi |_{\mathbbm{R}^{\nu}}$ is a global diffeomorphism from
  $\mathbbm{R}^{\nu}$ onto $\mathbbm{R}^{\nu}$. Moreover
  \[ \bigcup_{x_0 \in \mathbbm{R}^{\nu}} B \lp{1} \varphi (x_0), r (1 - M)
     \rp{1} \subset \bigcup_{x_0 \in \mathbbm{R}^{\nu}} \varphi \lp{1} B (x_0,
     r) \text{$\rp{1}$} . \]
  Then
  \[ D_{\mathbbm{R}^{\nu}, r (1 - M)} \subset \varphi \lp{1}
     D_{\mathbbm{R}^{\nu}, r} \rp{1} . \]
  Therefore $(\varphi |_{\mathbbm{R}^{\nu}})^{- 1}$ admits an analytic
  continuation on $D_{\mathbbm{R}^{\nu}, r (1 - M)}$ and
  \[ (\varphi |_{\mathbbm{R}^{\nu}})^{- 1} \lp{1} D_{\mathbbm{R}^{\nu}, r (1 -
     M)} \rp{1} \subset D_{\mathbbm{R}^{\nu}, r} . \]
\end{proof}

\section{Appendix B}

In this section we propose an interpretation of the tools used in Section
\ref{nevflora} with respect to standard Borel summation concepts.

\subsection{Gaussian integrals}

For $\theta \in] 0, \pi [$ and $R > 0$, let
\[ \mathcal{C}_{\prec, \theta} \assign \lbc{1} r e^{i \varphi} \in
   \mathbbm{C}|r > 0, \varphi \in] - \theta, \theta [ \rbc{1} \]
\[ \mathcal{C}_{R, \prec, \theta} \assign \lbc{1} r e^{i \varphi} \in
   \mathbbm{C}|r \in] 0, R [, \varphi \in] - \theta, \theta [ \rbc{1}, \]
\begin{definition}
  Let $\alpha \in] 0, 2 [$. We say that
  \begin{itemizedot}
    \item A function $f$ satisfies $\mathcal{P}_{\tmop{wat}, \alpha}$ if and
    only there exists $\varepsilon, \kappa, \rho, K > 0$ such that
    \begin{itemizeminus}
      \item $f$ is analytic on $\mathcal{C}_{\rho^{\alpha}, \prec, \alpha (
      \frac{\pi}{2} + \varepsilon)}$,
      
      \item there exist $a_0, a_1, \ldots \in \mathbbm{C}$, $R_0, R_1, \ldots$
      analytic functions on $\mathcal{C}_{\rho^{\alpha}, \prec, \alpha (
      \frac{\pi}{2} + \varepsilon)}$ such that, for every $r \geqslant 0$, for
      every $x \in \mathcal{C}_{\rho^{\alpha}, \prec, \alpha ( \frac{\pi}{2} +
      \varepsilon)}$,
      \begin{equation}
        \label{appendixanna2} \text{$f (x) = a_0 + \cdots + a_{r - 1} x^{r -
        1} + R_r (x)$},
      \end{equation}
      \item for every $r \geqslant 0$ and $x \in \mathcal{C}_{\rho^{\alpha},
      \prec, \alpha ( \frac{\pi}{2} + \varepsilon)}$,
      \begin{equation}
        \label{appendixanna4} |R_r (x) | \leqslant K \frac{\Gamma (1 + \alpha
        r)}{\kappa^{\alpha r}} |x|^r .
      \end{equation}
    \end{itemizeminus}
    \item A function $\text{$\hat{f}$}$ satisfies
    $\hat{\mathcal{P}}_{\tmop{wat}, \alpha}$ if and only if there exists
    $\varepsilon, \kappa, \rho, K > 0$ such that
    \begin{itemizeminus}
      \item $\hat{f}$ is analytic on $D_{\kappa^{\alpha}} \cup
      \mathcal{C}_{\prec, \alpha \varepsilon}$,
      
      \item for every $\xi \in D_{\kappa^{\alpha}} \cup \mathcal{C}_{\prec,
      \alpha \varepsilon}$
      \begin{equation}
        \label{barbara2} | \hat{f} (\xi) | \leqslant Ke^{^{} \frac{| \xi |^{1
        / \alpha}}{\rho}} .
      \end{equation}
    \end{itemizeminus}
  \end{itemizedot}
\end{definition}

One has (Watson's lemma)

\begin{theorem}
  \label{theorembarbara0} Let $\alpha \in] 0, 2 [$.
  \begin{itemizedot}
    \item If $f$ verifies $\mathcal{P}_{\tmop{wat}, \alpha}$, then
    \begin{equation}
      \label{barbara4} \hat{f} (\xi) \assign \sum_{r = 0}^{\infty} \frac{a_r
      \xi^r}{\Gamma (1 + \alpha r)}
    \end{equation}
    admits an analytic continuation which verifies
    $\text{$\hat{\mathcal{P}}_{\tmop{wat}, \alpha}$}$.
    
    \item If $\hat{f}$ verifies $\hat{\mathcal{P}}_{\tmop{wat}, \alpha}$, then
    \begin{equation}
      \label{barbara6} \text{$f (x) \assign \int_0^{+ \infty} \hat{f} (\xi)
      e^{- \frac{\xi^{1 / \alpha}}{x^{1 / \alpha}}} d_{\xi} \lp{2}
      \frac{\xi^{1 / \alpha}}{x^{1 / \alpha}} \rp{2}$} = \int_0^{+ \infty}
      \hat{f} (x \zeta^{\alpha}) e^{- \zeta} d \zeta
    \end{equation}
    verifies $\mathcal{P}_{\tmop{wat}, \alpha}$.
    
    \item $\hat{f}$ given by (\ref{barbara4}) is called the $\alpha$-Borel
    transform of $f$. $f$ given by (\ref{barbara6}) is called the
    $\alpha$-Laplace transform of $\hat{f}$. These two transforms are inverse
    each to other.
  \end{itemizedot}
\end{theorem}

We say that a function $f$ satisfying $\mathcal{P}_{\tmop{wat}, \alpha}$ is
$\frac{1}{\alpha}$-$\mathbbm{R}^+$-summable. The notion of
$\frac{1}{\alpha}$-summability is related to the notion of critical time and
celeration's theory [Bals, E3, Ma-Ra]. Notice that two different indices
$\alpha, \alpha'$ may yield two different notions of Borel-summability [Lo].

Let $\varepsilon, \kappa, \rho, K > 0$. Let $f$ be an analytic function on
\[ \mathcal{U} \assign \mathcal{C}_{\rho^{1 / 2}, \prec, \frac{1}{2} (
   \frac{\pi}{2} + \varepsilon)} \cup e^{i \pi} \mathcal{C}_{\rho^{1 / 2},
   \prec, \frac{1}{2} ( \frac{\pi}{2} + \varepsilon)} \]
satisfying (\ref{appendixanna2}) and (\ref{appendixanna4}) for $\alpha =
\frac{1}{2}$ and $x \in \mathcal{U}$. For $s \in \mathcal{C}_{\rho, \prec,
\frac{\pi}{2} + \varepsilon}$, let us define
\[ f_{\tmop{ev}} (s) = \frac{1}{2} \lp{1} f (s^{1 / 2}) + f (- s^{1 / 2})
   \rp{1}, \]
\[ f_{\tmop{odd}} (s) = \frac{s^{1 / 2}}{2} \lp{1} f (s^{1 / 2}) - f (- s^{1 /
   2}) \rp{1} . \]
Then the functions $f_{\tmop{ev}}$ and $f_{\tmop{odd}}$ are
$1$-$\mathbbm{R}^+$-summable and one can easily describe the
$\frac{1}{2}$-$\mathbbm{R}^+$-summable function $f$ with the help of the
$1$-$\mathbbm{R}^+$-summable functions $f_{\tmop{ev}}$ and $f_{\tmop{odd}}$
since
\[ f (x) = f_{\tmop{ev}} (x^2) + \frac{1}{x} f_{\tmop{odd}} (x^2) . \]
Of course, this process does not hold for an arbitrary
$\frac{1}{2}$-$\mathbbm{R}^+$-summable function. Let us now consider what
happens in the Borel plane. By Theorem \ref{theorembarbara0}, there exists a
function $\hat{f}$ satisfying $\hat{\mathcal{P}}_{\tmop{wat}, 1 / 2}$ such
that
\[ f (x) = \int_0^{+ \infty} \hat{f} (\xi) e^{- \frac{\xi^2}{x^2}} d_{\xi}
   \lp{2} \frac{\xi^2}{x^2} \rp{2} = \int_0^{+ \infty} \hat{f} (x \xi) e^{-
   \xi^2} d_{\xi} (\xi^2) . \]
By the properties of the function $f$, there exist $\varepsilon, \kappa, \rho,
K > 0$ such that the function $\hat{f}$ is analytic on $D_{\kappa^{1 / 2}}
\cup \mathcal{C}_{\prec, \frac{1}{2} \varepsilon} \cup e^{i \pi}
\mathcal{C}_{\prec, \frac{1}{2} \varepsilon}$ and satisfies (\ref{barbara2})
on $D_{\kappa^{1 / 2}} \cup \mathcal{C}_{\prec, \frac{1}{2} \varepsilon} \cup
e^{i \pi} \mathcal{C}_{\prec, \frac{1}{2} \varepsilon}$. Let
$\hat{f}_{\tmop{ev}}$ and $\hat{f}_{\tmop{odd}}$ be the analytic functions
defined on $D_{\kappa} \cup \mathcal{C}_{\prec, \varepsilon}$ by
\[ \hat{f}_{\tmop{ev}} (\zeta) \assign \frac{1}{2} \lp{1} \hat{f} (\zeta^{1 /
   2}) + \hat{f} (- \zeta^{1 / 2}) \rp{1}, \]
\[ \hat{f}_{\tmop{odd}} (\zeta) \assign \zeta^{1 / 2} \lp{1} \hat{f} (\zeta^{1
   / 2}) - \hat{f} (- \zeta^{1 / 2}) \rp{1} . \]
Then
\[ \hat{f} (\xi) = \hat{f}_{\tmop{ev}} (\xi^2) + \frac{1}{2 \xi}
   \hat{f}_{\tmop{odd}} (\xi^2) \]
and
\[ f_{\tmop{ev}} (s) = \int_0^{+ \infty} \hat{f}_{\tmop{ev}} (\zeta) e^{-
   \frac{\zeta}{s}} \frac{d \zeta}{s}, \]
\begin{equation}
  \label{appendixanna8} f_{\tmop{odd}} (s) = \int_0^{+ \infty}
  \hat{f}_{\tmop{odd}} (\xi^2) e^{- \frac{\xi^2}{s}} \frac{d \xi}{s^{1 / 2}} .
\end{equation}
In the last expression, we recognize the integral occuring in Proposition
\ref{nevflora0.44}. Then Proposition \ref{nevflora0.44} allows one to express
the integral in (\ref{appendixanna8}) as a standard Laplace transform: the
function $f_{\tmop{odd}}$ is $1$-$\mathbbm{R}^+$-Borel summable.

\subsection{\label{interpretationkarina}A remark about hypergeometric vection
transforms}

In this section, we do not give rigorous statements for the sake of
conciseness. Here we consider the following choice for the Laplace and the
Borel transform. If $\hat{f}$ is a function of a complex variable $\zeta$, we
denote
\[ \text{$\mathcal{L} \hat{f} (x) = \int_0^{+ \infty} \hat{f} (\zeta) e^{-
   \frac{\zeta}{x}} d \zeta$} . \]
Let $\sigma \in] 0, + \infty [$. Then the Borel transform $\mathcal{B}$, which
is the inverse transform of $\mathcal{L}$, satisfies
\[ \mathcal{B}(x^{\sigma}) (\zeta) = \frac{\zeta^{\sigma - 1}}{\Gamma
   (\sigma)} . \]
With this choice, the pointwise product is mapped on the natural convolution
product defined by (\ref{nevclara5.1}). Let $\gamma \in] 0, 1 [$. Let
\[ \hat{b} (\zeta) \assign \sum_{n \geqslant 0} \hat{b}_n \zeta^n . \]
Then
\[ b (x) \assign \mathcal{L}( \hat{b}) (\zeta) = \sum_{n \geqslant 0} \Gamma
   (n + 1) \hat{b}_n x^{n + 1} \]
and
\begin{equation}
  \label{interpretationkarina2} b (x) = x^{1 - \gamma} \times \sum_{n
  \geqslant 0} \Gamma (n + 1) \hat{b}_n x^{n + \gamma} .
\end{equation}
Then
\[ \hat{b} (\zeta) =\mathcal{B}(x^{1 - \gamma}) (\zeta) \ast \lp{2} \sum_{n
   \geqslant 0} \frac{\Gamma (n + 1) \hat{b}_n}{\Gamma (n + \gamma)} \zeta^{n
   + \gamma - 1} \rp{2} . \]
This induces a natural factorization:
\[ \hat{b} (\zeta) =\mathcal{B}(x^{1 - \gamma}) (\zeta) \ast \lp{2}
   \mathcal{B}(x^{\gamma}) (\zeta) \times \hat{a} (\zeta) \rp{2} \]
where
\[ \hat{a} (\zeta) \assign \sum_{n \geqslant 0} \frac{\Gamma (\gamma) \Gamma
   (n + 1)}{\Gamma (n + \gamma)} \hat{b}_n \zeta^n . \]
Notice that
\[ \mathcal{H}\mathcal{V}_{\gamma \rightarrow 1} ( \hat{a}) (\zeta)
   =\mathcal{B}(x^{1 - \gamma}) (\zeta) \ast \lp{2} \mathcal{B}(x^{\gamma})
   (\zeta) \times \hat{a} (\zeta) \rp{2} . \]
Then the content of Proposition \ref{nevalclara1} is related to the
factorization (\ref{interpretationkarina2}). Notice also that, for every $n
\in \mathbbm{N}$,
\begin{eqnarray*}
  \mathcal{H}\mathcal{V}_{\gamma \rightarrow 1} (\zeta^n) = &  & \frac{\Gamma
  (n + \gamma)}{\Gamma (\gamma) \Gamma (n + 1)} \zeta^n\\
  = &  & \frac{(\gamma)_n}{n!} \zeta^n
\end{eqnarray*}
where
\[ (\gamma)_n \assign \gamma (\gamma + 1) \cdots (\gamma + n - 1) . \]
Let $a, b, c \in] 0, 1 [$. Then hypergeometric vection transforms give, for
instance, the following decomposition of the standard hypergeometric function
$\text{}_2 F_1$
\begin{eqnarray*}
  \text{}_2 F_1 (a, b ; c ; z) \assign &  & \sum_{n \geqslant 0} \frac{(a)_n
  (b)_n}{(c)_n} \frac{z^n}{n!}\\
  = &  & \mathcal{H}\mathcal{V}_{a \rightarrow 1} \circ
  \mathcal{H}\mathcal{V}_{b \rightarrow 1} \circ \mathcal{H}\mathcal{V}_{1
  \rightarrow c} \lp{2} \frac{1}{1 - z} \rp{2} .
\end{eqnarray*}
{\texspace{med}{}}

\ \ \ \ \ \ \ \ \ \ \ \ \ \ \ \ \ \ \ \ \ \ \ \ \ \ \ \ \ \ \ \ \ \ \ \ \ \ \
\ \ \ REFERENCES

{\texspace{small}{}}

[A] A. Abdesselam, The ground state energy of the massless spin-boson model,
Ann. Henri Poincar\'e 12 (2011), 1321-1347.

[A-R] A. Abdesselam and V. Rivasseau, ``Trees, forests and jungles: a
botanical garden for cluster expansions'', Lectures Notes in Physics, vol 446
(1995), 7-36.

[Bals] W. Balser, From divergent power series to analytic functions,
Springer-Verlag.

[B-B] R. Balian and C. Bloch, \tmtextit{\tmtextup{Solutions of the
Schr\"odinger equation in terms of classical paths}}, Ann. of Phys.
\tmtextbf{85} (1974), 514-545.

[Br] D. C. Brydges Functional integrals and their applications (Notes for a
course for the ``troisi\`eme Cycle de la Physique en Suisse Romande'' given in
Lausanne, Switzerland, during the summer of 1992, Notes by R. Fernandez.

[B-K] D.C. Brydges, T. Kennedy, Mayer expansion and the Hamilton-Jacobi
Equation, Journal of Statistical Physics, Vol 48., Nos. 1/2, (1987).

[Co] L. Comtet, Advanced Combinatorics: The Art of Finite and Infinite
Expansions.

[D-D-P] E. Delabaere, H. Dillinger, F. Pham,
\tmtextit{\tmtextup{\tmtextmd{Exact semiclassical expansions for
one-dimensional quantum oscillator}}}, J. Math. Phys. \tmtextbf{38(12)}
(1997), 6126-6184.

[E1] J. Ecalle, Les fonctions r\'esurgentes (en trois parties), tome 1,
chapitre 4, pr\'epublications Orsay (1981)

[E2] J. Ecalle, Singularit\'es non abordables par la g\'eom\'etrie, Ann. Inst.
Fourier, 42 (1992), 73-164.

[E3] J. Ecalle, Six lectures on transseries, analysable functions and the
constructive proof of Dulac's conjecture, Bifurcations and periodic orbits of
vector fields, Nato ASI Series, Vol. 408, 1993, 75-184.

[E4] J. Ecalle, Celerations, scensions, vections (forthcomming).

[E-V] J. Ecalle, B. Vallet, The arborification-coarborification transform:
analytic, combinatorial, and algebraic aspects, Annales de la Facult\'e des
Sciences de Toulouse, Ser. 6, 13, no. 4, (2004).

[Fe] R.P. Feynman, The Theory of Positrons, Physical Review, vol 76 (6)
(1949), 749-759.

[Fu-Os-Wi] Y. Fujiwara, T.A. Osborn, S.F.J. Wilk; Wigner-Kirkwood expansions,
Physical Review A, 25-1 (1982), 14-34.

[Ha3] T. Harg\'e, Noyau de la chaleur en dimension quelconque avec potentiel
en temps petit, prepublication Orsay (1998).

[Ha4] T. Harg\'e, Borel summation of the small time expansion of the heat
kernel. The scalar potential case (2013).

[Ha6] T. Harg\'e, A deformation formula for the heat kernel (2013).

[Ha7] T. Harg\'e, Some remarks on the complex heat kernel on
$\mathbbm{C}^{\nu}$ in the scalar potential case (2013).

[Lo] M. Loday-Richaud, Introduction \`a la multisommabilit\'e, S.M.F., Gazette
des Math\'ematiciens 44 (1990).

[It] K. Ito, \tmtextup{Generalized uniform complex measures in the Hilbertian
metric space with their applications to the Feynman integral}, Fifth berkeley
Symp. on Math. Statist. and Prob. \tmtextbf{2} (1967), 145--161.

[Ma-Ra] J. Martinet, J.P. Ramis, Elementary acceleration and multisummability
I, Ann. Inst. Henri Poincarre, vol. 54, n$^{\circ}$4, 1991, 331-401.

[M-R] K.S Miller, B. Ross, An introduction to the fractional calculus and
fractional differential equations, Wiley-Interscience.

[On] E. Onofri, On the high-temperature expansion of the density matrix,
American Journal Physics, 46-4 (1978), 379-382.

[Si] B. Simon, Fifty years of eigenvalue perturbation theory, Bulletin (New
Series) of the American Mathematical Society, vol. 24, no. 2, 1991.

[V1] A. Voros, The return of the quartic oscillator.\tmtextit{ \tmtextup{The
complex WKB method}}, Annales de l'institut Henri Poincar\'e (A) Physique
th\'eorique \tmtextbf{39-3} (1983), 211-338.

{\texspace{med}{}}

D\'epartement de Math\'ematiques, Laboratoire AGM (CNRS), Universit\'e de
Cergy-Pointoise, 95000 Cergy-Pontoise, France.

\end{document}